\documentclass[acmsmall,screen,dvipsnames,nonacm]{acmart}
\pdfoutput=1

\makeatletter

\usepackage{hyperref}
\usepackage{enumitem}
\usepackage{cleveref}
\usepackage{graphicx}
\usepackage{mathpartir}
\usepackage{mathtools}
\usepackage{stmaryrd}
\usepackage{subcaption}
\usepackage{xcolor}
 \hypersetup{
    colorlinks,
    linkcolor={red!50!black},
    citecolor={blue!50!black},
    urlcolor={blue!80!black}
}
\usepackage{wrapfig}
\usepackage{thmtools} 
\usepackage{thm-restate}
\usepackage{xr}

\usepackage{utfsym}

\def\extended{0}
\usepackage{olmacros}

\ifx\extended\undefined
\fi

\ifx\apponly\undefined\else
\usepackage[41-71]{pagesel}
\fi



\setcopyright{acmlicensed}
\acmJournal{TOPLAS}
\acmYear{2025} 
\acmDOI{10.1145/3743131}

\usepackage{fancyhdr}
\begin{document}
\title{Outcome Logic: A Unified Approach to the Metatheory of Program Logics with Branching Effects}

\author{Noam Zilberstein}
\email{noamz@cs.cornell.edu}
\orcid{0000-0001-6388-063X}
\affiliation{%
  \institution{Cornell University}
  \city{New York, NY}
  \country{USA}
}

  \thanks{{\textcopyright\ 2025 Noam Zilberstein.
  This is the author's version of the work. It is posted here for your personal use. Not for redistribution. The definitive version was published in \emph{ACM Transactions on Programming Languages and Systems} (TOPLAS).
  \\
  \texttt{\href{https://doi.org/10.1145/3743131}{https://doi.org/10.1145/3743131}}.}}

\begin{abstract}
Starting with Hoare Logic over 50 years ago, numerous program logics have been devised to reason about the different kinds of programs encountered in the real world. This includes reasoning about computational effects, particularly those effects that cause the program execution to branch into multiple paths due to, \eg nondeterministic or probabilistic choice.

Outcome Logic reimagines Hoare Logic with branching at its core, using an algebraic representation of choice to capture programs that branch into many outcomes. In this article, we give a comprehensive account of the Outcome Logic metatheory. This includes a relatively complete proof system for Outcome Logic with the ability to reason about general purpose looping. We also show that this proof system applies to programs with various types of branching, that it subsumes some well known logics such as Hoare Logic,  and that it facilitates the reuse of proof fragments across different kinds of specifications.
\end{abstract}

\maketitle              

\section{Introduction}

The seminal work of \citet{FLOYD67} and \citet{hoarelogic} on program logics in the 1960s has paved the way towards modern program analysis. 
The resulting \emph{Hoare Logic}---still ubiquitous today---defines triples $\hoare PCQ$ to specify the behavior of a program $C$ in terms of a precondition $P$ and a postcondition $Q$.
In the ensuing years, many variants of Hoare Logic have emerged, in part to handle the numerous computational effects found in real-world programs.

Relevant effects include nontermination, arising from while loops; nondeterminism, useful for modeling adversarial behavior or concurrent scheduling; and randomization, required for security and machine learning applications.
These effects have historically warranted specialized program logics with distinct inference rules. For example, partial correctness \cite{FLOYD67,hoarelogic} vs total correctness \cite{manna1974axiomatic} can be used to specify that the postcondition holds \emph{if} the program terminates vs that it holds \emph{and} the programs terminates, respectively. While Hoare Logic has classically taken a demonic view of nondeterminism (the postcondition must apply to \emph{all} possible outcomes), recent work on formal methods for incorrectness \cite{il,ilalgebra} has motivated the need for new program logics based on angelic nondeterminism (the postcondition applies to \emph{some} reachable outcome). Further, probabilistic Hoare Logics are quantitative, allowing one to specify the likelihood of each outcome, not just that they may occur \cite{hartog,ellora,rand2015vphl,den_hartog1999verifying}.

Despite the fact that distinct logics are used to handle different types of branching (\eg nondeterministic vs probabilistic), all of the aforementioned program logics share common reasoning principles. For instance, no-ops preserve the precondition, sequences of commands $C_1\fatsemi C_2$ are analyzed compositionally and the precondition (resp., postcondition) can be strengthened (resp., weakened) using logical consequences, as shown below.
\begin{mathpar}
\inferrule{\;}{\hoare P\skp{P}}

\inferrule{\hoare P{C_1}Q \\ \hoare Q{C_2}R}{\hoare P{C_1\fatsemi C_2}R}

\inferrule{P'\Rightarrow P \\ \hoare PCQ \\ Q \Rightarrow Q'}{\hoare{P'}C{Q'}}
\end{mathpar}
In this article, we show that those common reasoning principles are no mere coincidence. We develop \emph{Outcome Logic} (OL), which provides a uniform metatheory for a variety of computational effects---including nondeterminism and randomization---culminating in a single proof system for all of them.
We also show how specialized reasoning principles (\eg loop invariants for partial correctness) are derived from more general rules and how proof fragments can be shared between programs with different effects.

This work is also valuable in the context of static analysis.
Recent interest in bug finding tools \cite{gorogiannis2019true,blackshear2018racerd} prompted the development of Incorrectness Logic (IL) \cite{il,isl,ilalgebra,cisl}, which under-approximates a program's reachable states in order to identify true positive bugs. This theory was put to use in Meta's Pulse tool, which is deployed in large scale industrial codebases \cite{realbugs}.
Subsequently---and largely with the goal of consolidating static analysis tools---more logics were proposed to capture \emph{both} correctness \emph{and} incorrectness.
Some of those approaches involved combining Hoare Logic and Incorrectness Logic \cite{maksimovi_c2023exact,Bruni2021ALF,brunijacm}, but came with the significant downside of needing to simultaneously over- and under-approximate the set of reachable states, thereby inheriting the weaknesses of both logics without being able to leverage their strengths. One specific downside of IL is the inability to reason in abstract domains \cite{ascari}.

\citet{ilalgebra} and \citet{outcome} proposed to instead capture incorrectness with an \emph{angelic} version of Hoare Logic, often referred to as \emph{Lisbon Logic}, which has the advantage of easily identifying \emph{manifest errors} \cite{realbugs}, bugs that occur regardless of context. It was later shown that Pulse does not take advantage of Incorrectness Logic's unique features, and therefore could equivalently be modeled using Lisbon Logic \cite{raad2024nontermination,zilberstein2024outcome}. Lisbon Logic was then further developed as part of new bug finding tools, for which Incorrectness Logic was not a good fit \cite{ascari2025revealing,raad2024nontermination} and there is ongoing interest in developing more consolidated tools \cite{loow2024compositional,zilberstein2024unified}.

The key insight from the aforementioned work is that correctness and incorrectness can be represented by a choice between demonic or angelic nondeterminism. Outcome Logic supports both angelic and demonic reasoning, as well as more expressive ways to reason about the reachability of multiple different nondeterministic outcomes, or other effects such as probabilistic computation.
Outcome Logic was first proposed as a unified basis for correctness and incorrectness reasoning in nondeterministic and probabilistic programs, with semantics parametric on a \emph{monad} and a \emph{monoid} \cite{outcome}. The semantics was later refined such that each trace is weighted using an element of a semiring \cite{zilberstein2024outcome}. For example, Boolean weights specify which states are in the set of outcomes for a nondeterministic program, whereas real-valued weights quantify the probabilities of outcomes in a probabilistic program. Exposing these weights in pre- and postconditions means that a single program logic can express multiple termination criteria, angelic \emph{and} demonic nondeterminism, probabilistic properties, and more.

But while the aforementioned papers demonstrated the value of Outcome Logic in correctness and incorrectness applications, they left large gaps. Most notably, no proof strategies were given for unbounded iteration---loops could only be analyzed via bounded unrolling---and the connections to other logics were not deeply explored. In this article, we fill those gaps and provide a more authoritative and comprehensive reference on the Outcome Logic metatheory. In addition to demonstrating the applications of Outcome Logic to correctness and incorrectness reasoning, we provide a cleaner account of the semiring weights, which supports more models; we provide inference rules for unbounded looping and a relative completeness proof; and we more deeply explore the connections between OL and other logics by showing how the rules of those logics can be derived from the OL proof system.
The structure and contributions are follows:

\begin{itemize}[leftmargin=*]

\item We give an overview of the technical approach, highlighting the key ideas that will be formalized throughout the remainder of the article (\Cref{sec:overview}).

\item We give a cleaner Outcome Logic semantics and give six models (\Cref{ex:detsemi,ex:powersetsemi,ex:multisetsemi,ex:probsemi,ex:tropicalsemi,ex:langsemi}), including a multiset model (\Cref{ex:multisetsemi}) not supported by previous formalizations due to more restrictive algebraic constraints (\Cref{sec:semantics,sec:ol}). Our new looping construct naturally supports deterministic (while loops), nondeterministic (Kleene star), and probabilistic iteration---whereas previous OL versions supported fewer kinds of iteration \cite{zilberstein2024outcome} or used a non-unified, ad-hoc semantics \cite{outcome}.

\item We provide an extensional proof system based on semantic pre- and postconditions and prove that it is sound and relatively complete (\Cref{sec:proof}). Relative completeness means that the rules are sufficient for deriving any true specification, provided an oracle for deciding consequences between assertions used in pre- and postcondition \cite{cook1978soundness}; it is the best case scenario for program logics because those consequences are necessarily undecidable \cite{apt1981ten}. This is the first OL proof system that handles loops that iterate an indeterminate number of times. Our \ruleref{Iter} rule is sufficient for analyzing any iterative command, and from it we derive the typical loop invariant rule (for partial correctness), loop variants (with termination guarantees), and probabilistic loops (\Cref{sec:derived,sec:dynamichoare}).

\item In \Cref{sec:connections}, we prove that OL subsumes Hoare Logic and derive the entire Hoare Logic proof system (\eg loop invariants) in Outcome Logic. Inspired by Dynamic Logic \cite{pratt1976semantical,harel2001dynamic}, our encoding of Hoare Logic uses modalities to generalize partial correctness to types of branching beyond just nondeterminism. We also show that OL subsumes Lisbon Logic (a new logic for incorrectness), and its connections to Hyper Hoare Logic for proving hyperproperties of nondeterministic programs \cite{hyperhoare}.

\item Through case studies, we demonstrate the reusability of proofs across different effects (\eg nondeterminism or randomization) and properties (\eg angelic or demonic nondeterminism) (\Cref{sec:examples}). Whereas choices about how to handle loops typically require selecting a specific logic (\eg partial vs total correctness), loop analysis strategies can be mixed within a single OL derivation; we discuss the implications to program analysis in \Cref{sec:analysis}.
We also perform combinatorial analysis of graph algorithms based on alternative computation models (\Cref{sec:graphs}).

\item We contextualize the paper in terms of related work (\Cref{sec:discussion}) and discuss limitations and opportunities for future development (\Cref{sec:limitations}).
\end{itemize}

\section{Overview}
\label{sec:overview}

Outcome Logic is formalized in layers. At the bottom layer, there is a denotational model of programs, which uses various kinds of weights to represent computational effects arising due to different kinds of branching. Those weights are then reflected up into the logical layer, where they are used to quantify collections of outcomes, providing significant expressive power. In this section, we give an informal overview of these semantic and logical models, and show their applicability toward correctness and incorrectness.

\subsection{Unified Denotational Semantics}
\label{sec:ov-semantics}
Typical denotational models select a particular mathematical domain, corresponding to the computational effects of the program. For example, semantics of nondeterministic programs use powerdomains \cite{plotkin1976powerdomain,smyth1978power} to represent the result of the program as a set of possible states, \ie let $C$ be a program and $\mathcal S$ be the set of all states, then the semantics is captured by a function $\de{C}_{\mathsf{ND}} \colon \mathcal S \to \bb{2}^{\mathcal S}$ (where $\bb{2}^X$ is the powerset of $X$). A program that involves a nondeterministic choice---expressed as $C_1 + C_2$---produces a set of two outcomes:
\[
  \de{(x \coloneqq 0) + (x \coloneqq 1)}_{\mathsf{ND}}(s) = \{ s[x\coloneqq 0], s[x\coloneqq 1] \}
\]
On the other hand, probabilistic programs are interpreted in the probabilistic powerdomain \cite{jones1989probabilistic}; the result of the program is represented as a distribution of outcomes $\de{C}_{\mathsf{Prob}} \colon \mathcal S \to \mathcal D(\mathcal S)$, where distributions are maps from states to probabilities $\mathcal D(\mathcal S) = \mathcal S \to [0,1]$. A convex choice $C_1 +_p C_2$---where $C_1$ is executed with probability $p$ and $C_2$ is executed with probability $1-p$---results in a distribution over two outcomes.
\[
  \de{(x \coloneqq 0) +_p (x \coloneqq 1)}_{\mathsf{Prob}}(s) = \left\{
    \begin{array}{lll}
      s[x\coloneqq 0] & \mapsto & p
      \\
      s[x\coloneqq 1] & \mapsto & 1-p
    \end{array}
  \right.
\]
While these two computational domains at first appear distinct, they can in fact be unified through the lens of \emph{weighted programming} \cite{batz2022weighted}, where different kinds of weights are used to represent different kinds of effects. It is quite clear that the weights in a probabilistic program are real numbers on the unit interval, \ie probabilities. That is, the semantics of a probabilistic program weights each end state by a probability.

However, it is also well known that sets $S \in \bb{2}^{\mathcal S}$ are isomorphic to maps into the Booleans $f \colon \mathcal S \to \mathbb B$ (where $\mathbb B = \{0, 1\}$), \ie $S \cong f$ iff $S = \{ s\in\mathcal S \mid f(s) = 1\}$. Using this fact, we can rewrite the nondeterministic semantics above to look more similar to the probabilistic one, omitting any states that map to 0.
\[
  \de{(x \coloneqq 0) + (x \coloneqq 1)}_{\mathsf{ND}}(s) = \left\{
    \begin{array}{lll}
      s[x\coloneqq 0] & \mapsto & 1
      \\
      s[x\coloneqq 1] & \mapsto & 1
    \end{array}
  \right.
\]
Probabilities and Booleans are two examples of valid weights, which more generally must carry a \emph{semiring} structure, meaning that they can be added and multiplied. Addition corresponds to branching; the semantics of the two programs above can more abstractly be captured by a lifted addition operation $+$, which sums the weights of each outcome in both computation branches:
\[
  \de{C_1 + C_2}(s) = \de{C_1}(s) + \de{C_2}(s)
\]
In the Boolean instance, addition is logical disjunction, meaning that the weight of an outcome is 1 if it appears in either branch. In the probabilistic instance, addition has the standard arithmetic meaning, so the weight of an outcome is the cumulative probability across the two branches.

Multiplication is used to interpret sequential composition. In the Boolean instance, multiplication corresponds to logical conjunction, so in the program $(C_1 + C_2) \fatsemi (C_3 + C_4)$, the weight of executing both $C_1$ and $C_3$ is $1\land 1 = 1$. In the probabilistic instance, multiplication is also arithmetic, so the weight of executing both $C_1$ and $C_3$ in $(C_1 +_p C_2) \fatsemi (C_3 +_q C_4)$ is $p\cdot q$. In \Cref{sec:semantics}, we will give all the formal properties of these weights, and introduce a few more interpretations.

\subsection{A Unified Program Logic}
\label{ov:logic}

Outcome Logic's expressive power comes from exposing the semiring weights from \Cref{sec:ov-semantics} to the program logic, so that specifications can quantify the weight of each outcome. While this ability has previously featured in probabilistic logics---where expressing the probabilities of each event is essential \cite{hartog,ellora,rand2015vphl,den_hartog1999verifying}---it had not appeared in a logic for nondeterministic programs. Exposing the weights allows Outcome Logic to reason about nondeterminism in either an angelic or demonic style, or even mix the two. This mixture makes Outcome Logic a good foundation for unifying correctness and incorrectness, as we will see in \Cref{sec:ov-incorrectness}. In addition, the metatheoretic unification means that a single logic can be used for quantitative reasoning about probabilistic programs \emph{and} correctness and incorrectness in nondeterministic ones.

Weights are exposed to the logic via \emph{outcome assertions}, which are described formally in \Cref{sec:assertions}. In this section, we focus on two of these assertions. The first one---$\wg Pu$---lifts a basic assertion $P$ describing individual states to an assertion about weighted collections of states, such that the collection has cumulative weight $u$ and every state with nonzero weight satsifies $P$.
When the cumulative weight is 1, we simply write $\sure P$. To describe how the states are distributed across a collection, the \emph{outcome conjunction} $\varphi\oplus\psi$ asserts that the weight of the collection is split across the two outcomes $\varphi$ and $\psi$. This means that if a program is composed of two branches $C_1 + C_2$, then an outcome conjunction can be used to specify the result of that program:
\[
    \de{C_1}(s) \vDash \varphi
    \quad\text{and}\quad
    \de{C_2}(s) \vDash \psi
    \quad\implies\quad
    \de{C_1 + C_2}(s) \vDash \varphi \oplus \psi
\]
Just as in Hoare Logic, Outcome Logic judgements are given as triples $\triple{\varphi}C\psi$, but the pre- and postconditions are outcome assertions rather than assertions about individual states. This allows for inference rules to compositionally reason about the branching behavior of a program.
\[
\inferrule{
  \triple{\varphi}{C_1}{\psi_1}
  \\
  \triple{\varphi}{C_2}{\psi_2}
}{
  \triple{\varphi}{C_1 + C_2}{\psi_1 \oplus \psi_2}
}{\ruleref{Plus}}
\]
Stemming from this one rule, we can derive specifications for both the probabilistic and nondeterministic programs in \Cref{sec:ov-semantics}. The probabilistic specification expresses the precise probabilities of the two outcomes.
\[
\inferrule{
    \triple{\sure{\tru}}{x \coloneqq 0}{\sure{ x=0 }}
    \\
    \triple{\sure{\tru}}{x \coloneqq 1}{\sure{ x=1 }}
}{
  \triple{\sure{\tru}}{(x \coloneqq 0) +_p (x \coloneqq 1)}{\sure{x=0}^{(p)} \oplus \sure{x=1}^{(1-p)}}
}
\]
The nondeterministic one also gives a precise weight for the outcomes (\ie 1), meaning that both are definitely represented in the set of possible end states.
\[
\inferrule{
    \triple{\sure{\tru}}{x \coloneqq 0}{\sure{ x=0 }}
    \\
    \triple{\sure{\tru}}{x \coloneqq 1}{\sure{ x=1 }}
}{
  \triple{\sure{\tru}}{(x \coloneqq 0) + (x \coloneqq 1)}{\sure{x=0} \oplus \sure{x=1}}
}
\]
As such, this specification is a mix of both demonic and angelic styles. On the angelic side, we know there \emph{exist} traces that reach both outcomes. On the demonic side, we know that \emph{all} traces fall into the conjunction of the two outcomes. This is something that prior logics could not express, and it is useful for unifying correctness and incorrectness reasoning, as we will now see.

\newcommand{\bug}{\raisebox{-1pt}{\usym{1F41C}}}
\subsection{Applications to Correctness and Incorrectness}
\label{sec:ov-incorrectness}

We begin by explaining why typical \emph{correctness} logics like Hoare Logic are not suitable for incorrectness reasoning. In a Hoare triple $\hoare PCQ$, all reachable states must fall within the postcondition $Q$. Suppose that $C$ is a nondeterministic program, which \emph{sometimes} displays a bug ($\bug$), and sometimes results in a desirable outcome ($\checkmark$):
\[
  \de{C}(s) = \left\{
    \begin{array}{lll}
      \bug &\mapsto& 1
      \\
      \checkmark &\mapsto& 1
    \end{array}
  \right.
\]
Now let $Q_\ok$ be the desired postcondition of the program, where $\checkmark \vDash Q_\ok$ and $\bug \not\vDash Q_\ok$. The above program is incorrect, because it does not satisfy the correctness specification $\hoare PC{Q_\ok}$. However, the best we can do to specify the bug in Hoare logic is to use a disjunction: $\hoare PC{Q_\ok \vee Q_\er}$, where $Q_\er$ is a description of the bug such that $\bug\vDash Q_\er$. The problem is that the disjunction is not strong enough to guarantee the reachability of the bug; a program where $\checkmark$ is the only outcome would also satisfy the postcondition $Q_\ok \vee Q_\er$, meaning that reporting a bug based on the above Hoare triple may in fact be a false positive. This is a significant problem, and in practice many bugs do arise nondeterministically, \eg when the programmer forgets a null check after allocating a pointer \cite{il,isl}.

As we already saw, the above problem can be mitigated in Outcome Logic by using the specification $\triple{\sure P}C{\sure{Q_\ok} \oplus \sure{Q_\er}}$, which guarantees that $Q_\er$ is non-vacuously reachable. In fact, we can even drop information about the \emph{ok} outcome, to obtain the simpler specification $\triple{\sure P}C{\sure{Q_\er} \oplus \top}$, stating that $Q_\er$ is reachable, and additional outcomes may be reachable too ($\top$ is a trivial assertion that is satisfied by any collection of states).
Triples of this type are equivalent to specifications in Lisbon Logic, an increasingly popular foundation for bug finding tools \cite{ascari2025revealing,raad2024nontermination}. As we show in \Cref{sec:dynamichoare}, Outcome Logic subsumes Hoare Logic as well.

However, the real power of Outcome Logic comes not only from the ability to express correctness and incorrectness \emph{separately}, but also to express specifications that can be repurposed for both. For example, consider a program $C_1 \fatsemi C_2$, where nondeterministic branching occurs in $C_1$, and depending on which branch the program ends up in, $C_2$ may or may not result in a bug:
\[
  \de{C_1}(s) = \left\{
    \begin{array}{lll}
      s_1 &\mapsto& 1
      \\
      s_2 &\mapsto& 1
    \end{array}
  \right.
  \qquad
  \de{C_2}(s_1) = \left\{
    \begin{array}{lll}
      \bug &\mapsto& 1
    \end{array}
  \right.
  \qquad
  \de{C_2}(s_2) = \left\{
    \begin{array}{lll}
      \checkmark &\mapsto& 1
    \end{array}
  \right.
\]
Letting $Q_1$ and $Q_2$ be such that $s \vDash Q_i$ iff $s = s_i$ for each $i\in\{1,2\}$, we could obtain the triples below.
\begin{mathpar}
  \triple{\sure P}{C_1}{\sure{Q_1}\oplus\sure{Q_2}}

  \triple{\sure{Q_1}}{C_2}{\sure{Q_\er}}

  \triple{\sure{Q_2}}{C_2}{\sure{Q_\ok}}
\end{mathpar}
Now, suppose that we are analyzing the program in a forward fashion. After obtaining the above specification for $C_1$, we do not yet know \emph{if} a bug will be encountered, or \emph{which branch} will encounter it. As such, it is useful to retain the intermediate assertion $\sure{Q_1} \oplus \sure{Q_2}$ to guarantee that both branches are reachable (in case one encounters a bug) and also that they cover all possible outcomes (in case we want to ensure correctness). To complete the proof, we need the two composition rules shown below: \ruleref{Seq} is the standard sequential composition rule from, \eg Hoare Logic; and \ruleref{Choice} allows us to compositionally reason about each outcome.
\begin{mathpar}
  \inferrule*[right=\ruleref{Seq}]{
    \triple{\varphi}{C_1}{\vartheta}
    \\
    \triple{\vartheta}{C_2}{\psi}
  }{
    \triple{\varphi}{C_1\fatsemi C_2}{\psi}
  }

  \inferrule*[right=\ruleref{Choice}]{
    \triple{\varphi_1}C{\psi_1}
    \\
    \triple{\varphi_2}C{\psi_2}
  }{
    \triple{\varphi_1 \oplus \varphi_2}C{\psi_1\oplus\psi_2}
  }
\end{mathpar}
By applying these rules, we can obtain a specification for the complete program with a postcondition indicating that an error state is reachable, signifying that the program is \emph{incorrect}.
\[
\inferrule*[right=\ruleref{Seq}]{
  \triple{\sure P}{C_1}{\sure{Q_1}\oplus\sure{Q_2}}
  \\
  \inferrule*[Right=\ruleref{Choice}]{
    \triple{\sure{Q_1}}{C_2}{\sure{Q_\er}}
    \\
    \triple{\sure{Q_2}}{C_2}{\sure{Q_\ok}}
  }{
    \triple{\sure{Q_1}\oplus\sure{Q_2}}{C_2}{\sure{Q_\er} \oplus \sure{Q_\ok}}
  }
}{
  \triple{\sure P}{C_1\fatsemi C_2}{\sure{Q_\er} \oplus \sure{Q_\ok}}
}
\]
Of course, if a bug is encountered on one side of the $\oplus$, then it is not necessarily worthwhile to continue analyzing the other outcomes. This was one of \cites{il} key motivations in developing Incorrectness Logic, as its ability to \emph{drop disjuncts} from the postcondition meant that extraneous traces could be ignored. In our case, knowing that ${Q_\ok}$ is a reachable outcome is unimportant; the fact that ${Q_\er}$ is reachable already tells us everything that we need to know. So, it desirable to be able to \emph{drop outcomes} using that following inference rule that allows us to trivially analyze a program using the vacuous postcondition $\top$, which is satisfied by any collection of states.
\begin{mathpar}
%
\inferrule*[right=\ruleref{True}]{\;}{
  \triple{\varphi}C{\top}
}
\end{mathpar}
Indeed, the specification $\triple{\sure{Q_2}}{C_2}{\sure{Q_\ok}}$ may not be simple to prove; $C_2$ itself may be a complex sequence of commands. By contrast, \ruleref{True} can be applied regardless of how complex $C_2$ is. With this in mind, we perform a simpler derivation for the composite program below.
\[
\inferrule*[right=\ruleref{Seq}]{
  \triple{\sure P}{C_1}{\sure{Q_1}\oplus\sure{Q_2}}
  \\
  \inferrule*[Right=\ruleref{Choice}]{
    \triple{\sure{Q_1}}{C_2}{\sure{Q_\er}}
    \\
    \inferrule*[Right=\ruleref{True}]{\;}{
      \triple{\sure{Q_2}}{C_2}{\top}
    }
  }{
    \triple{\sure{Q_1}\oplus\sure{Q_2}}{C_2}{\sure{Q_\er} \oplus \top}
  }
}{
  \triple{\sure P}{C_1\fatsemi C_2}{\sure{Q_\er} \oplus \top}
}
\]
Here, we have obtained the postcondition $\sure{Q_\er}\oplus\top$, which states that $Q_\er$ is reachable and there may be any number of other reachable outcomes too.

If $C_1$ were reused in a different context $C_1 \fatsemi C_3$, where $C_3$ does not encounter a bug on state $s_1$, then we could instead collapse the two outcomes into a concise correctness specification $\triple{\sure P}{C_1\fatsemi C_2}{\sure{Q_\ok}}$\footnote{This could be done either by analyzing $C_3$ twice using two triples $\triple{\sure{Q_1}}{C_3}{\sure{Q_\ok}}$ and $\triple{\sure{Q_2}}{C_3}{\sure{Q_\ok}}$, or by first weakening the postcondition $\sure{Q_1}\oplus\sure{Q_2}\Rightarrow\sure Q$ for some $Q$ such that $Q_1\Rightarrow Q$ and $Q_2\Rightarrow Q$, and then proving a single specification $\triple{\sure Q}{C_3}{\sure{Q_\ok}}$.}. Importantly, this correctness specification does not involve $\top$---the program is only correct if it behaves properly in \emph{all} traces. The ability to witness reachability of multiple outcomes is a unique capability of Outcome Logic---which cannot be achieved in any prior logics such as Hoare Logic, Incorrectness Logic, or Lisbon Logic---and is what makes it suitable for both correctness and incorrectness.

\medskip

\noindent In the remainder of this article, we formalize the details that have been outlined thus far. We begin with the algebraic properties of weights and program semantics in \Cref{sec:semantics}. Next, we define Outcome Logic and its proof system in \Cref{sec:ol,sec:state}. In \Cref{sec:connections}, we prove that Outcome Logic subsumes a variety of other logics including Hoare Logic and Lisbon Logic, and derive the proof systems of those logics as well. We then present a variety of case studies in \Cref{sec:examples,sec:graphs}, and finally we conclude by discussing related and future work in \Cref{sec:discussion,sec:limitations}.

\section{Weighted Program Semantics}\label{sec:semantics}

We begin the technical development by defining a basic programming language and describing its semantics based on various interpretations of choice. The syntax for the language is shown below.
\begin{align*}
C &\Coloneqq\skp \mid C_1 \fatsemi C_2 \mid C_1 + C_2 \mid \assume e \mid
\iter{C}{e}{e'}
 \mid
 a 
 \qquad (a\in\mathsf{Act}, u\in U, t\in\mathsf{Test})\\
e &\Coloneqq b \mid u \\
b &\Coloneqq \tru \mid \fls \mid b_1 \lor b_2 \mid b_1 \land b_2 \mid \lnot b \mid t
\end{align*}
This language is similar to imperative languages such as \cites{GCL} Guarded Command Language (GCL), with familiar constructs such as $\skp$, sequential composition ($C_1 \fatsemi C_2$), branching ($C_1 + C_2$), and atomic actions $a \in\mathsf{Act}$. The differences arise from the generalized $\code{assume}$ operation, which weights the current computation branch using an expression $e$ (either a test $b$ or a \emph{weight} $u \in U$, to be described fully in \Cref{sec:algebra}). We already saw two examples of weights in \Cref{sec:overview}: probabilities and Booleans.

Weighting is also used in the iteration command $\iter C{e}{e'}$, which iterates $C$ with weight $e$ and exits with weight $e'$. It is a generalization of the Kleene star $C^\star$, and is also more general than the iteration constructs found in previous Outcome Logic work \cite{outcome,zilberstein2024outcome}. In \Cref{sec:sugar}, we will show how to encode while loops, Kleene star, and probabilistic loops using $\iter Ce{e'}$. Although the latter constructs can be encoded using while loops and auxiliary variables, capturing this behavior \emph{without} state can be advantageous, as it opens up the possibility for equational reasoning over programs with uninterpreted atomic commands \cite{KAT,r_o_zowski2023probabilistic}.

Tests $b$ contain the typical operations of Boolean algebras as well as primitive tests $t \in \mathsf{Test}$, assertions about a program state. Primitive tests are represented semantically, so $\mathsf{Test} \subseteq \bb{2}^\Sigma$ where $\Sigma$ is the set of program states (each primitive test $t\subseteq \Sigma$ is the set of states that it describes). Tests evaluate to $\zero$ or $\one$, which represent the Boolean values false and true, respectively.

The values $\zero$ and $\one$ are examples of weights from the set $\{\zero,\one\}\subseteq U$. These weights have algebraic properties, which were introduced in \Cref{sec:ov-semantics} and will be described fully in \Cref{sec:algebra}. Using a test $b$, the command $\assume b$ chooses whether or not to continue evaluating the current branch, whereas $\assume u$ more generally picks a weight for the branch, which may be a Boolean ($\zero$ or $\one$), but may also be some other type of weight such as a probability.
In the remainder of this section, we will define the semantics formally.

%

\subsection{Algebraic Preliminaries}\label{sec:algebra}

We begin by reviewing some algebraic structures.
As explained in \Cref{sec:ov-semantics}, our program semantics relies on semiring weights to represent different kinds of branching.
We begin by formally defining the properties of those weights.
%
\begin{definition}[Monoid] A monoid $\langle U, +, \zero\rangle$ consists of a carrier set $U$, an associative binary operation $+\colon U\times U\to U$, and an identity element $\zero\in U$ ($u + \zero = \zero + u = u$). If $+\colon U\times U\rightharpoonup U$ is partial, then the monoid is partial. If $+$ is commutative ($u+v = v+u$), then the monoid is commutative.
\end{definition}
As an example, $\langle \{0,1\}, \lor, 0\rangle$ is a monoid on Booleans. Building on monoids, we now introduce semirings. As we saw in \Cref{sec:overview}, the outcome weights used to formalize our program semantics must carry a semiring structure.
%
\begin{definition}[Semiring]\label{def:semiring}
A semiring $\langle U, +, \cdot, \zero, \one\rangle$ is a structure such that $\langle U, +, \zero\rangle$ is a commutative monoid, $\langle U,\cdot, \one\rangle$ is a monoid, and the following holds:
\begin{enumerate}[leftmargin=2em]
\item Distributivity: $u \cdot (v+w) = u\cdot v + u\cdot w$ and $(u+v)\cdot w = u\cdot w + v\cdot w$
\item Annihilation: $\zero \cdot u = u\cdot \zero = \zero$
\end{enumerate}
The semiring is partial if $\langle U, +, \zero\rangle$ is a partial monoid (but $\cdot$ remains total).
\end{definition}
In the remainder of this section, we introduce additional properties of semirings needed to make the semantics of loops well defined.
First, we define an order based on the addition operator of the semiring.
\begin{definition}[Natural Ordering]
Given a (partial) semiring $\langle U,+,\cdot,\zero,\one\rangle$, the natural order is defined to be:
\[
u\le v
\qquad\text{iff}\qquad
\exists w.~ u+w=v
\]
The semiring is naturally ordered if the natural order $\le$ is a partial order. Note that $\le$ is trivially reflexive and transitive, but it remains to show that it is anti-symmetric.
\end{definition}
Based on this order, we also define the notion of Scott continuity, which will be important for defining infinite sums and the semantics of while loops. This definition relies on directed sets and dcpos. A set $D$ is directed iff every pair of elements $x,y \in D$ have an upper bound in $D$. A dcpo is a poset where every directed set has a supremum.
\begin{definition}[Scott Continuity \cite{semirings}]
\label{def:continuity}
A (partial) semiring with order $\le$ is Scott continuous if it is a dcpo and for any directed set $D\subseteq X$, the following hold:
\[
\sup_{x\in D} (x+y) = (\sup D) +y
\qquad
\sup_{x\in D} (x\cdot y) = (\sup D)\cdot y
\qquad
\sup_{x\in D} (y\cdot x) = y\cdot \sup D
\]
\end{definition}
Given a Scott continuous semiring, we can also define a notion of infinite sums.
\begin{definition}[Infinite Sums]
Let $\langle U, +, \cdot, \zero,\one\rangle$ be a Scott continuous semiring. Now, for any (possibly infinite) indexed collection $(u_i)_{i\in I}$, we define the following operator:
\[
  \sum_{i\in I} u_i \triangleq \sup \{ u_{i_1} + \cdots + u_{i_n} \mid n\in \mathbb N, \{ i_1, \ldots, i_n \} \subseteq_{\mathsf{fin}} I \}
\]
That is, an infinite sum is the supremum of the sums of all finite subsequences. As shown by \citet[Corollary 1]{kuich2011algebraic}, $\sum$ makes $U$ a complete semiring, meaning it upholds the following properties:
\begin{enumerate}
\item If $I = \{i_1, \ldots, i_n\}$ is finite, then $\sum_{i\in I} u_i = u_{i_1} + \cdots + u_{i_n}$
\item If $\sum_{i\in I}u_i$ is defined, then $v\cdot\sum_{i\in I} u_i = \sum_{i\in I} v\cdot u_i$ and $(\sum_{i\in I}u_i)\cdot v = \sum_{i\in I} u_i\cdot v$
\item Let $(J_k)_{k\in K}$ be a family of nonempty disjoint subsets of $I$ ($I = \bigcup_{k\in K} J_k$ and $J_k \cap J_\ell = \emptyset$ if $k\neq\ell$), then $\sum_{k\in K}\sum_{j\in J_k} u_j = \sum_{i\in I} u_i$
\end{enumerate}
\end{definition}

\subsection{Weighting Functions}

Semirings elements will act as the \emph{weights} for traces in our semantics. The interpretation of a program at a state $\sigma\in\Sigma$ will map each end state to a semiring element $\de{C}(\sigma) \colon \Sigma\to U$. Varying the semiring will yield different kinds of effects. For example, as we saw in \Cref{sec:overview}, the Boolean semiring where $U = \{0,1\}$ corresponds to nondeterminism; $\de{C}(\sigma) \colon \Sigma\to \{0,1\} \cong \bb{2}^\Sigma$ tells us which states are in the set of outcomes. A probabilistic semiring where $U = [0,1]$ (the unit interval of real numbers) gives us a map from states to probabilities---a \emph{distribution} of outcomes. More formally, the result is a \emph{weighting function}, defined below.

\begin{definition}[Weighting Function]\label{def:weighting}
Given a set $X$ and a partial semiring $\mathcal A = \langle U, +, \cdot, \zero,\one\rangle$, the set of weighting functions is:
\[
\mathcal W_{\mathcal A}(X) \triangleq \Big\{ m \colon X \to U  ~\Big| ~|m| ~\text{is defined and} ~ \supp(m) ~\text{is countable} \Big\}
\]
Where $\supp(m) \triangleq \{ \sigma \mid m(\sigma)\neq\zero \}$ and $|m| \triangleq {\sum_{\sigma\in\supp(m)}} m(\sigma)$.
\end{definition}

Weighting functions can encode the following types of computation.

\begin{example}[Nondeterminism]
\label{ex:powersetsemi}
Nondeterministic computation is based on the Boolean semiring $\mathsf{Bool}=\langle \mathbb B, \vee, \wedge, 0, 1\rangle$, where weights are drawn from $\mathbb B = \{0,1\}$ and \emph{conjunction} $\land$ and \emph{disjunction} $\lor$ are the usual logical operations. This gives us $\mathcal W_{\mathsf{Bool}}(X) \cong \bb{2}^X$---weighting functions on $\mathsf{Bool}$ are isomorphic to sets, \ie any $m \in \mathcal W_{\mathsf{Bool}}(X)$, is isomorphic to the set $\{ x \in X \mid m(x) = 1\}$.

\end{example}
\begin{example}[Determinism]
\label{ex:detsemi}
Deterministic computation also uses Boolean weights, but with a different interpretation of the semiring $+$. That is, $0\dotplus x = x\dotplus 0 = x$, but $1\dotplus 1$ is undefined. The semiring is therefore $\mathsf{Bool}' = \langle \mathbb B, \dotplus, \land, 0, 1\rangle$. With this definition of $\dotplus$, the requirement of \Cref{def:weighting} that $|m|$ is defined means that $|\supp(m)| \le 1$, so we get that $\mathcal W_{\mathsf{Bool}'}(X) \cong X + \{\lightning\}$---it is either a single value $x\in X$, or $\lightning$, indicating that the program diverged.

\end{example}
\begin{example}[Multiset Nondeterminism]
\label{ex:multisetsemi}
Rather than indicating which outcomes are possible using Booleans, we use natural numbers (extended with $\infty$) $n \in \mathbb{N}^\infty$ to count the traces leading to each outcome. This yields the semiring $\mathsf{Nat} = \langle \mathbb{N}^\infty, +, \cdot, 0, 1\rangle$ with the standard arithmetic operations, and we get that $\mathcal W_{\mathsf{Nat}}(X) \cong \mathcal M(X)$ where $\mathcal M(X)$ is the set of \emph{multisets} over $X$. Multisets are useful for encoding probabilistic nondeterminism \cite{jacobs2021multisets,kozen2023multisets,ong2024probability}.
\end{example}
\begin{example}[Randomization]
\label{ex:probsemi}
Probabilities $p \in [0,1] \subset \mathbb R$ form a partial semiring $\mathsf{Prob}=\langle [0,1], +, \cdot, 0, 1\rangle$ where $+$ and $\cdot$ are real-valued arithmetic operations, but $+$ is undefined if $x+y > 1$ (just like in \Cref{ex:detsemi}). This gives us $\mathcal W_{\mathsf{Prob}}(X) \cong \mathcal D(X)$, where $\mathcal D(X)$ is the set of discrete probability sub-distributions over $X$ (the mass can be less than 1 if some traces diverge).
\end{example}


\begin{example}[Tropical Computation]
\label{ex:tropicalsemi}
 $\mathsf{Tropical}=\langle [0,\infty], \min, +, \infty, 0\rangle$ uses real-valued weights, but semiring addition is minimum and semiring multiplication is arithmetic addition. Computations in $\mathcal W_{\mathsf{Tropical}}(X)$ correspond to programs that choose the \emph{cheapest} path for each outcome.
\end{example}

\begin{example}[Formal Languages]
\label{ex:langsemi}
For some alphabet $\Gamma$, let $\Gamma^\ast$ be the set of finite strings over that alphabet, $\Gamma^\omega$ is the set of infinite strings, and $\Gamma^\infty = \Gamma^\ast\cup \Gamma^\omega$ is the set of all strings. Then
$\mathsf{Lang} = \langle \bb{2}^{\Gamma^\infty}, \cup,  \cdot, \emptyset, \{\varepsilon\}\rangle$ is the semiring of formal languages (sets of strings) where addition is given by the standard union, multiplication is given by concatenation (if the first string is finite):
\[
\ell_1 \cdot \ell_2 \triangleq \{ st \mid s\in \ell_1, t\in\ell_2 \}
\qquad\text{where}\qquad
st = s \;\; \text{if} \;\; s\in\Gamma^\omega
\]
The unit of addition is the empty language and the unit of multiplication is the language containing only the empty string $\varepsilon$. Note that multiplication (\ie concatenation) is not commutative. Computations in $\mathcal W_{\mathsf{Lang}}(\Sigma)$ correspond to programs that log sequences of letters \cite{batz2022weighted}.
\end{example}

We will occasionally write $\mathcal W(X)$ instead of $\mathcal W_{\mathcal A}(X)$ when $\mathcal A$ is obvious. The semiring operations for addition, scalar multiplication, and zero are lifted pointwise to weighting functions as follows
\begin{align*}
(m_1 + m_2)(x) &\triangleq m_1(x) + m_2(x)
&
\zero(x) &\triangleq \zero
\\
(u \cdot m)(x) &\triangleq u\cdot m(x)
&
(m \cdot u)(x) &\triangleq m(x) \cdot u
\end{align*}
Natural orders also extend to weighting functions, where $m_1 \sqsubseteq m_2$ iff there exists $m$ such that $m_1 + m = m_2$. This corresponds exactly to the pointwise order, so $m_1 \sqsubseteq m_2$ iff $m_1(\sigma) \le m_2(\sigma)$ for all $\sigma\in\supp(m_1)$.

These lifted semiring operations give us a way to interpret branching, but we also need an interpretation for sequential composition.
As is standard in program semantics with effects, we use a monad, which we define as a Klesli triple \cite{manes1976algebraic,moggi91}.

\begin{definition}[Kleisli Triple]\label{def:kleisli}
A Kleisli triple $\langle T, \eta, (-)^\dagger\rangle$ in $\mathbf{Set}$ consists of a functor $T \colon \mathbf{Set} \to \mathbf{Set}$, and two morphisms $\eta : \mathsf{Id} \Rightarrow T$ and for any sets $X$ and $Y$, $(-)^\dagger \colon (X \to T(Y)) \to T(X) \to T(Y)$ such that:
\[
\eta^\dagger = \mathsf{id}
\qquad\qquad
f^\dagger \circ \eta = f
\qquad\qquad
f^\dagger \circ g^\dagger = (f^\dagger \circ g)^\dagger
\]
\end{definition}

For any semiring $\mathcal A$, $\langle \mathcal W_{\mathcal A}, \eta, (-)^\dagger\rangle$ is a Kleisli triple where $\eta$ and $(-)^\dagger$ are defined below.
\[\arraycolsep=3pt
\eta(x)(y) \triangleq \left\{
\begin{array}{lll}
\one & \text{if} & x=y \\
\zero& \text{if} & x\neq y
\end{array}\right.
\qquad\qquad
f^\dagger(m)(y) \triangleq \smashoperator{\sum_{x\in\supp(m)}} m(x)\cdot f(x)(y)
\]
%

\subsection{Denotational Semantics}\label{sec:denotation}

We interpret the semantics of our language using the five-tuple $\langle \mathcal A, \Sigma, \mathsf{Act}, \mathsf{Test}, \de{\cdot}_{\mathsf{Act}}\rangle$, where the components are:
\begin{enumerate}[leftmargin=2em]
\item $\mathcal A = \langle U, +, \cdot,\zero,\one\rangle$ is a naturally ordered, Scott continuous, partial semiring with a top element $\top\in U$ such that $\top\ge u$ for all $u\in U$.
\item $\Sigma$ is the set of concrete program states.
\item $\mathsf{Act}$ is the set of atomic actions.
\item $\mathsf{Test} \subseteq \bb{2}^\Sigma$ is the set of primitive tests.
\item $\de{-}_{\mathsf{Act}} \colon \mathsf{Act} \to \Sigma \to \mathcal W_{\mathcal A}(\Sigma)$ is the  interpretation of atomic actions.
\end{enumerate}
This definition is a generalized version of the one used in Outcome Separation Logic \cite{zilberstein2024outcome}. For example, we have dropped the requirement that $\top = \one$, meaning that we can capture more types of computation, such as the multiset model (\Cref{ex:multisetsemi}). Multiset nondeterminism has been shown to be useful in computational domains that combine nondeterministic and probabilistic computation \cite{kozen2023multisets,jacobs2021multisets,ong2024probability}.

Commands are interpreted denotationally as maps from states $\sigma\in\Sigma$ to weighting functions on states $\de{C}\colon \Sigma\to\mathcal W_{\mathcal A}(\Sigma)$, as shown in \Cref{fig:denotation}. 
The first three commands are defined in terms of the monad (\Cref{def:kleisli}) and semiring operations (\Cref{def:weighting,def:semiring}): $\skp$ uses $\eta$, sequential composition $C_1\fatsemi C_2$ uses $(-)^\dagger$, and $C_1 +C_2$ uses the (lifted) semiring $+$. Since $+$ is partial, the semantics of $C_1+C_2$ may be undefined. In \Cref{sec:sugar}, we discuss simple syntactic checks to ensure that the semantics is total. Atomic actions are interpreted using $\de{-}_{\mathsf{Act}}$.

\begin{figure}[t]
\begin{align*}
\de{\skp}(\sigma) &\triangleq \eta(\sigma) \\
\de{C_1 \fatsemi C_2}(\sigma) &\triangleq \dem{C_2}{\de{C_1}(\sigma)}\\
\de{C_1 + C_2}(\sigma) &\triangleq \de{C_1}(\sigma) + \de{C_2}(\sigma) \\
\de{a}(\sigma) &\triangleq \de{a}_{\mathsf{Act}}(\sigma) \\
\de{\assume e}(\sigma) &\triangleq \de{e}(\sigma)\cdot\eta(\sigma) \\
\de{\iter C{e}{e'}}(\sigma) &\triangleq \mathsf{lfp}\Big(\Phi_{\langle C, e, e' \rangle}\Big)(\sigma)
\end{align*}
\[
\text{where} \quad \Phi_{\langle C, e, e' \rangle}(f)(\sigma) \triangleq \de{e}(\sigma)\cdot f^\dagger(\de{C}(\sigma)) + \de{e'}(\sigma) \cdot \eta(\sigma)
\]
%
%
\caption{Denotational semantics for commands $\de{C} \colon \Sigma\rightharpoonup \mathcal W_{\mathcal A}(\Sigma)$, given a partial semiring $\mathcal A = \langle X,+, \cdot,\zero,\one\rangle$, a set of program states $\Sigma$, atomic actions $\mathsf{Act}$, primitive tests $\mathsf{Test}$, and an interpretation of atomic actions $\de{a}_{\mathsf{Act}} \colon \Sigma\to \mathcal W_{\mathcal A}(\Sigma)$.}
\label{fig:denotation}
\end{figure}

The interpretation of $\code{assume}$ relies on the ability to interpret expressions and tests. We first describe the interpretation of tests, which maps $b$ to the weights $\zero$ or $\one$, that is $\de{b}_{\mathsf{Test}} \colon \Sigma \to \{\zero,\one\}$, with $\zero$ representing false and $\one$ representing true, so $\de{\fls}(\sigma) = \zero$ and $\de{\tru}(\sigma) = \one$. The operators $\land$, $\lor$, and $\lnot$ are interpreted in the obvious ways, and for primitive tests $\de{t}(\sigma) = \one$ if $\sigma\in t$ otherwise $\de{t}(\sigma) = \zero$. The full semantics of tests is given in \Aref{app:tests}.

Since an expression is either a test or a weight, it remains only to describe the interpretation of weights, which is $\de{u}(\sigma) = u$ for any $u\in U$. So, $\assume e$ uses $\de{e}\colon \Sigma\to U$ to obtain a program weight, and then scales the current state by it. If a test evaluates to false, then the weight is $\zero$, so the branch is eliminated. If it evaluates to true, then it is scaled by $\one$---the identity of multiplication---so the weight is unchanged.
The iteration command continues with weight $e$ and terminates with weight $e'$. We can start by defining it recursively as follows.
\begin{align*}
\de{\iter C{e}{e'}}(\sigma)
&= \de{\assume{e} \fatsemi C \fatsemi \iter C{e}{e'} + \assume{e'}}(\sigma) \\
&= \de{e}(\sigma)\cdot \dem{\iter Ce{e'}}{\de C(\sigma)} + \de{e'}(\sigma)\cdot\eta(\sigma)
\end{align*}
To ensure that this equation has a well-defined solution, we formulate the semantics as a least fixed point. Requiring that the semiring is Scott continuous (\Cref{def:continuity}), ensures that this fixed point exists. For the full details, see \Aref{app:totality}.

\subsection{Syntactic Sugar for Total Programs}
\label{sec:sugar}

As mentioned in the previous section, the semantics of $C_1 + C_2$ and $\iter C{e}{e'}$ are not always defined given the partiality of the semiring $+$. The ways that $+$ can be used in programs depends on the particular semiring instance. We begin with guarded branching, \ie if-statements, which are encoded below.
\[
\iftf b{C_1}{C_2} \triangleq (\assume b\fatsemi C_1) + (\assume{\lnot b}\fatsemi C_2)
\]
Regardless of which semiring is used, guarded choice is always valid. Since the two branches are guarded by opposing tests, the weight of one branch will be multiplied by $\one$ and the other by $\zero$. Since $\zero$ is the identity of addition, then $\zero + x = x$ for any $x$, so the semantics of the sum of the two branches is equal to just the nonzero branch, which is defined. For example, suppose that $\de{b}(\sigma) = \one$, then we get:
\begin{align*}
  \de{\iftf b{C_1}{C_2}}(\sigma)
  &= \de{(\assume b\fatsemi C_1) + (\assume{\lnot b}\fatsemi C_2)}(\sigma)
  \\
  &= \de{b}(\sigma) \cdot \de{C_1}(\sigma) + \de{\lnot b}(\sigma) \cdot \de{C_2}(\sigma)
  \\
  &= \one\cdot \de{C_1}(\sigma) + \zero \cdot\de{C_2}(\sigma)
  \\
  &= \de{C_1}(\sigma)
\end{align*}
Since $\mathsf{Bool}$, $\mathsf{Nat}$, and $\mathsf{Tropical}$ are total semirings, unguarded choice is always valid in those execution models. In the probabilistic case, choice can be used as long as the sum of the weights of both branches is at most 1. One way to achieve this is to weight one branch by a probability $p\in[0,1]$ and the other branch by $1-p$, a biased coin-flip. We provide syntactic sugar for that operation:
\[
C_1 +_p C_2 \triangleq (\assume p\fatsemi C_1) + (\assume{1-p}\fatsemi C_2)
\]
We also provide syntactic sugar for iterating constructs. While loops use a test to determine whether iteration should continue, making them deterministic.
\[
\whl bC \triangleq \iter Cb{\lnot b}
\]
The Kleene star $C^\star$ is defined for interpretations based on total semirings only; it iterates $C$ nondeterministically many times.\footnote{\label{foot:star}
In nondeterministic languages, $\whl bC\equiv (\assume b\fatsemi C)^\star\fatsemi \assume{\lnot b}$, however this encoding does not work in general since $(\assume b\fatsemi C)^\star$ is not a well-defined program when using a partial semiring (\eg \Cref{ex:detsemi,ex:probsemi}).}
\[
C^\star \triangleq \iter C\one\one
\]
Finally, the probabilistic iterator $C^{\langle p\rangle}$ continues to execute with probability $p$ and exits with probability $1-p$.
\[
C^{\langle p\rangle} \triangleq \iter Cp{1-p}
\]
This behavior can be replicated using a while loop and auxiliary variables, but adding state complicates reasoning about the programs and precludes, \eg devising equational theories over uninterpreted atomic commands \cite{r_o_zowski2023probabilistic}. This construct---which was not included in previous Outcome Logic work---is therefore advantageous.

In \Aref{app:totality}, we prove that programs constructed using appropriate syntax have total semantics. For the remainder of the paper, we assume that programs are constructed in this way, and are thus always well-defined.

\section{Outcome Logic}\label{sec:ol}

In this section, we define Outcome Logic, provide extensional definitions for assertions used in the pre- and postconditions of Outcome Triples, and present a relatively complete proof system.

\subsection{Outcome Assertions}
\label{sec:assertions}

Outcome assertions are the basis for expressing pre- and postconditions in Outcome Logic. Unlike pre- and postconditions of Hoare Logic---which describe \emph{individual states}---outcome assertions expose the weights from \Cref{sec:algebra} to enable reasoning about branching and the weights of reachable outcomes.
We represent these assertions semantically; outcome assertions $\varphi,\psi \in \bb{2}^{\mathcal W_{\mathcal A}(\Sigma)}$ are the sets of elements corresponding to their true assignments. For any $m\in\mathcal W_{\mathcal A}(\Sigma)$, we write $m\vDash \varphi$ ($m$ satisfies $\varphi$) to mean that $m\in\varphi$.

The use of semantic assertions makes our approach extensional. We will therefore show that the Outcome Logic proof system is sufficient for analyzing programs structurally, but it cannot be used to decide entailments between the pre- and postconditions themselves. No program logic is truly complete, as analyzing loops inevitably reduces to the (undecidable) halting problem \cite{apt1981ten,cook1978soundness}---it is well known that the ability to express intermediate assertions and loop invariants means that the assertion language must at least contain Peano arithmetic \cite{lipton1977necessary}. As a result, many modern developments such as Separation Logic \cite{semanticsep,yang}, Incorrectness Logic \cite{il}, Iris \cite{iris1,iris}, probabilistic Hoare-style logics \cite{ellora,kaminski}, and others \cite{hyperhoare,raad2024nontermination,cousot2012abstract,ascari2025revealing} use semantic assertions.

We now define useful notation for assertions, which is also repeated in \Cref{fig:astsem}. For example $\top$ (always true) is the set of all weighted collections, $\bot$ (always false) is the empty set, and logical negation is the complement.
\[
\top\triangleq \mathcal W_{\mathcal A}(\Sigma)
\qquad\quad
\bot\triangleq \emptyset
\qquad\quad
\lnot\varphi \triangleq \mathcal W_{\mathcal A}(\Sigma) \setminus \varphi
\]
Conjunction, disjunction, and implication are defined as usual:
\[
\varphi\lor\psi \triangleq \varphi\cup\psi
\qquad\qquad
\varphi\land\psi \triangleq \varphi\cap\psi
\qquad\qquad
\varphi\Rightarrow\psi \triangleq (\mathcal W_{\mathcal A}(\Sigma)\setminus \varphi) \cup \psi
\]
Given a predicate $\phi\colon T \to \bb{2}^{\mathcal W_{\mathcal A}(\Sigma)}$ on some (possibly infinite) set $T$, existential quantification over $T$ is the union of $\phi(t)$ for all $t\in T$, meaning it is true iff there is some $t\in T$ that makes $\phi(t)$ true.
\[
\exists x:T.~\phi(x) \triangleq \bigcup_{t\in T}\phi(t)
\]
Next, we define notation for assertions based on the operations of the semiring $\mathcal A = \langle U,+,\cdot,\zero,\one\rangle$.
The \emph{outcome conjunction} $\varphi\oplus\psi$ asserts that the collection of outcomes $m$ can be split into two parts $m = m_1+m_2$ such that $\varphi$ holds in $m_1$ and $\psi$ holds in $m_2$.
For example, in the nondeterministic interpretation, we can view $m_1$ and $m_2$ as sets (not necessarily disjoint), such that $m = m_1 \cup m_2$, so $\varphi$ and $\psi$ each describe subsets of the reachable states.
We define outcome conjunctions formally for a predicate $\phi \colon T\to\bb{2}^{\mathcal W_{\mathcal A}(\Sigma)}$ over some (possibly infinite) set $T$.
\[
\bigoplus_{x\in T}\phi(x) \triangleq \left\{ \sum_{t\in T} m_t ~\Big|~ \forall t\in T. ~m_t \in \phi(t) \right\}
\]
That is, $m \in \bigoplus_{x\in T} \phi(x)$ if $m$ is the sum of a sequence of $m_t$ elements for each $t\in T$ such that each $m_t \in \phi(t)$. The binary $\oplus$ operator can be defined in terms of the one defined above.
\[
\varphi \oplus\psi \triangleq \smashoperator{\bigoplus_{i\in\{1,2\}}}\phi(i)
\qquad\text{where}\qquad
\phi(1) = \varphi
\qquad\text{and}\qquad
\phi(2) = \psi
\]
The weighting operations $\varphi\odot u$ and $u\odot \varphi$, inspired by \citet{batz2022weighted}, scale the outcome $\varphi$ by a literal weight $u\in U$. Note that there are two variants of this assertion for scaling on the left and right, since multiplication is not necessarily commutative (see \Cref{ex:langsemi}).
\[
\varphi\odot u \triangleq \{ m\cdot u \mid m\in\varphi \}
\qquad\qquad
u\odot \varphi \triangleq \{ u\cdot m \mid m\in\varphi \}
\]
Finally, given a semantic assertion on states $P \subseteq \Sigma$, we can lift $P$ to be an outcome assertion, $\wg Pu$, meaning that $P$ covers all the reachable states ($\supp(m) \subseteq P$) and the cumulative weight is $u$.
\[
\wg Pu\triangleq \Big\{ m\in \mathcal W_{\mathcal A}(\Sigma) ~\Big|~ |m| = u, \supp(m)\subseteq P \Big\}
\]
When the weight is $\one$, we will write $\sure P$ instead of $\wg P\one$. These assertions can interact with weighting assertions, for instance $u\odot \wg Pv \Rightarrow \wg P{u\cdot v}$ and $\wg Pv \odot u \Rightarrow \wg P{v\cdot u}$.
We also permit the use of tests $b$ as assertions, for instance:
\[
\sure{P\land b} = \left\{ m \in \mathcal W(\Sigma) ~\middle|~ |m| = \one, \forall \sigma\in\supp(m).~ \sigma\in P \land \de{b}_{\mathsf{Test}}(\sigma) = \one \right\}
\]
There is a close connection between the $\oplus$ of outcome assertions and the choice operator $C_1+C_2$ for programs. If $P$ is an assertion describing the outcome of $C_1$ and $Q$ describes the outcome of $C_2$, then $\sure P\oplus \sure Q$ describes the outcome of $C_1 + C_2$ by stating that both $P$ and $Q$ are reachable outcomes via a non-vacuous program trace. This is more expressive than using the disjunction $\sure P\vee \sure Q$ or $\sure{P\vee Q}$, since the disjunction cannot not guarantee that \emph{both} $P$ and $Q$ are reachable. Suppose $P$ describes a desirable program outcome whereas $Q$ describes an erroneous one; then $\sure{P}\oplus \sure{Q}$ tells us that the program has a bug (it can reach an error state) whereas neither $\sure{P}\vee \sure{Q}$ nor $\sure{P\vee Q}$ is not strong enough to make this determination \cite{outcome}.

Similar to the syntactic sugar for probabilistic programs in \Cref{sec:sugar}, we define:
\[
\varphi \oplus_p \psi \triangleq (p \odot \varphi) \oplus ((1-p) \odot\psi)
\]
If $\varphi$ and $\psi$ are the results of running $C_1$ and $C_2$, then $\varphi \oplus_p \psi$---meaning that $\varphi$ occurs with probability $p$ and $\psi$ occurs with probability $1-p$---is the result of running $C_1 +_p C_2$.

\begin{figure}[t]

\begin{align*}
\top &\triangleq \mathcal W_{\mathcal A}(\Sigma)
&
\bot &\triangleq \emptyset
&
\lnot \varphi &\triangleq \top \setminus \varphi
\\
\varphi \land \psi &\triangleq \varphi \cap \psi
&
\varphi \vee \psi &\triangleq \varphi \cup \psi
&
\varphi\Rightarrow\psi &\triangleq \lnot\varphi \cup \psi
\\
\exists x:T.\phi(x) &\triangleq \bigcup_{t\in T} \phi(t)
&
\varphi \odot u &\triangleq \{ m\cdot u \mid m\in\varphi \}
&
u\odot \varphi &\triangleq \{ u\cdot m \mid m\in\varphi \}
\end{align*}
\begin{align*}
\bigoplus_{x\in T}\phi(x) & \triangleq \left\{ \sum_{t\in T} m_t ~\Big|~ \forall t\in T. ~m_t \in \phi(t) \right\}
\\
\wg Pu &\triangleq \Big\{ m\in \mathcal W_{\mathcal A}(\Sigma) ~\Big|~ |m| = u, \supp(m)\subseteq P \Big\}
\end{align*}
\caption{Outcome assertion semantics, given a partial semiring $\mathcal A = \langle U,+,\cdot,\zero,\one\rangle$ where $u\in U$, $\phi\colon T \to \bb{2}^{\mathcal W_{\mathcal A}(\Sigma)}$, and $P \in \bb{2}^\Sigma$.}
\label{fig:astsem}
\end{figure}

\subsection{Outcome Triples}

Inspired by Hoare Logic, Outcome Triples $\triple\varphi{C}\psi$ specify program behavior in terms of pre- and postconditions \cite{outcome}. The difference is that Outcome Logic describes weighted collections of states as opposed to Hoare Logic, which can only describe individual states. We write $\vDash\triple\varphi{C}\psi$ to mean that a triple is semantically valid, as defined below.
\begin{definition}[Outcome Triples]
\label{def:triples}
Given $\langle \mathcal A, \Sigma, \mathsf{Act}, \mathsf{Test}, \de{\cdot}_{\mathsf{Act}}\rangle$, the semantics of outcome triples is defined as follows:
\[
\vDash\triple\varphi{C}\psi
\quad\text{iff}\quad
\forall m\in\mathcal{W}_{\mathcal A} (\Sigma).~ m\vDash\varphi\implies \dem Cm\vDash\psi
\]
\end{definition}
Informally, $\triple{\varphi}C\psi$ is valid if running the program $C$ on a weighted collection of states satisfying $\varphi$ results in a collection satisfying $\psi$.
Using outcome assertions to describe these collections of states in the pre- and postconditions means that Outcome Logic can express many types of properties including reachability ($\sure{P}\oplus \sure{Q}$), probabilities ($\varphi \oplus_p \psi$), and nontermination (the lack of outcomes, $\wg\tru\zero$). In \Cref{sec:connections}, we will see how Outcome Logic can encode several familiar program logics.

\subsection{Inference Rules}\label{sec:proof}

We now describe the Outcome Logic rules of inference, which are shown in \Cref{fig:rules,fig:struct-rules}. The rules are split into three categories.

\heading{Sequential Commands.} The rules for sequential (non-looping) commands mostly resemble those of Hoare Logic. The \ruleref{Skip} rule stipulates that the precondition is preserved after running a no-op. \ruleref{Seq} derives a specification for a sequential composition from two sub-derivations for each command. Similarly, \ruleref{Plus} joins the derivations of two program branches using an outcome conjunction. 

\ruleref{Assume} has a side condition that $\varphi\vDash e=u$, where $u\in U$ is a semiring element. Informally, this means that the precondition entails that the expression $e$ is some concrete weight $u$. More formally, it is defined as follows:
\[
\varphi\vDash e=u
\qquad\text{iff}\qquad
\forall m\in\varphi.\quad 
\forall \sigma\in\supp(m).\quad
\de{e}(\sigma) = u
\]
If $e$ is a weight literal $u$, then $\varphi\vDash u=u$ vacuously holds for any $\varphi$, so the rule can be simplified to $\vdash \triple{\varphi}{\assume u}{\varphi\odot{u}}$. But if it is a test $b$, then $\varphi$ must contain enough information to conclude that $b$ is true or false. Additional rules to decide $\varphi\vDash e=u$ are given in \Aref{app:soundcomplete}.

\heading{Iteration} The \ruleref{Iter} rule uses two families of predicates: $\varphi_n$ represents the result of $n$ iterations of $\assume{e}\fatsemi C$ and $\psi_n$ is the result of iterating $n$ times and then weighting the result by $e'$, so $\bigoplus_{n\in\mathbb N}\psi_n$ represents all the aggregated terminating traces. This is captured by the assertion $\psi_\infty$, which must have the following property.
\begin{definition}[Converging Assertions]
A family $(\psi_n)_{n\in\mathbb{N}}$ \emph{converges} to $\psi_\infty$ (written $\conv{\psi}$)
iff for any collection $(m_n)_{n\in\mathbb N}$, if $m_n\vDash\psi_n$ for each $n\in\mathbb N$, then $\sum_{n\in\mathbb N} m_n \vDash\psi_\infty$.
\end{definition}

\begin{figure}
\begin{mathpar}
\ruledef{Skip}{\;}{\triple{\varphi}{\skp}{\varphi}}

\ruledef{Seq}{
  \triple{\varphi}{C_1}{\vartheta}
  \\
  \triple{\vartheta}{C_2}{\psi}
}{
  \triple{\varphi}{C_1\fatsemi C_2}{\psi}
}
\\
\ruledef{Plus}{
  \triple{\varphi}{C_1}{\psi_1}
  \\
  \triple{\varphi}{C_2}{\psi_2}
}{
  \triple{\varphi}{C_1+ C_2}{\psi_1 \oplus\psi_2}
}

\inferrule{\varphi \vDash e = u}{\triple{\varphi}{\assume e}{\varphi\odot u}}{\rulename{Assume}}

\inferrule{
  \conv{\psi}
  \quad
  \forall n\in\mathbb N.
  \quad
  \triple{\varphi_n}{\assume{e}\fatsemi C}{\varphi_{n+1}}
  \quad
  \triple{\varphi_n}{\assume{e'}}{\psi_n}
}{
  \triple{\varphi_0}{C^{\langle e, e'\rangle}}{\psi_\infty}
}{\rulename{Iter}}
\end{mathpar}
\caption{Inference rules for program commands}
\label{fig:rules}
\end{figure}

\heading{Structural Rules.}
We also give rules that are not dependent on the program command in \Cref{fig:struct-rules}. This includes rules for trivial preconditions (\ruleref{False}) and postconditions (\ruleref{True}). The \ruleref{Scale} rule states that we may multiply the pre- and postconditions by a weight to obtain a new valid triple. 

Subderivations can be combined with logical connectives using the \ruleref{Disj}, \ruleref{Conj}, and \ruleref{Choice}) rules.
Existential quantifiers are introduced using \ruleref{Exists}. Finally, the rule of \ruleref{Consequence} can be used to strengthen preconditions and weaken postconditions in the style of Hoare Logic. These implications are semantic ones: $\varphi' \Rightarrow \varphi$ iff $\varphi'\subseteq\varphi$. We do not explore the proof theory for outcome assertions, although prior work in this area exists as outcome conjunctions are similar to the separating conjunction from Bunched Implications \cite{bi,simon}.

\begin{figure}[t]
\begin{mathpar}
\inferrule{\;}{\triple{\bot}C\varphi}{\rulename{False}}

\inferrule{\;}{\triple{\varphi}C\top}{\rulename{True}}

\inferrule{\triple{\varphi}C{\psi}}{\triple{u\odot \varphi}C{u\odot \psi}}{\rulename{Scale}}

\inferrule{
  \triple{\varphi_1}C{\psi_1}
  \\
  \triple{\varphi_2}C{\psi_2}
}{
  \triple{\varphi_1 \vee \varphi_2}C{\psi_1 \vee \psi_2}
}{\rulename{Disj}}

\inferrule{
  \triple{\varphi_1}C{\psi_1}
  \\
  \triple{\varphi_2}C{\psi_2}
}{
  \triple{\varphi_1 \land \varphi_2}C{\psi_1 \land \psi_2}
}{\rulename{Conj}}

\inferrule{
  \forall t\in T.\quad \triple{\phi(t)}C{\phi'(t)}
}{
  \triple{\textstyle\bigoplus_{x\in T} \phi(x)}C{\textstyle\bigoplus_{x\in T}\phi'(x)}
}{\rulename{Choice}}

\inferrule
  {\forall t\in T. \quad\triple{\phi(t)}C{\phi'(t)}}
  {\triple{\exists x:T.~\phi(x)}C{\exists x:T.~\phi'(x)}}
  {\rulename{Exists}}

\hypertarget{rule:Cons}{
\inferrule{
  \varphi' \Rightarrow \varphi
  \\
  \triple{\varphi}C\psi
  \\
  \psi \Rightarrow \psi'
}{
  \triple{\varphi'}C{\psi'}
}{\rulename{Consequence}}}
\end{mathpar}
\caption{Structural rules.}
\label{fig:struct-rules}
\end{figure}

\subsection{Soundness and Relative Completeness}
\label{sec:soundness}
Soundness of the Outcome Logic proof system means that any derivable triple (using the inference rules in \Cref{fig:rules} and axioms about atomic actions) is semantically valid according to \Cref{def:triples}. We write $\Gamma\vdash\triple\varphi{C}\psi$ to mean that the triple $\triple\varphi{C}\psi$ is derivable given a collection of axioms $\triple{\varphi}{a}{\psi} \in \Gamma$. Let $\Omega$ consist of all triples $\triple{\varphi}a\psi$ such that $a\in\mathsf{Act}$, and $\vDash\triple{\varphi}a\psi$ (all the true statements about atomic actions).
We also presume that the program $C$ is well-formed as described in \Cref{sec:sugar}. The soundness theorem is stated formally below.
\begin{restatable}[Soundness]{theorem}{thmsoundness}\label{thm:soundness}
\[
\Omega\vdash\triple\varphi{C}\psi
\qquad\implies\qquad
\vDash\triple\varphi{C}\psi
\]
\end{restatable}
\noindent The full proof is shown in \Aref{app:soundcomplete} and proceeds by induction on the structure of the derivation $\Omega\vdash\triple\varphi{C}\psi$, with cases in which each rule is the last inference. Most of the cases are straightforward, but the following lemma is needed to justify the soundness of the \ruleref{Iter} case, where $C^0 =\skp$ and $C^{n+1} = C^n\fatsemi C$.
\begin{restatable}{lemma}{lemwhilesem}\label{lem:whilesem}
The following equation holds:
\[
\de{\iter C{e}{e'}}(\sigma) = \sum_{n\in\mathbb{N}} \de{(\assume{e}\fatsemi C)^n \fatsemi \assume{e'}}(\sigma)
\]
\end{restatable}
\noindent Completeness---the converse of soundness---tells us that our inference rules are sufficient to deduce any true statement about a program.
As is typical, Outcome Logic is \emph{relatively} complete, meaning that proving any valid triple can be reduced to implications $\varphi\Rightarrow\psi$ in the assertion language. For OL instances involving state (and Hoare Logic), those implications are undecidable since they must, at the very least, encode Peano arithmetic \cite{cook1978soundness,apt1981ten,lipton1977necessary}.

The first step is to show that given any program $C$ and precondition $\varphi$, we can derive the triple $\triple\varphi{C}\psi$, where $\psi$ is the \emph{strongest postcondition} \cite{dijkstra1990strongest}, \ie the strongest assertion making that triple true. As defined below, $\psi$ is exactly the set resulting from evaluating $C$ on each $m\in\varphi$. The proceeding lemma shows that the triple with the strongest postcondition is derivable.

\begin{definition}[Strongest Postcondition]\label{def:post}
\[
\spost(C, \varphi) \triangleq \{ \dem Cm \mid m\in \varphi \}
\]
\end{definition}
\begin{restatable}[Derivability of the Strongest Postcondition]{lemma}{lemcompleteness}\label{lem:completeness}
\[
\Omega\vdash\triple\varphi{C}{\spost(C,\varphi)}
\]
\end{restatable}

The proof is by induction on the structure of the program, and is shown in its entirety in \Aref{app:soundcomplete}. The cases for $\skp$ and $C_1\fatsemi C_2$ are straightforward, but the other cases are more challenging and involve existential quantification. To give an intuition as to why existentials are needed, let us examine an example involving branching. We use a concrete instance of Outcome Logic with variable assignment (formalized in \Cref{sec:state}).

Consider the program $\skp + (\var x\coloneqq \var x+1)$ and the precondition $\sure{\var x \ge 0}$. It is tempting to say that $\spost$ is obtained compositionally by joining the $\spost$ of the two branches using $\oplus$:
\begin{align*}
\spost(\skp + (\var x \coloneqq \var x+1), ~ \sure{\var x \ge 0}) 
&= \spost(\skp, \sure{\var x\ge 0}) \oplus \spost(\var x:=\var x+1, ~\sure{\var x\ge 0}) \\
&= \sure{\var x\ge0} \oplus \sure{\var x\ge 1}
\end{align*}
However, that is incorrect. While $\sure{\var x\ge0} \oplus \sure{\var x\ge 1}$ is \emph{a} valid postcondition, it is not the strongest one because it does not account for the relationship between the values of $\var x$ in the two branches; if $x=n$ in the first branch, then it must be $n+1$ in the second branch. A second attempt could use existential quantification to dictate that relationship.  
\[
\exists n:\mathbb N. ~ \sure{\var x = n}\oplus\sure{\var x = n+1}
\]
Unfortunately, that is also incorrect; it does not account for the fact that that precondition $\sure{\var x \ge 0}$ may be satisfied by a \emph{set} of states in which $\var x$ has many different values---the existential quantifier requires that $\var x$ takes on a single value in all the initial outcomes. The solution is to quantify over the collections $m\in\varphi$ satisfying the precondition, and then to take the $\spost$ of $\ind m = \{m\}$.
\[
\spost(C_1 +C_2, \varphi) = \exists m : \varphi. ~\spost(C_1, \ind m) \oplus \spost(C_2, \ind m)
\]
While it may seem unwieldy that the strongest post is hard to characterize even in this seemingly innocuous example, the same problem arises in logics for probabilistic \cite{ellora,hartog} and hyperproperty \cite{hyperhoare} reasoning, both of which are encodable in OL. Although the \emph{strongest} postcondition is quite complicated, something weaker suffices in most cases. We will later see how rules for those simpler cases are derived (\Cref{sec:derived}) and used (\Cref{sec:examples,sec:graphs}).
The main relative completeness result is now a straightforward corollary of \Cref{lem:completeness} using the rule of \ruleref{Consequence}, since any valid postcondition is implied by the strongest one.

\begin{theorem}[Relative Completeness]\label{thm:completeness}
\[
\vDash \triple\varphi{C}\psi
\qquad\implies\qquad
\Omega\vdash\triple\varphi{C}{\psi}
\]
\end{theorem}
\begin{proof}
We first establish that $\spost(C, \varphi) \Rightarrow \psi$. Suppose that $m\in \spost(C, \varphi)$. That means that there must be some $m' \in \varphi$ such that $m = \dem C{m'}$. Using $\vDash\triple\varphi{C}\psi$, we get that $m\vDash\psi$. Now, we complete the derivation as follows:
\[
\inferrule*[right=\ruleref{Consequence}]{
  \inferrule*[right=\Cref{lem:completeness}]{
    \Omega
  }{
    \triple\varphi{C}{\spost(C,\varphi)}
  }
  \\
  \spost(C,\varphi) \Rightarrow \psi
}{
  \triple\varphi{C}\psi
}
\]
\end{proof}

\subsection{Derived Rules for Syntactic Sugar}
\label{sec:derived}

Recall from \Cref{sec:sugar} that if statements and while loops are encoded using the choice and iteration constructs. We now derive convenient inference rules for if and while. The full derivations of these rules are shown in \Aref{app:derived}.

If statements are defined as $(\assume b \fatsemi C_1) + (\assume \lnot b\fatsemi C_2)$.
Reasoning about them generally requires the precondition to be separated into two parts, $\varphi_1$ and $\varphi_2$, representing the collections of states in which $b$ is true and false, respectively. This may require---\eg in the probabilistic case---that $\varphi_1$ and $\varphi_2$ quantify the weight (likelihood) of the guard.

If it is possible to separate the precondition in that way, then $\varphi_1$ and $\varphi_2$ act as the preconditions for $C_1$ and $C_2$, respectively, and the overall postcondition is an outcome conjunction of the results of the two branches.
\[
\inferrule{
  \varphi_1\vDash b
  \\
  \triple{\varphi_1}{C_1}{\psi_1}
  \\
  \varphi_2 \vDash\lnot b
  \\
    \triple{\varphi_2}{C_2}{\psi_2}
}{
  \triple{\varphi_1 \oplus\varphi_2}{\iftf b{C_1}{C_2}}{\psi_2\oplus\psi_2}
}{\rulename{If}}
\]
From \ruleref{If}, we can also derive one-sided rules, which apply when one of the branches is certainly taken.
\[
\ruledef{If1}{
  \varphi\vDash b
  \\
  \triple{\varphi}{C_1}{\psi}
}{
  \triple{\varphi}{\iftf b{C_1}{C_2}}{\psi}
}
\qquad\qquad
\ruledef{If2}{
  \varphi\vDash \lnot b
  \\
  \triple{\varphi}{C_2}{\psi}
}{
  \triple{\varphi}{\iftf b{C_1}{C_2}}{\psi}
}
\]
The rule for while loops is slightly simplified compared to \ruleref{Iter}, as it only generates a proof obligation for a single triple instead of two. There are still two families of assertions, but $\varphi_n$ now represents the portion of the program configuration where the guard $b$ is true, and $\psi_n$ represents the portion where it is false. So, on each iteration, $\varphi_n$ continues to evaluate and $\psi_n$ exits; the final postcondition $\psi_\infty$ is an aggregation of all the terminating traces. Note that in all the rules below, the universal quantifiers apply to everything to their right.
\[
\inferrule{
  \conv{\psi}
  \\
  \forall n\in\mathbb N.
  \quad
  \triple{\varphi_n}{C}{\varphi_{n+1} \oplus \psi_{n+1}}
  \\
  \varphi_n \vDash b
  \\
  \psi_n\vDash \lnot b
}{
  \triple{\varphi_0 \oplus \psi_0}{\whl bC}{\psi_\infty}
}
{\rulename{While}}
\]
This \ruleref{While} rule is similar to those found in probabilistic Hoare Logics \cite{den_hartog1999verifying,ellora} and Hyper Hoare Logic \cite{hyperhoare}.

Loop variants are an alternative way to reason about loops that terminate after a finite number of steps. They were first studied in the context of total Hoare Logic \cite{manna1974axiomatic}, but are also used in other logics that require termination guarantees such as Reverse Hoare Logic \cite{reversehoare}, Incorrectness Logic \cite{il}, and Lisbon Logic \cite{ilalgebra,raad2024nontermination,ascari2025revealing}.\footnote{
Outcome Logic guarantees the \emph{existence} of terminating traces, but it is not a total correctness logic in that it cannot ensure that \emph{all} traces terminate. This stems from the program semantics, which collects the finite traces, but does not preclude additional nonterminating ones. For example, $\de{\skp}(\sigma) = \de{\skp + \whl\tru\skp}(\sigma)$. The exception is the probabilistic interpretation, where \emph{almost sure termination} can be established by proving that the probability of terminating is 1. See \Cref{sec:limitations} for a more in depth discussion.
}

The rule uses a family of \emph{variants} $(\varphi_n)_{n\in\mathbb N}$ such that $\varphi_n$ implies that the loop guard $b$ is true for all $n>0$, and $\varphi_0$ implies that it is false, guaranteeing that the loop exits.
The inference rule is shown below, and states that starting at some $\varphi_n$, the execution will eventually count down to $\varphi_0$, at which point it terminates.
\[
\inferrule{
  \forall n\in\mathbb N.
  \\
  \varphi_0 \vDash \lnot b
  \\
  \varphi_{n+1}\vDash b
  \\
  \triple{\varphi_{n+1}}C{\varphi_n}
}{
  \triple{\exists n:\mathbb N.\varphi_n}{\whl bC}{\varphi_0}
}{\rulename{Variant}}
\]
Since the premise guarantees termination after precisely $n$ steps, it is easy to establish convergence---the postcondition only consists of a single trace.

\section{Adding Variables and State}
\label{sec:state}

We now develop a concrete Outcome Logic instance with variable assignment as atomic actions. Let $\mathsf{Var}$ be a countable set of variable names
 and $\mathsf{Val} = \mathbb Z$ be integer program values. Program stores $s\in\mathcal S \triangleq \mathsf{Var} \to \mathsf{Val}$ are maps from variables to values and we write $s[\var{x}\mapsto v]$ to denote the store obtained by extending $s\in\mathcal S$ such that $\var x$ has value $v$. Actions $a\in\mathsf{Act}$ are variable assignments $\var x \coloneqq E$, where $\var x \in \mathsf{Var}$ and $E\in\mathsf{Exp}$ can be a variable $\var x\in\mathsf{Var}$, constant $v \in \mathsf{Val}$, test $b$, or an arithmetic operation ($+$, $-$, $\times$).
 \begin{align*}
\mathsf{Act} \ni a &\Coloneqq \var x\coloneqq E \\
\mathsf{Exp} \ni E &\Coloneqq \var x\in\mathsf{Var} \mid v\in\mathsf{Val} \mid b \mid E_1 + E_2 \mid E_1 - E_2 \mid E_1 \times E_2
\end{align*}
In addition, we let the set of primitive tests $\mathsf{Test} = \bb{2}^{\mathcal S}$ be all subsets of the program states $\mathcal S$.
We will often write these tests symbolically, for example $x \ge 5$ represents the set $\{ s\in\mathcal S \mid s(x) \ge 5 \}$.
The interpretation of atomic actions is shown below, where the interpretation of expressions $\de{E}_\mathsf{Exp} \colon \mathcal S \to \mathsf{Val}$ is in \Aref{app:state}.
\[
\de{x\coloneqq E}_{\mathsf{Act}}(s) \triangleq \eta(s[x \mapsto \de{E}_{\mathsf{Exp}}(s)])
\]
We define substitutions in the standard way \cite{hyperhoare,ellora,ascari2025revealing,kaminski}, as follows:
\[
\varphi[E/x] \triangleq \{ m \in \mathcal W(\mathcal S) \mid (\lambda s. \eta(s[x \mapsto \de{E}_{\mathsf{Exp}}(s))])^\dagger(m) \in\varphi \}
\]
That is, $m\in\varphi[E/x]$ exactly when assigning $x$ to $E$ in $m$ satisfies $\varphi$. This behaves as expected in conjunction with symbolic tests, for example $\sure{x \ge 5}[y+1/x] = \sure{y+1\ge 5} = \sure{y\ge 4}$. It also distributes over most of the operations in \Cref{fig:astsem}, \eg $(\varphi\oplus\psi)[E/x] = \varphi[E/x] \oplus \psi[E/x]$.
Using substitution, we add an inference rule for assignment, mirroring the typical weakest-precondition style rule of \citet{hoarelogic} Logic.
\[
\inferrule{\;}{
  \triple{\varphi[E/x]}{x \coloneqq E}{\varphi}
}{\rulename{Assign}}
\]
When used in combination with the rule of \ruleref{Consequence}, \ruleref{Assign} can be used to derive any semantically valid triple about variable assignment.
Though it is not needed for completeness, we also include the rule of \ruleref{Constancy}, which allows us to add information about unmodified variables to a completed derivation. Here, $\mathsf{free}(P)$ is the set of \emph{free variables} that are used by the assertion $P$ (\eg $\mathsf{free}(x \ge 5) = \{x\}$) and $\mathsf{mod}(C)$ are the variables modified by $C$, both defined in \Aref{app:state}. In addition, the $\always P$ modality means that $P$ holds over the entire support of the weighting function, but does not specify the total weight. It is defined $\always P \triangleq \exists u:U.\ \wg{P}u$, and is discussed further in \Cref{sec:dynamichoare}. We use $\always P$ to guarantee that the rule of \ruleref{Constancy} applies regardless of whether or not $C$ terminates, branches, or alters the weights of traces.
\[
\inferrule{
  \triple\varphi{C}\psi
  \\
  \mathsf{free}(P) \cap \mathsf{mod}(C) = \emptyset
}{
  \triple{\varphi\land \always P}C{\psi\land \always P}
}{\rulename{Constancy}}
\]
In the Outcome Logic instance with variable assignment as the only atomic action, all triples can be derived without the axioms $\Omega$ from \Cref{thm:completeness}.
\begin{restatable}[Soundness and Completeness]{theorem}{soundcompleteassign}
\label{thm:statecomplete}
\[
\vDash\triple{\varphi}{C}{\psi}
\qquad\iff\qquad
\vdash\triple{\varphi}{C}{\psi}
\]
\end{restatable}

\section{Connections to Other Logics}
\label{sec:connections}

Outcome Logic, in its full generality, allows one to quantify the precise weights of each outcome, providing significant expressive power.
Nevertheless, many common program logics do not provide this much power, which can be advantageous as they offer simplified reasoning principles---for example, Hoare Logic's loop \ruleref{Invariant} rule is considerably simpler than the \ruleref{While} rule needed for general Outcome Logic (\Cref{sec:derived}).
In this section, we show the connections between Outcome Logic and several other logics by first showing that OL can capture the semantics of specifications in those logics, and then also deriving the proof rules of those logics using the OL proof system.

\subsection{Dynamic Logic, Hoare Logic, and Lisbon Logic}\label{sec:dynamichoare}

We will now devise an assertion syntax to show the connections between Outcome Logic and Hoare Logic. We take inspiration from modal logic and Dynamic Logic \cite{pratt1976semantical,harel2001dynamic}, using the modalities $\always$ and $\sometimes$ to express that assertions always or sometimes occur, respectively. We encode these modalities using the operations from \Cref{sec:assertions}, where $U$ is the set of semiring weights.
\[\arraycolsep=.5em
\begin{array}{rclcl}
\always P &\triangleq& \exists u:U. ~\wg{P}{u} 
&=& \{ m \mid \supp(m) \subseteq P \}
\\
\sometimes P &\triangleq& \exists u: (U \setminus \{\zero\}).~\wg Pu \oplus \top
&=& \{ m \mid \supp(m) \cap P \neq \emptyset \}
\end{array}
\]
We define $\always P$ to mean that $P$ occurs with some weight, so $m\vDash\always P$ exactly when $\supp(m)\subseteq P$. Dually, $\sometimes P$ requires that $P$ has nonzero weight and the $-\oplus\top$ permits additional elements to appear in the support. So, $m\vDash \sometimes P$ when $\sigma\in P$ for some $\sigma\in\supp(m)$. It is relatively easy to see that these two modalities are De Morgan duals, that is $\always P \Leftrightarrow \lnot\sometimes \lnot P$ and $\sometimes P \Leftrightarrow \lnot\always\lnot P$.
Defining these constructs as syntactic sugar allows us to reason about them with the inference rules in \Cref{sec:proof}, rather than new specialized ones.
For Boolean-valued semirings (\Cref{ex:detsemi,ex:powersetsemi}), we get the following:
\[
  \always P
  \quad=\quad
  \exists u\colon\{0,1\}.\ \wg Pu
  \quad=\quad
  \wg{P}{0} \vee \wg{P}{1}
\]
Only $\zero$, the empty collection, satisfies $\wg{P}{0}$, indicating that there are no outcomes and therefore the program diverged (let us call this assertion $\mathsf{div}$), and $\wg{P}{1}$ is equivalent to $\sure P$. So, $\always P = \sure P\vee \mathsf{div}$, meaning that either $P$ covers all the reachable outcomes, or the program diverged ($\always$ will be useful for expressing partial correctness). Similarly, in Boolean semirings, we have:
\[
  \sometimes P
  \quad=\quad
  \exists u:(\{0, 1\} \setminus \{0\}).\ \wg{P}u \oplus \top
  \quad=\quad
  \wg{P}1 \oplus \top
  \quad=\quad
  \sure P \oplus\top
\]
So, $\sometimes P = \sure P \oplus\top$, which means that $P$ is one of the possibly many outcomes. This is useful for incorrectness applications, as we saw in \Cref{sec:ov-incorrectness}.

Now, we are going to use these modalities to encode other program logics in Outcome Logic. We start with nondeterministic, partial correctness Hoare Logic, where the meaning of the triple $\hoare PCQ$ is that any state resulting from running the program $C$ on a state satisfying $P$ must satisfy $Q$.
There are many equivalent ways to formally define the semantics of Hoare Logic; we will use a characterization based on Dynamic Logic \cite{pratt1976semantical,harel2001dynamic}, which is inspired by modal logic in that it defines modalities similar to $\always$ and $\sometimes$.
\[
\dlbox CQ = \{ \sigma \mid \de{C}(\sigma) \subseteq Q \}
\qquad\qquad
\dldia CQ = \{ \sigma \mid  \de{C}(\sigma) \cap Q \neq \emptyset \}
\]
That is, $\dlbox CQ$ asserts that $Q$ must hold after running the program $C$ (if it terminates). In the predicate transformer literature, $\dlbox CQ$ is called the weakest liberal precondition \cite{Dijkstra76,GCL}. The dual modality $\dldia CQ$ states that $Q$ might hold after running $C$ (also called the weakest possible precondition \cite{wpp,ilalgebra}).
A Hoare Triple $\hoare PCQ$ is valid iff $P \subseteq \dlbox CQ$, so to show that Outcome Logic subsumes Hoare Logic, it suffices to prove that we can express $P \subseteq \dlbox CQ$. We do so using the $\always$ modality defined previously.


\begin{restatable}[Subsumption of Hoare Logic]{theorem}{thmhoare}\label{thm:hoare}
\[
\vDash\triple{\sure P}C{\always Q}
\qquad\text{iff}\qquad
P \subseteq \dlbox CQ
\qquad\text{iff}\qquad
\vDash\hoare PCQ
\]
\end{restatable}
\noindent While \citet{outcome} previously showed that Outcome Logic subsumes Hoare Logic, our characterization is not tied to nondeterminism; the triple $\triple{\sure P}C{\always Q}$ does not necessarily have to be interpreted in a nondeterministic way, but can rather be taken to mean that running $C$ in a state satisfying $P$ results in $Q$ covering all the terminating traces with some weight. 
We will shortly develop rules for reasoning about loops using invariants, which will be applicable to \emph{any} instance of Outcome Logic.

Given that the formula $P\subseteq \dlbox CQ$ gives rise to a meaningful program logic, it is natural to ask whether the same is true for $P \subseteq \dldia CQ$. In fact, this formula is colloquially known as Lisbon Logic, which was proposed by Derek Dreyer and Ralf Jung during a meeting in Lisbon as a possible foundation for incorrectness reasoning \cite{il,ilalgebra,outcome}. The semantics of Lisbon triples, denoted $\lisbon PCQ$, is that  for any start state satisfying $P$, there exists a state resulting from running $C$ that satisfies $Q$. Given that $Q$ only covers a subset of the outcomes, it is not typically suitable for correctness, however it is useful for incorrectness as some bugs only occur in some of the traces.
\begin{restatable}[Subsumption of Lisbon Logic]{theorem}{thmlisbon}\label{thm:lisbon}
\[
\vDash\triple{\sure P}C{\sometimes Q}
\qquad\text{iff}\qquad
P \subseteq \dldia CQ
\qquad\text{iff}\qquad
\vDash\lisbon PCQ
\]
\end{restatable}
We now present derived rules for simplified reasoning in our embeddings of Hoare and Lisbon Logic within Outcome Logic.
For the full derivations, refer to \Aref{app:derived}.

\heading{Sequencing}
\label{sec:seq}

The \ruleref{Seq} rule requires that the postcondition of the first command exactly matches the precondition of the next. This is at odds with our encodings of Hoare and Lisbon Logic, which have asymmetry between the modalities used in the pre- and postconditions. Still, sequencing is possible using derived rules.
\[
\inferrule{
  \triple {\sure P}{C_1}{\always Q}
  \quad
  \triple {\sure Q}{C_2}{\always R}
}{
  \triple {\sure P}{C_1\fatsemi C_2}{\always R}
}{\rulename{Seq (Hoare)}}
\qquad
\inferrule{
  \triple {\sure P}{C_1}{\sometimes Q}
  \quad
  \triple {\sure Q}{C_2}{\sometimes R}
}{
  \triple {\sure P}{C_1\fatsemi C_2}{\sometimes R}
}{\rulename{Seq (Lisbon)}}
\]
These rules rely on the fact that $\triple{\sure P}C{\always Q} \vdash\triple{\always P}C{\always Q}$ and $\triple{\sure P}C{\sometimes Q} \vdash\triple{\sometimes P}C{\sometimes Q}$, which is true since for any $m \in \always P$, it must be that each $\sigma\in\supp(m)$ satisfies $P$, meaning that $\de{C}(\sigma) \in \always Q$ and so $\dem Cm$ must also satisfy $\always Q$. The $\sometimes$ case is similar, also making use of the fact that $(\sure{R} \oplus\top)\oplus\top \Leftrightarrow \sure{R}\oplus \top$.
Lisbon Logic adds an additional requirement on the semiring; $\zero$ must be the \emph{unique} annihilator of multiplication ($u\cdot v = \zero$ iff $u=\zero$ or $v=\zero$), which ensures that a finite sequence of commands does not eventually cause a branch to have zero weight. \Cref{ex:detsemi,ex:powersetsemi,ex:multisetsemi,ex:probsemi,ex:tropicalsemi,ex:langsemi} all obey this property.

Note that since $\sure{P} \Rightarrow \always P$ and $\sure{P} \Rightarrow \sometimes P$, the rule of \ruleref{Consequence} also gives us $\triple{\always P}C{\always Q}\vdash \triple{\sure P}C{\always Q}$ and $\triple{\sometimes P}C{\sometimes Q} \vdash \triple{\sure P}C{\sometimes Q}$, so we could have equivalently defined Hoare and Lisbon Logic as $\vDash\triple{\always P}C{\always Q}$ and $\vDash\triple{\sometimes P}C{\sometimes Q}$, respectively. 
We prefer the asymmetric use of modalities, as it allows for specifications of the form $\triple{\sure P}C{\sure Q}$, which can be easily weakened for use in both Hoare and Lisbon Logic, as we will see in \Cref{sec:analysis}.

\heading{If Statements and While Loops}\label{sec:ifwhile}

The familiar rule for if statements in Hoare Logic is also derivable, and does not require any semantic entailments, instead using the fact that $\sure{P}\Rightarrow \always(P\land b) \oplus\always(P\land\lnot b)$.
\[
\inferrule{
  \triple{\sure{P\land b}}{C_1}{\always Q}
  \\
  \triple{\sure{P\land\lnot b}}{C_2}{\always Q}
}{
  \triple{\sure{P}}{\iftf b{C_1}{C_2}}{\always Q}
}{\rulename{If (Hoare)}}
\]
A similar rule is derivable for Lisbon Logic, although the derivation is a bit more complex due to the reachability guarantees provided by Lisbon Logic, and the fact that $\sure P\not\Rightarrow \sometimes(P\land b)\oplus\sometimes (P\land\lnot b)$. Instead, the derivation involves case analysis on whether $\sometimes(P\land b)$ or $\sometimes(P\land \lnot b)$ is true.
\[
\inferrule{
  \triple{\sure{P\land b}}{C_1}{\sometimes Q}
  \\
  \triple{\sure{P\land\lnot b}}{C_2}{\sometimes Q}
}{
  \triple{\sure{P}}{\iftf b{C_1}{C_2}}{\sometimes Q}
}{\rulename{If (Lisbon)}}
\]

\heading{Loop Invariants}
\label{sec:invariant}

Loop invariants are a popular analysis technique in partial correctness logics. The idea is to find an invariant $P$ that is preserved by the loop body and therefore must remain true when---\emph{and if}---the loop terminates. Because loop invariants are unable to guarantee termination, the Outcome Logic rule must indicate that the program may diverge. We achieve this using the $\always$ modality from \Cref{sec:dynamichoare}. The rule for Outcome Logic loop invariants is as follows:
\[
\inferrule{
  \triple{\sure{P\land b}}C{\always P}
}{
  \triple{\sure P}{\whl bC}{\always(P \land \lnot b)}
}
{\rulename{Invariant}}
\]
This rule states that if the program starts in a state described by $P$, which is also preserved by each execution of the loop, then $P\land\lnot b$ is true in every terminating state. If the program diverges and there are no reachable end states, then $\always(P\land\lnot b)$ is vacuously satisfied, just like in Hoare Logic.

\ruleref{Invariant} is derived using the \ruleref{While} rule with $\varphi_n = \always (P\land b)$ and $\psi_n = \always(P\land\lnot b)$.
To show $\conv\psi$, first note that $m_n\vDash \always(P\land\lnot b)$ simply means that $\supp(m_n) \subseteq (P\land\lnot b)$. Since this is true for all $n\in\mathbb N$, then all the reachable states satisfy $P\land\lnot b$.


It is well known that \ruleref{Skip}, \ruleref{Seq (Hoare)}, \ruleref{If (Hoare)}, \ruleref{Invariant}, \ruleref{Assign}, and \ruleref{Consequence} constitute a relatively complete proof system for Hoare Logic \cite{cook1978soundness,kozen2001completeness}. It follows that these rules are complete for deriving any Outcome Logic triples of the form $\triple{\sure P}C{\always Q}$, avoiding the more complex machinery of \Cref{lem:completeness}\footnote{\emph{N.B.}, this only includes the deterministic program constructs---if statements and while loops instead of $C_1+C_2$ and $\iter Ce{e'}$. The inclusion of a few more derived rules completes the proof system for nondeterministic programs.}.

\heading{Loop Variants}

Although loop variants are valid in any Outcome Logic instance, they require loops to be \emph{deterministic}---the loop executes for the same number of iterations regardless of any computational effects that occur in the body. Examples of such scenarios include for loops, where the number of iterations is fixed upfront.

We also present a more flexible loop variant rule geared towards Lisbon triples. In this case, we use the $\sometimes$ modality to only require that \emph{some} trace is moving towards termination. 
\[
\inferrule{
  \forall n\in\mathbb N.
  \\
  \sure{P_0} \vDash \lnot b
  \\
  \sure{P_{n+1}}\vDash b
  \\
  \triple{\sure{P_{n+1}}}C{\sometimes P_n}
}{
  \triple{\exists n:\mathbb N.\ \sure{P_n}}{\whl bC}{\sometimes P_0}
}{\rulename{Lisbon Variant}}
\]
In other words, \ruleref{Lisbon Variant} witnesses a single terminating trace. As such, it does not require the lockstep termination of all outcomes like \ruleref{Variant} does.

\subsection{Hyper Hoare Logic}
\label{sec:hyper}

Hyper Hoare Logic (HHL) is a generalized version of Hoare Logic that uses predicates over \emph{sets of states} as pre- and postconditions to enable reasoning about \emph{hyperproperties}---properties of multiple executions of a single program \cite{hyperhoare}. Hyperproperties are useful for expressing information flow security properties of programs, among others. In this section, we show how hyperproperty reasoning inspired by HHL can be done in Outcome Logic.

Although HHL triples have equivalent semantics to the powerset instance of Outcome Logic (\Cref{ex:powersetsemi}), the spirit of those triples are quite different.\footnote{
  Note that we also used the powerset instance of Outcome Logic to encode Hoare Logic and Lisbon Logic in \Cref{sec:dynamichoare}, however those examples used restricted pre- and postconditions to limit their expressivity. As was shown by \citet[Propositions 2 and 9]{hyperhoare}, HHL also subsumes Hoare and Lisbon Logic.
}
In Outcome Logic the precondition often describes a single state, whereas the postcondition specifies the nondeterministic outcomes that stem from that state. By contrast, in HHL the precondition may describe the relationship between several executions of a program. To achieve this, HHL uses assertions that quantify over states. We provide equivalent notation for this below:
\[
  \forall \langle\sigma\rangle.\varphi
   \triangleq \{ m \mid \forall s\in\supp(m).~ m \in \varphi[s/\sigma] \}
  \qquad
  \exists \langle \sigma\rangle.\varphi
   \triangleq \{ m \mid \exists s\in\supp(m).~ m \in \varphi[s/\sigma] \}
\]
These bound metavariables $\sigma$, referring to states, can then be referenced in hypertests $B$, using the syntax below, where $\mathop\asymp \in \{ =, \le, \ldots \}$ ranges over the usual comparators.
\[
  B\Coloneqq \tru \mid \fls \mid \sigma(x) \mid B_1 \land B_2 \mid B_1 \lor B_2 \mid \lnot B \mid B_1 \asymp B_2 \mid \cdots
\]
Unlike normal tests $b$, which can reference program variables, hypertests can only reference variables from particular executions $\sigma(x)$.
So, for example, the assertion $\forall \langle\sigma\rangle.\ \sigma(x) = 5$ means that the variable $x$ has the value 5 in every execution, similar to the assertion $\always(x=5)$ that we saw in \Cref{sec:dynamichoare}. However, the quantifiers introduced by HHL provide significant expressive power over the $\always$ and $\sometimes$ modalities, since they allow us to express the relationship between program variables in multiple executions. For example, we can define the following $\mathsf{low}$ predicate \cite{hyperhoare}, which states that the value of some variable $\ell$ is the same in all executions.
\[
   \mathsf{low}(\ell) \triangleq \forall\langle\sigma\rangle.\ \forall\langle\tau\rangle.\ \sigma(\ell) = \tau(\ell)
\]
We call this \emph{low}, since we will use it to indicate that a variable has low sensitivity from an information security point of view. This allows us to both prove and disprove \emph{noninterference}, a hyperproperty stating that high sensitivity information cannot flow into the low sensitivity program variables. For example, the program below on the left is secure; if two executions have the same initial values of $\ell$, then they will also have the same final values for $\ell$. On the other hand, the program on the right is insecure; information flows from $h$ (a high-sensitivity input) to $\ell$, so the final values of $\ell$ will differ in any pair of executions where the initial values of $h$ differ.
\[
  \triple{\mathsf{low}(\ell)}{\ell \coloneqq \ell+1}{\mathsf{low}(\ell)}
  \qquad\qquad
  \triple{\mathsf{low}(\ell) \land \lnot \mathsf{low}(h)}{\ell \coloneqq h+1}{\lnot\mathsf{low}(\ell)}
\]
We will not explore the full expressiveness of Hyper Hoare Logic here, as it is discussed extensively by \citet{hyperhoare}. Instead, we will show how some of the Hyper Hoare Logic rules can be derived using the Outcome Logic proof system. In particular, although the use of quantifiers over states may appear to complicate reasoning significantly, \citet[\S4]{hyperhoare} showed that this kind of reasoning can be achieved with simple syntactic rules. We recreate this result here by deriving the syntactic rules of HHL in Outcome Logic.

The first step is to fix a syntax for assertions. We use the one below, where the semantics of the assertions are the same as were defined notationally in \Cref{fig:astsem}, with a few caveats. First, existential qualification can now only range over program values (as defined in \Cref{sec:state}), so that $T\subseteq\mathsf{Val}$. Rather than use a predicate $\phi\colon T \to \bb{2}^{\mathcal W_{\mathcal A}(\Sigma)}$, we also now presume that the syntactic assertion $\varphi$ may reference the newly bound variable $v$. We have added universal quantification, which is defined analogously to existential quantification, but using an intersection instead of a union. The quantifiers over states $\forall\langle\sigma\rangle.\varphi$ and $\exists\langle\sigma\rangle.\varphi$ are as described above. Finally, hypertests $B$ only have a meaning if they are \emph{closed}, that is, they do not contain any unbound state metavariables (notationally, we use $\sigma$ and $\tau$). Closed hypertests can be evaluated to Booleans in the usual way, and open hypertests (containing unbound variables) are considered to be false.
\[
  \varphi \Coloneqq
    \varphi\land \psi
    \mid \varphi\lor\psi
    \mid \forall v\colon T.\ \varphi
    \mid \exists v\colon T.\ \varphi
    \mid \forall\langle \sigma\rangle.\ \varphi
    \mid \exists\langle \sigma\rangle.\ \varphi
    \mid B
\]
We define $\lnot$ inductively, for example $\lnot (\varphi\land\psi) \triangleq \lnot\varphi\lor \lnot\psi$ and $\lnot( \forall\langle\sigma\rangle.\varphi) \triangleq \exists\langle\sigma\rangle.\lnot\varphi$, the full details are in \Aref{app:hyperhoare}. In addition, we let $\varphi\Rightarrow\psi \triangleq \lnot\varphi\lor\psi$.

We will now give syntactic rules due to \citet{hyperhoare} for commands that interact with state. These rules are significant, since they define precisely how substitution and satisfaction of tests interact with the new quantifiers introduced in HHL.

\subsubsection*{Variable Assignment}
We do not need a new rule for variable assignment, but rather we will show how to syntactically transform the postcondition $\varphi$ to match the semantic substitutions that we use in the \ruleref{Assign} rule from \Cref{sec:state}. We first define an operation $E[\sigma]$ which transforms an expression $E$ into a hyper-expression by replacing all occurrences of variables $x$ with $\sigma(x)$. So for example, $(x + 2\times y)[\sigma] = \sigma(x) + 2\times \sigma(y)$. We write $\varphi[ E[\sigma] / \sigma(x)]$ to be a standard capture-avoiding substitution, syntactically replacing any occurrence of $\sigma(x)$ with $E[\sigma]$. Now, we define a transformation on assertions $\mathcal A_x^E[-]$, which substitutes the expression $E$ for the variable $x$.
\[\arraycolsep=1pt
  \def\arraystretch{1.2}
  \begin{array}{rclcrcl}
    \mathcal A^E_x[ &\varphi \land\psi &] \triangleq \mathcal A^E_x[\varphi] \land \mathcal A^E_x[\psi]
    &\qquad&
    \mathcal A^E_x[&\varphi \lor\psi &] \triangleq \mathcal A^E_x[\varphi] \lor \mathcal A^E_x[\psi]
    \\
    \mathcal A^E_x[& \forall v\colon T.\ \varphi &] \triangleq \forall v\colon T.\ \mathcal A^E_x[\varphi]
    &&
    \mathcal A^E_x[& \exists v\colon T.\ \varphi &] \triangleq \exists v\colon T.\ \mathcal A^E_x[\varphi]
    \\
    \mathcal A^E_x[& \forall\langle\sigma\rangle.\ \varphi &] \triangleq \forall\langle\sigma\rangle.\ \mathcal A^E_x[\varphi[E[\sigma]/ \sigma(x)]]
    &&
    \mathcal A^E_x[& \exists\langle\sigma\rangle.\ \varphi &] \triangleq \exists\langle\sigma\rangle.\ \mathcal A^E_x[\varphi[E[\sigma]/ \sigma(x)]]
    \\
    \mathcal A^E_x[& B &] \triangleq B
  \end{array}
\]
For most of the syntactic assertions, the $\mathcal A_x^E[-]$ operation is just propagated recursively. The interesting cases are the state quantifiers. Whenever a state $\sigma$ is quantified (either universally or existentially), then $E[\sigma]$ is syntactically substituted for $\sigma(x)$.
In \Aref{lem:syntactic-subst} we show that $\mathcal A_x^E[-]$ is equivalent to the semantic substitution defined in \Cref{sec:state}, meaning that the \ruleref{Assign} rule can be written as follows:
\[
  \inferrule{\;}{
    \triple{\mathcal A_x^E[\varphi]}{x \coloneqq E}{\varphi}
  }{\ruleref{Assign}}
\]
As an example, we will show how the \ruleref{Assign} rule can be used to derive the specification that we gave above for the insecure program $\ell \coloneqq h+1$. We begin by applying the syntactic substitution to our desired postcondition $\lnot\mathsf{low}(\ell)$, as follows:
\begin{align*}
\mathcal A_\ell^{h+1}[\lnot\mathsf{low}(\ell)]
  &= \mathcal A_\ell^{h+1}[\lnot ( \forall\langle\sigma\rangle.\ \forall\langle\tau\rangle.\ \sigma(\ell) = \tau(\ell))]
  \\
  &= \mathcal A_\ell^{h+1}[\exists\langle\sigma\rangle.\ \exists\langle\tau\rangle.\ \sigma(\ell) \neq \tau(\ell))]
  \\
  &= \exists\langle\sigma\rangle.\ \exists\langle\tau\rangle.\ \sigma(h) + 1 \neq \tau(h) + 1
  \\
  &= \lnot\mathsf{low}(h + 1)
\end{align*}
Now, it is relatively easy to see that $\mathsf{low}(\ell) \land\lnot \mathsf{low}(h) \Rightarrow \lnot \mathsf{low}(h)\Rightarrow \lnot \mathsf{low}(h+1)$, so a simple application of \ruleref{Assign} and the rule of \ruleref{Consequence} completes the proof.

\subsubsection*{Assume}
We now discuss a syntactic rule for assume statements. This will again involve defining some syntactic transformations on the postconditions. First, for any test $b$, we let $b[\sigma]$ be the hypertest obtained by replacing all occurrences of variables $x$ with $\sigma(x)$. For example, $(x = y + 1)[\sigma] = (\sigma(x) = \sigma(y) + 1)$. Now, we define $\Pi_b[-]$, which transforms an assertion $\varphi$ such that the test $b$ is true in all the quantified states.
\[\arraycolsep=1pt
  \def\arraystretch{1.2}
  \begin{array}{rclcrcl}
    \Pi_b[ &\varphi \land\psi &] \triangleq \Pi_b[\varphi] \land \Pi_b[\psi]
    &\qquad&
    \Pi_b[&\varphi \lor\psi &] \triangleq \Pi_b[\varphi] \lor \Pi_b[\psi]
    \\
    \Pi_b[& \forall v\colon T.\ \varphi &] \triangleq \forall v\colon T.\ \Pi_b[\varphi]
    &&
    \Pi_b[& \exists v\colon T.\ \varphi &] \triangleq \exists v\colon T.\ \Pi_b[\varphi]
    \\
    \Pi_b[& \forall\langle\sigma\rangle.\ \varphi &] \triangleq \forall\langle\sigma\rangle.\ b[\sigma] \Rightarrow \Pi_b[\varphi]
    &&
    \Pi_b[& \exists\langle\sigma\rangle.\ \varphi &] \triangleq \exists\langle\sigma\rangle.\ b[\sigma] \land \Pi_b[\varphi]
    \\
    \Pi_b[& B &] \triangleq B
  \end{array}
\]
Once again, all cases except for the state quantifiers simply recursively apply $\Pi_b[-]$. In the case of state quantifiers, we modify quantification to be only over all states where $b$ holds. As we show in \Aref{lem:syntactic-assume}, this means that $\Pi_b[\varphi] \Rightarrow (\varphi \land\always b) \oplus \always(\lnot b)$, so that $\varphi$ will still hold after executing $\assume b$.  This allows us to derive the following \ruleref{Assume HHL} rule:
\[
  \inferrule{\;}{
    \triple{\Pi_b[\varphi]}{\assume b}{\varphi}
  }{\rulename{Assume HHL}}  
\]

\subsubsection*{Nondeterministic Assignment}
The final syntactic rules that we will define pertain to nondeterministic assignment. We again define a syntactic transformation on postconditions $\mathcal H_x^S[-]$ for programs where $x$ is nondeterministically assigned a value from the set $S$.\footnote{This varies slightly from the definition of \citet{hyperhoare}, $\mathcal H_x[-]$, in which $x$ is assigned any value, not from a particular set. We can recover their operation as $\mathcal H_x[-] = \mathcal H_x^{\mathsf{Val}}[-]$.}
\[\arraycolsep=1pt
  \def\arraystretch{1.2}
  \begin{array}{rclcrcl}
    \mathcal H_x^S[ &\varphi \land\psi &] \triangleq \mathcal H_x^S[\varphi] \land \mathcal H_x^S[\psi]
    &\qquad&
    \mathcal H_x^S[&\varphi \lor\psi &] \triangleq \mathcal H_x^S[\varphi] \lor \mathcal H_x^S[\psi]
    \\
    \mathcal H_x^S[& \forall v\colon T.\ \varphi &] \triangleq \forall v\colon T.\ \mathcal H_x^S[\varphi]
    &&
    \mathcal H_x^S[& \exists v\colon T.\ \varphi &] \triangleq \exists v\colon T.\ \mathcal H_x^S[\varphi]
    \\
    \mathcal H_x^S[& \forall\langle\sigma\rangle.\ \varphi &] \triangleq \forall\langle\sigma\rangle.\ \forall v\colon S.\ \mathcal H_x^S[\varphi[v/\sigma(x)]
    &&
    \mathcal H_x^S[& \exists\langle\sigma\rangle.\ \varphi &] \triangleq \exists\langle\sigma\rangle.\ \exists v\colon s.\ \mathcal H_x^S[\varphi[v/\sigma(x)]
    \\
    \mathcal H_x^S[& B &] \triangleq B
  \end{array}
\]
The operation $\mathcal H^S_x[-]$ is similar to $\mathcal A_x^E[-]$, except that in the cases for state quantifiers, $\sigma(x)$ is replaced with any (or some) value in the set $S$. So, for example, if we wish to nondeterministically set $x$ to be either 1 or 2, and then assert that $\forall\langle \sigma\rangle.\ \exists\langle\tau\rangle.\ \sigma(x) \neq\tau(x)$, then we get:
\begin{align*}
  \mathcal H_x^{\{1, 2\}}[\forall\langle \sigma\rangle.\ \exists\langle\tau\rangle.\ \sigma(x) \neq\tau(x)]
  &= \forall\langle \sigma\rangle.\ \forall v\colon \{1, 2\}.\ \exists\langle\tau\rangle.\ \exists v'\colon \{1,2\}.\ v \neq v'
\end{align*}
The above assertion is trivially true. We can now use $\mathcal H_x^S[-]$ to derive inference rules for nondeterministic assignments. First, we will define syntactic sugar for an additional programming construct that nondeterministically assigns a variable $x$ to be any natural number. We define this using the Kleene star $C^\star$ defined in \Cref{sec:sugar}, which repeats $C$ a nondeterministic number of times. 
\[
  x \coloneqq \bigstar
  \qquad\triangleq\qquad
  x \coloneqq 0 \fatsemi (x \coloneqq x+1)^\star
\]
Now, by selecting the appropriate set $S$, we can derive the following havoc rules for nondeterministic assignments. The derivations of these rules are shown in \Aref{app:havoc}.
\[
  \inferrule{\;}{
    \triple{\mathcal H^{\{a, b\}}_x[\varphi]}{(x \coloneqq a) + (x \coloneqq b)}{\varphi}
  }{\rulename{Havoc-2}} 
  \qquad
  \inferrule{\;}{
    \triple{\mathcal H^{\mathbb{N}}_x[\varphi]}{x \coloneqq \bigstar}{\varphi}
  }{\rulename{Havoc-N}} 
\]

%
%
%

\section{Case Study: Reusing Proof Fragments}\label{sec:examples}

The following case study serves as a proof of concept for how Outcome Logic's unified reasoning principles can benefit large-scale program analysis. The efficiency of such systems relies on pre-computing procedure specifications, which can simply be inserted whenever those procedures are invoked rather than being recomputed at every call-site.
Existing analysis systems operate over homogenous effects. Moreover---when dealing with nondeterministic programs---they must also fix either a demonic interpretation (for correctness) or an angelic interpretation (for bug-finding).

But many procedures do not have effects---they do not branch into multiple outcomes and use only limited forms of looping where termination is easily established (\eg iterating over a data structure)---suggesting that specifications for such procedures can be reused across multiple types of programs (\eg nondeterministic or probabilistic) and specifications (\eg partial or total correctness). Indeed, this is the case for the program in \Cref{sec:division}. We then show how a single proof about that program can be reused in both a partial correctness specification (\Cref{sec:collatz}) and a probabilistic program (\Cref{sec:prob}). The full derivations are given in \Aref{app:examples}.

\subsection{Integer Division}
\label{sec:division}

In order to avoid undefined behavior related to division by zero, our expression syntax from \Cref{sec:state} does not include division. However, we can write a simple procedure to divide two natural numbers $\var a$ and $\var b$ using repeated subtraction.
\[
\code{Div} \triangleq \left\{
\begin{array}{l}
\var q  \coloneqq0 \fatsemi \var r \coloneqq \var a \fatsemi \phantom{x} \\
\whl{\var r \ge \var b}{\;} \\
\quad \var r \coloneqq \var r - \var b \fatsemi\phantom{x} \\
\quad \var q \coloneqq \var q+1
\end{array}
\right.
\]
At the end of the execution, $\var q$ holds the quotient and $\var r$ is the remainder.
Although the $\code{Div}$ program uses a while loop, it is quite easy to establish that it terminates. To do so, we use the \ruleref{Variant} rule with the family of variants $\varphi_n$ shown below.
\[
\varphi_n \triangleq \left\{
\begin{array}{lll}
\sure{\var q +n = \lfloor \var a\div \var b\rfloor \land \var r =(\var a\bmod \var b) + n \times \var b} & \text{if} & n \le \lfloor \var a\div \var b\rfloor \\
\bot & \text{if} & n > \lfloor \var a\div \var b\rfloor
\end{array}
\right.
\]
Executing the loop body in a state satisfying $\varphi_n$ results in a state satisfying $\varphi_{n-1}$. At the end, $\varphi_0$ stipulates that $\var q = \lfloor \var a\div\var b\rfloor$ and $\var r = \var a\bmod \var b$, which immediately implies that $\var r < \var b$, so the loop must exit. This allows us to give the following specification for the program.
\[
\triple{\sure{\var a \ge 0 \land \var b > 0}}{\code{Div}}{\sure{\var q = \lfloor \var a \div \var b \rfloor \land \var r = \var a \bmod \var b}}
\]
Note that the $\code{Div}$ program is deterministic; it does not use branching and we did not make any assumptions about which interpretation of choice is used. This will allow us to reuse the proof of $\code{Div}$ in programs with different kinds of effects in the remainder of the section.

\subsection{The Collatz Conjecture}
\label{sec:collatz}

Consider the function $f$ defined below.
\[
f(n) \triangleq \left\{
\begin{array}{lll}
n\div 2 & \text{if} & n \bmod 2 = 0 \\
3n+1 & \text{if} & n \bmod 2 = 1
\end{array}
\right.
\]
The Collatz Conjecture---an elusive open problem in the field of mathematics---postulates that for any positive $n$, repeated applications of $f$ will eventually yield the value 1. Let the \emph{stopping time} $S_n$ be the minimum number of applications of $f$ to $n$ that it takes to reach 1. For example, $S_1 = 0$, $S_2 = 1$, and $S_3 = 7$. When run in an initial state where $\var a = n$, the following program computes $S_n$, storing the result in $\var i$. Note that this program makes use of $\mathsf{DIV}$, defined previously.
\[
\code{Collatz} \triangleq \left\{
\begin{array}{l}
\var i \coloneqq 0 \fatsemi \phantom{x} \\
\whl{\var a \neq 1}{\;} \\
\quad \var b \coloneqq 2 \fatsemi
 \code{Div} \fatsemi\phantom{x}\\
\quad \iftf{\var r = 0}{\var a \coloneqq \var q}{\var a \coloneqq 3 \times \var a + 1} \fatsemi\phantom{x} \\
\quad \var i \coloneqq \var i+1
\end{array}
\right.
\]
Since some numbers may not have a finite stopping time---in which case the program will not terminate---this is a perfect candidate for a partial correctness proof.
Assuming that $\var a$ initially holds the value $n$, we can use a loop invariant stating that $\var a = f^{\var i}(n)$ on each iteration. \emph{If} the program terminates, then $\var a = f^{\var i}(n) = 1$, and so $S_n = \var i$. We capture this using the following triple, where the $\always$ modality indicates that the program may diverge.
\[
\triple{\sure{\var a = n \land n > 0}}{\code{Collatz}}{\always(\var i = S_n)}
\]

\subsection{Embedding Division in a Probabilistic Program}
\label{sec:prob}

The following program loops for a random number of iterations, deciding whether to continue by flipping a fair coin. It is interpreted using the $\mathsf{Prob}$ semiring from \Cref{ex:probsemi}.
\[
\var a \coloneqq 0 \fatsemi \var r \coloneqq 0 \fatsemi \left(\var a \coloneqq \var a +1\fatsemi \var b\coloneqq 2 \fatsemi \code{Div}\right)^{\langle \frac12\rangle}
\]
Suppose we want to know the probability that it terminates after an even or odd number of iterations.
The program makes use of $\mathsf{DIV}$ to divide the current iteration number $\var a$ by 2, therefore the remainder $\var r$ will indicate whether the program looped an even or odd number of times. We can analyze the program with the \ruleref{Iter} rule, using the following two families of assertions.
\[
\varphi_n \triangleq \wg{\var a = n \land \var r = \var a \bmod 2}{\frac1{2^n}}
\qquad\quad
\psi_n \triangleq \wg{\var a = n \land \var r = \var a \bmod 2}{\frac1{2^{n+1}}}
\]
According to \ruleref{Iter}, the final postcondition can be obtained by taking an outcome conjunction of all the $\psi_n$ for $n\in\mathbb N$. However, we do not care about the precise value of $\var a$, only whether $\var r$ is 0 or 1. The probability that $\var r = 0$ is $\frac12 +\frac18+\frac1{32}+\cdots$, a geometric series whose sum converges to $\frac23$. A similar calculation for the $\var r=1$ case gives us the following specification, indicating that the program terminates after an even or odd number of iterations with probability $\frac23$ and $\frac13$, respectively.
\[
\triple{\sure\tru}{
  \var a \coloneqq 0 \fatsemi \var r \coloneqq 0 \fatsemi \left(\var a \coloneqq \var a +1\fatsemi \var b\coloneqq 2 \fatsemi \code{Div}\right)^{\langle \frac12\rangle}
}{
  \sure{\var r = 0} \oplus_{\frac23} \sure{\var r = 1}
}
\]

\subsection{Implications to Program Analysis}
\label{sec:analysis}

Building on the proof reusability demonstrated in this case study, we now explore how Outcome Logic can be used as a unifying foundation for correctness and incorrectness static analysis. Although correctness and incorrectness can have many meanings, we follow the lead of \citet{il}, with the distinction coming down to demonic vs angelic nondeterminism. More precisely, a correctness specification covers \emph{all} traces whereas an incorrectness specification witnesses the \emph{existence} of a single faulty one.

Many real world static analysis systems operate in this way, such as Meta's Infer tool \cite{infer}, which was initially developed as a verification engine to prove memory safety in large codebases. It accordingly reports specifications as Hoare Triples $\hoare PCQ$, where $P$ specifies the resources that must be available in order for the program $C$ to execute safely in all traces \cite{biab,biabjacm}. However, it was later discovered that Infer was more effective as a bug finding tool, both because correctness analysis was sometimes computationally intractable and also because the codebases contained bugs \cite{distefano2019scaling}. 

The problem is that Hoare Logic is not sound as a logical foundation for bug finding---it admits false positives \cite{il}. More specifically, if an erroneous outcome occurs nondeterministically, then Hoare Logic has no way to witness a trace proving that the bug occurs. For example, if $Q_\ok$ represents the desired outcome and $Q_\er$ represents the erroneous one, then the Hoare Logic postcondition must be a disjunction of the two: $\hoare PC{Q_\ok\vee Q_\er}$, but given a specification of that form, it is possible that every execution falls into the $Q_\ok$ branch. In practice, postconditions become imprecise due to the use of abstraction \cite{biab}, and so Infer often fails to prove the absence of bugs, but that cannot be used as evidence that a bug exists \cite{isl}.

In response, the Infer team developed a new analysis called Pulse, which is based on angelic nondeterminism and therefore has a different logical foundation \cite{isl,realbugs}. As we saw in \Cref{sec:dynamichoare}, both forms of nondeterminism can be represented in Outcome Logic, by using different modalities in the postcondition\footnote{Although Pulse was initially based on Incorrectness Logic \cite{il}, it was observed by \citet{zilberstein2024outcome} and \citet{raad2024nontermination} that it can also be modeled using Lisbon Logic}.
\[
  \triple {\sure P}C{\always Q}
  \qquad\qquad
  \triple {\sure P}C{\sometimes Q}
\]
However, these two types of specifications are still incompatible; one cannot be used in a sub-derivation of the other. Crucially---as is the case with Infer and Pulse---this means that intermediate procedure specifications cannot be shared between correctness and incorrectness analyses. This is unfortunate, as computing those intermediate specifications in a large codebase is costly.

Fortunately, sharable specifications can be expressed in Outcome Logic.
The trick is to use specifications of the form $\triple {\sure P}C{\sure{Q_1} \oplus \cdots \oplus \sure{Q_n}}$ wherever possible, which both cover all the outcomes (for correctness) and also guarantee reachability of each $\sure{Q_k}$ (for bug finding). Many procedures---even looping ones---fit into this format, \eg the one we saw in \Cref{sec:division}.

On the other hand, some procedures will need to be analyzed using specialized techniques. For example, the use of a loop invariant in \Cref{sec:collatz} introduced a $\always$ modality to indicate that there may not be any terminating outcomes. Alternatively, if a bug arises in one of the paths, the analysis can introduce a $\sometimes$ modality in order to retain less information about the other paths. Both of these can be achieved by an application of the rule of consequence.
\[
  \inferrule{
    \triple{\sure{P}}C{\sure{Q_1} \oplus \cdots \oplus \sure{Q_n}}
    \quad
    \forall k.\ Q_k \Rightarrow Q
  }{
    \triple{\sure P}C{\always Q}
  }
\qquad\qquad
  \inferrule{
    \triple{\sure P}C{\sure{Q_1} \oplus \cdots \oplus \sure{Q_n}}
  }{
    \triple{\sure P}C{\sometimes Q_k}
  }
\]
\citet{zilberstein2024outcome} designed an algorithm for this type of analysis, showing not only that Outcome Logic models Infer and Pulse, but also that the engines of those tools could be consolidated so as to share procedure specifications in many cases. That algorithm is described in \Cref{sec:limitations}.

\section{Case Study: Graph Problems}
\label{sec:graphs}

We now examine case studies using Outcome Logic to derive quantitative properties in alternative models of computation.

\subsection{Counting Random Walks}
\label{sec:walk}

Suppose we wish to count the number of paths between the origin and the point $(N, M)$ on a two dimensional grid. To achieve this, we first write a program that performs a random walk on the grid; while the destination is not yet reached, it nondeterministically chooses to take a step on either the $x$ or $y$-axis (or steps in a fixed direction if the destination on one axis is already reached).
\[
\code{Walk} \triangleq \left\{
\begin{array}{l}
\whl{\var x < N \vee \var y < M}{\;} \\
\quad \code{if} ~ \var x < N \land \var y < M ~\code{then} \\
\quad\quad (\var x \coloneqq \var x +1 ) + (\var y \coloneqq \var y +1 ) \\
\quad \code{else if}~\var x \ge N ~\code{then} \\
\quad\quad \var y \coloneqq \var y + 1 \\
\quad \code{else} \\
\quad\quad \var x \coloneqq \var x + 1
\end{array}\right.
\]
Using a standard program logic, it is relatively easy to prove that the program will always terminate in a state where $x=N$ and $y=M$. However, we are going to interpret this program using the $\mathsf{Nat}$ semiring (\Cref{ex:multisetsemi}) in order to count how many traces (\ie random walks) reach that outcome.

First of all, we know it will take exactly $N+M$ steps to reach the destination, so we can analyze the program using the \ruleref{Variant} rule, where the loop variant $\varphi_n$ records the state of the program $n$ steps away from reaching $(N, M)$.

If we are $n$ steps away, then there are several outcomes ranging from $x = N-n \land y =M$ to $x = N \land y = M-n$. More precisely, let $k$ be the distance to $N$ on the $x$-axis, meaning that the distance to $M$ on the $y$-axis must be $n-k$, so $x = N-k$ and $y = M- (n-k)$. At all times, it must be true that $0 \le x\le N$ and $0\le y\le M$, so it must also be true that $0 \le N-k \le N$ and $0\le M- (n-k) \le M$. solving for $k$, we get that $0 \le k \le N$ and $n-M \le k \le n$. So, $k$ can range between $\max(0, n-M)$ and $\min(N, n)$.

In addition, the number of paths to $(x, y)$ is ${x+y \choose x}$, \ie the number of ways to pick $x$ steps on the $x$-axis out of $x+y$ total steps. Putting all of that together, we define our loop variant as follows:
\[
\varphi_n \triangleq \smashoperator{\bigoplus_{k = \max(0, n-M)}^{\min(N, n)}} \wg{\var x = N-k \land \var y = M - (n-k)}{{N+M - n \choose N-k}}
\]
The loop body moves the program state from $\varphi_{n+1}$ to $\varphi_n$. The outcomes of $\varphi_{n+1}$ get divided among the three if branches. In the outcome where $x=N$ already, $y$ must step, so this goes to the second branch. Similarly, if $y=M$ already, then $x$ must step, corresponding to the third branch. All other outcomes go to the first branch, which further splits into two outcomes due to the nondeterministic choice.
Since we start $N+M$ steps from the destination, we get the following precondition:
\[
\varphi_{N+M} = \bigoplus_{k=N}^N \wg{x = N-k \land y= N-k}{{0 \choose N-N}} = \sure{x=0\land y=0}
\]
In addition, the postcondition is:
\[
\varphi_0 = \bigoplus_{k=0}^0 \wg{x = N-k \land y = M+k}{{N+M \choose N}} = \wg{x=N\land y=M}{{N+M\choose N}}
\]
This gives us the final specification below, which tells us that there are ${N+M\choose N}$ paths to reach $(N, M)$ from the origin. The full derivation is given in \Aref{app:walk}.
\[
\triple{\sure{x=0 \land y=0}}{\code{Walk}}{\wg{x=N\land y=M}{{N+M\choose N}}}
\]

\subsection{Shortest Paths}

We will now use an alternative interpretation of computation to analyze a program that nondeterministically finds the shortest path from $\var s$ to $\var t$ in a directed graph.
Let $G$ be the $N\times N$ Boolean adjacency matrix of a directed graph, so that $G[i][j] =\tru$ if there is an edge from $i$ to $j$ (or $\fls$ if no such edge exists). We also add the following expression syntax to read edge weights in a program, noting that $G[E_1][E_2] \in \mathsf{Test}$ since it is Boolean-valued.
\[
E ::= \cdots \mid G[E_1][E_2]
\]
\[
\de{G[E_1][E_2]}_{\mathsf{Exp}}(s) \triangleq
G
\left[ \de{E_1}_{\mathsf{Exp}}(s) \right]
\left[ \de{E_2}_{\mathsf{Exp}}(s)
\right]
\]
The following program loops until the current position $\var{pos}$ reaches the destination $\var t$. At each step, it nondeterministically chooses which edge ($\var{next})$ to traverse using an iterator; for all $\var{next} \le N$, each trace is selected if there is an edge from $\var{pos}$ to $\var{next}$, and a weight of 1 is then added to the path, signifying that we took a step.
\[
\code{SP} \triangleq \left\{
\begin{array}{l}
\whl{\var{pos}\neq \var{t}}{\;} \\
\quad \var{next} \coloneqq 1 \fatsemi \phantom{x} \\
\quad \iter{\left( \var{next} \coloneqq \var{next}+1\right)}{\var{next} < N}{~G[\var{pos}][\var{next}]} \fatsemi \phantom{x} \\
\quad \var{pos} \coloneqq \var{next}\fatsemi\phantom{x} \\
\quad \assume 1 
\end{array}
\right.
\]
We will interpret this program using the $\mathsf{Tropical}$ semiring from \Cref{ex:tropicalsemi}, in which addition corresponds to $\min$ and multiplication corresponds to addition. So, path lengths get accumulated via addition and nondeterministic choices correspond to taking the path with minimal weight. That means that at the end of the program execution, we should end up in a scenario where $\var{pos} = \var t$, with weight equal to the shortest path length from $\var s$ to $\var t$.

To prove this, we first formalize the notion of shortest paths below: $\short^t_n(G, s, s')$ indicates whether there is a path of length $n$ from $s$ to $s'$ in $G$ in without passing through $t$ and $\short(G, s, t)$ is the shortest path length from $s$ to $t$. Let $I = \{ 1, \ldots , N \} \setminus \{t\}$.
\begin{align*}
\short_0^t(G, s, s') &\triangleq (s = s')
\\
\short_{n+1}^t(G, s, s') &\triangleq \smashoperator{\bigvee_{i\in I}} \short_n^t(G, s,i)  \land G[i][s']
\\
\short(G, s, t) &\triangleq {\inf {\{n \in \mathbb{N}^\infty \mid \short_n^t (G,s,t) \}}} 
\end{align*}
We analyze the while loop using the \ruleref{While} rule, which requires $\varphi_n$ and $\psi_n$ to record the outcomes where the loop guard is true or false, respectively, after $n$ iterations. We define these as follows, where $+$ denotes regular arithmetic addition rather than addition in the tropical semiring.
\[
\varphi_n = \smashoperator{\bigoplus_{i \in I}} \wg{\var{pos} = i}{\short_n^t(G, \var s, i) + n}
\qquad
\psi_n = \wg{\var{pos} = \var t}{\short_n^t(G, \var s, \var t) + n}
\qquad
\psi_\infty = \wg{\var{pos} = \var t}{\short(G, \var s, \var t)}
\]
Recall that in the tropical semiring $\fls = \infty$ and $\tru = 0$. So, after $n$ iterations, the weight of the outcome $\var{pos} = i$ is equal to $n$ if there is an $n$-step path from $s$ to $i$, and $\infty$ otherwise. The final postcondition $\psi_\infty$ is the shortest path length to $t$, which is also the infimum of $\short_n^t(G,s,t)+n$ for all $n$. Using the \ruleref{Iter} rule we get the following derivation for the inner loop:
\[\def\arraystretch{1.5}
\begin{array}{l}
\lrob{\bigoplus_{i\in I }\wg{\var{pos} = i \land \var{next}=1}{\short_n^t(G, s, i) + n}}\\
\quad {\iter{\left( \var{next} \coloneqq \var{next}+1\right)}{\var{next} < N}{~G[\var{pos}][\var{next}]}} \fatsemi\phantom{X}\\
\lrob{\bigoplus_{j=1}^N{\bigoplus_{i\in I }} \wg{\var{pos} = i \land \var{next} = j}{(\short_n^t(G, s, i) \land G[i][j]) + n} } \\
\quad \var{pos} \coloneqq \var{next}\fatsemi\assume 1\\
\lrob{\bigoplus_{j=1}^N\bigoplus_{i\in I } \wg{\var{pos}= j}{(\short_n^t(G, s, i) \land G[i][j]) + n + 1} } \implies \\
\lrob{\bigoplus_{j=1}^N \wg{\var{pos} = j}{\short_{n+1}^t(G, s, j) + n+1}} \\
\end{array}
\]
The outcome conjunction over $i\neq t$ (corresponding to the minimum weight path) gives us $\var{pos} = j$ with weight $\short_{n+1}^t(G, s, j)+n+1$---it is $n+1$ if there is path of length $n$ to some $i$ and $G[i][j]$.

The precondition is $\varphi_0\oplus\psi_0 = (\var{pos} = s)$, since $\short_0^t(G, s, i) = \fls$ when $i\neq s$ and $\tru$ when $i=s$.
Putting this all together, we get the following triple, stating that the final position is $t$ and the weight is equal to the shortest path.
\[
\triple{\sure{\var{pos} = s}}{\code{SP}}{\wg{\var{pos} = t}{\short(G, s, t)}}
\]
This program does not terminate if there is no path from $s$ to $t$. In that case there are no reachable outcomes, so the interpretation of the program is $\zero$. Indeed, $\zero = \infty$ in the tropical semiring, which is the shortest path between two disconnected nodes. The postcondition is $\wg{\var{pos} = t}{\infty}$, meaning that the program diverged.

%
%

\section{Related Work}\label{sec:discussion}


\heading{Correctness, Incorrectness, and Unified Program Logics.}
While formal verification has long been the aspiration for static analysis, bug-finding tools are often more practical in real-world engineering settings. This partly comes down to efficiency---bugs can be found without considering all traces---and partly due to the fact that most real world software just is not correct \cite{distefano2019scaling}.

However, the standard logical foundations of program analysis such as Hoare Logic are prone to \emph{false positives} when used for bug-finding---they cannot witness the existence of erroneous traces. In response, \citet{il} developed Incorrectness Logic, which under-approximates the reachable states (as opposed to Hoare Logic's over-approximation) so as to only report bugs that truly occur.

Although Incorrectness Logic successfully serves as a logical foundation for bug-finding tools \cite{isl,realbugs}, it is semantically incompatible with correctness analysis, making sharing of toolchains difficult or impossible. Attention has therefore turned to unifying correctness and incorrectness theories. This includes Exact Separation Logic, which combines Hoare Logic and Incorrectness Logic to generate specifications that are  valid for both, but that also precludes under- or over-approximation via the rule of consequence \cite{maksimovi_c2023exact}. Local Completeness Logic combines Incorrectness Logic with an over-approximate abstract domain to similar effect; it also precludes dropping paths \cite{Bruni2021ALF,brunijacm}.

There have also been recent efforts to organize correctness and incorrectness logics into taxonomies, where their similarities and differences are expressed in terms of weakening, contrapositives, and Galois connections \cite{cousot2024calculational,qsp,ascari2025revealing,verscht2025taxonomy}.
Our goal in this article is orthogonal; we capture many different types of specifications with a \emph{single} semantics and proof theory, so as to reuse proof fragments as much as possible. We also include several kinds of effects rather than just nondeterminism.

\heading{Outcome Logic.}
Outcome Logic unifies correctness and incorrectness reasoning without compromising the use of logical consequences. This builds on an idea colloquially known as \emph{Lisbon Logic}, first proposed by Derek Dreyer and Ralf Jung in 2019, that has similarities to the diamond modality of Dynamic Logic \cite{pratt1976semantical,harel2001dynamic} and \cites{wpp} calculus of \emph{possible correctness}. The idea was briefly mentioned in the Incorrectness Logic literature \cite{il,ilalgebra,realbugs}, but using Lisbon Logic as a foundation of incorrectness analysis was not fully explored until the introduction of Outcome Logic \cite{outcome}, which subsumes both Lisbon Logic and Hoare Logic, as we saw in \Cref{sec:dynamichoare}. The metatheory of Lisbon Logic has  subsequently been explored more deeply and extended \cite{ascari2025revealing,raad2024nontermination}. Hyper Hoare Logic also generalizes Hoare and Lisbon Logics \cite{hyperhoare}, and is semantically equivalent to the Boolean instance of OL (\Cref{ex:powersetsemi}), but does not support effects other than nondeterminism.

Outcome Logic initially used a model based on both a monad and a monoid, with looping defined via the Kleene star $C^\star$ \cite{outcome}. The semantics of $C^\star$ had to be justified for each instance,
however it is not compatible with probabilistic computation (see \Cref{foot:star}), so an ad-hoc semantics was used in the probabilistic case. 
Moreover, only the \textsc{Induction} rule was provided for reasoning about $C^\star$, amounting to unrolling the loop one time. Some loops can be analyzed by applying \textsc{Induction} repeatedly, but it is inadequate if the number of iterations depends at all on the program state.
Our $\iter Ce{e'}$ construct fixes this, defining iteration in a way that supports both Kleene star ($\iter C\one\one$) and also probabilistic computation. As we showed in \Cref{sec:examples,sec:graphs}, \ruleref{Iter} can be used to reason about any loop, even ones that iterate an unbounded number of times.

The next Outcome Logic paper focused on a particular separation logic instance \cite{zilberstein2024outcome}. The model was refined to use semirings, and the programming language included while loops instead of $C^\star$ so that a single well-definedness proof could extend to all instances. However, the evaluation model included additional constraints ($\one = \top$ and normalization) that preclude, \eg the multiset model (\Cref{ex:multisetsemi}). Rather than giving inference rules, the paper provided a symbolic execution algorithm, which also only supported loops via bounded unrolling. 

This article extends the prior conference papers on Outcome Logic by giving a more general model with more instances and better support for iteration, providing a relatively complete proof system that is able to handle any loop, deriving additional inference rules for Hoare and Lisbon Logic embeddings, and exploring case studies related to previously unsupported types of computation and looping.

\heading{Computational Effects.}
Effects have been present since the early years of program analysis. Even basic programming languages with while loops introduce the possibility of \emph{nontermination}. Partial correctness was initially used to sidestep the termination question \cite{FLOYD67,hoarelogic}, but total correctness (requiring termination) was later introduced too \cite{manna1974axiomatic}. More recently, automated tools were developed to prove (non)termination in real-world software \cite{berdine2006automatic,brockschmidt2013better,cook2014disproving,cook2006termination1,cook2006termination2,raad2024nontermination}.

Nondeterminism also showed up in early variants of Hoare Logic, stemming from Dijkstra's Guarded Command Language (GCL) \cite{GCL} and Dynamic Logic \cite{pratt1976semantical,harel2001dynamic};
it is useful for modeling backtracking algorithms \cite{floyd1967nondeterministic} and opaque aspects of program evaluation such as user input and concurrent scheduling. While Hoare Logic has traditionally used demonic nondeterminism \cite{broy1981algebraic}, other program logics have recently arisen to deal with nondeterminism in different ways, particularly for incorrectness \cite{reversehoare,il,ilalgebra,outcome,raad2024nontermination,ascari2025revealing,hyperhoare}.

Beginning with the seminal work of \citet{psem,ppdl}, the study of probabilistic programs has a rich history. This eventually led to the development of probabilistic Hoare Logic variants \cite{hartog,ellora,rand2015vphl,den_hartog1999verifying} that enable reasoning about programs in terms of likelihoods and expected values. 
Doing so requires pre- and postconditions to be predicates on probability distributions.

Outcome Logic provides abilities to reason about those effects using a common set of inference rules. This opens up the possibility for static analysis tools that soundly share proof fragments between different types of programs, as shown in \Cref{sec:examples}.

\heading{Relative Completeness and Expressivity.}
Any sufficiently expressive program logic must necessarily be incomplete since, for example, the Hoare triple $\hoare{\tru}C{\fls}$ states that the program $C$ never halts, which is not provable in an axiomatic deduction system. In response, \citet{cook1978soundness} devised the idea of \emph{relative} completeness to convey that a proof system is adequate for analyzing a program, but not necessarily assertions about the program states.

Expressivity requires that the assertion language used in pre- and postconditions can describe the intermediate program states needed to, \eg apply the \ruleref{Seq} rule. In other words, the assertion syntax must be able to express $\spost(C, P)$ from \Cref{def:post}. Implications for an \emph{expressive} language quickly become undecidable, as they must encode Peano arithmetic \cite{apt1981ten,lipton1977necessary}.
With this in mind, the best we can hope for is a program logic that is complete \emph{relative} to an oracle that decides implications in the rule of \ruleref{Consequence}.

The question of what an expressive (syntactic) assertion language for Outcome Logic looks like remains open. In fact, the question of expressive assertion languages for probabilistic Hoare Logics (which are subsumed by Outcome Logic) is also open \cite{den_hartog1999verifying,ellora}. \citet{hartog} devised a relatively complete probabilistic logic with syntactic assertions, but the programming language does not include loops and is therefore considerably simplified; it is unclear if this approach would extend to looping programs. In addition \citet{batz2021relatively} created an expressive language for expectation-based reasoning, however the language only has constructs to describe states and expected values; it does not contain a construct like $\oplus$ to express properties involving of multiple traces of the program. As we saw in the discussion following \Cref{def:post}, the ability to describe multiple executions makes the expressivity question more complex.

Several program logics (including our own) use semantic assertions, which are trivially expressive \cite{semanticsep,yang,il,iris1,iris,ellora,kaminski,hyperhoare,raad2024nontermination,cousot2012abstract,ascari2025revealing}. This includes logics that are mechanized within proof assistants \cite{ellora,iris,iris1,hyperhoare}, so we do not see the extensional nature of our approach as a barrier to mechanization in the future.


\heading{Quantitative Reasoning and Weighted Programming.}
Whereas Hoare Logic provides a foundation for \emph{propositional} program analysis, quantitative program analysis has been explored too. 
Probabilistic Propositional Dynamic Logic \cite{ppdl} and weakest pre-expectation calculi \cite{wpe,kaminski,mciver2005abstraction} are used to reason about randomized programs in terms of expected values. This idea has been extended to non-probabilistic quantitative properties too \cite{qsp,batz2022weighted,zhang2024quantitative}.

Weighted programming \cite{batz2022weighted} generalizes pre-expectation reasoning using semirings to model branch weights, much like this paper. Outcome Logic is a propositional analogue to weighted programming's quantitative model, but it is also more expressive in its ability to reason about quantities over \emph{multiple outcomes}. For example, in \Cref{sec:prob}, we derive a single OL triple that gives the probabilities of two outcomes, whereas weighted programming (or weakest pre-expectations) would need to compute the weight of each branch individually. In addition, \citet{batz2022weighted} only support total semirings, so they cannot analyze standard probabilistic programs. Weighted programming was extended to handle hyperproperties by \citet{zhang2024quantitative}, which can be seen as a weakest precondition calculus for Outcome Logic.


\section{Conclusion, Limitations, and Future Work}
\label{sec:limitations}

In this article, we have presented a proof system for Outcome Logic, which is sufficient for reasoning about programs with branching effects. That is, effects that can be semantically encoded by assigning weights to a collection of final states. However, there remains room for future development of Outcome Logic in order to support additional kinds of effects and other capabilities. In this section, we describe the limitations of the present formalization, and how those limitations could be addressed in future work.

\heading{Total Correctness and Nontermination.}

Although Outcome Logic provides some ability to reason about nontermination---\ie it can be used to prove that a program \emph{sometimes} or \emph{never} terminates---it cannot be used for total correctness \cite{manna1974axiomatic}, as the semantics only tracks terminating traces, not the existence of infinite ones. Given some specification $\triple\varphi{C}\psi$, it is therefore impossible to know whether additional nonterminating behavior outside of $\psi$ is possible.

The challenge in tracking infinite traces is that the semantics can easily become non-continuous, making it more difficult to prove the existence of fixed points used to define the semantics of loops. More precisely, \citet{apt1986countable} showed that no continuous semantics is possible that both distinguishes nontermination and also supports unbounded nondeterminism. That means that, at least in the nondeterministic interpretation (\Cref{ex:powersetsemi}), a primitive action of the following form, which assigns $x$ to a nondeterministically chosen natural number, must be precluded.
\[
  \de{x \coloneqq \bigstar}(s)
  \quad\triangleq\quad
  \sum_{n\in\mathbb{N}} \eta(s[ x \coloneqq n ])
  \quad\cong\quad
  \{ s[ x \coloneqq n ]\mid n\in\mathbb{N} \}
\]
Note that in the present formulation, we could define the above construct as syntactic sugar, as we did in \Cref{sec:hyper}. The program below has the exact same semantics as above, since our model in \Cref{sec:semantics} does not record the fact that it also has an infinite (nonterminating) execution.
\[
  x \coloneqq \bigstar
  \quad\triangleq\quad
  x \coloneqq 0 \fatsemi (x \coloneqq x+1)^\star
\]
However, in a total correctness version, the semantics of those programs would not be the same; the version defined with the Kleene star has an additional nonterminating trace.
This corresponds to \cites{Dijkstra76} observation that a program cannot make infinitely many choices in a finite amount of time,  giving an operational argument for why unbounded nondeterminism should not be allowed. However, the story is more murky when generalizing to the semiring weighted model that we have presented in this paper. Although unbounded choice is not valid in the nondeterministic model, it \emph{is} valid in the probabilistic model (\Cref{ex:probsemi}). For example, the following probabilistic program makes infinitely many choices, but it also \emph{almost surely terminates} (there is an infinite trace, but it occurs with probability 0, so we do not consider it a possible outcome).
\[
  x \coloneqq 0 \fatsemi (x \coloneqq x+1)^{\langle \frac12\rangle}
\]
So, in the total correctness setting, the continuity of program operators depends on the semiring in a non-obvious way, which cannot be explained by existing work on powerdomains \cite{plotkin1976powerdomain,smyth1978power,s_ondergaard1992non}. Some work in this area has already been done by \citet{li2024totaloutcomelogicproving}, who develop a semantic model for Total Outcome Logic that explains this discrepancy in terms of properties of the semirings.

\heading{Mixing Effects.} While Outcome Logic, as described in this article, can be used to analyze programs that have various types of branching, it cannot handle programs that display multiple kinds of branching in tandem. Of particular interest is the combination of probabilistic choice and nondeterminism, which would be useful to analyze randomized distributed systems (where nondeterminism arises due to concurrent scheduling).
Unfortunately, the typical powerset monad representation of nondeterminism (\Cref{ex:powersetsemi}) does not compose well with probability distributions (\Cref{ex:probsemi}) \cite{varacca_winskel_2006}, and similar restrictions apply to other combinations of semirings \cite{parlant2020monad,zwart2019no,zwart2020non}.

Demonic Outcome Logic (DOL) \cite{zilberstein2025demonic} was recently introduced for reasoning about programs that are both probabilistic and nondeterministic, with semantics based on the convex powerset monad \cite{morgan1996refinement,jifeng1997probabilistic}, a specialized computational domain, which does not generalize like the semiring model that we use in this paper does. DOL supports reasoning about \emph{probabilistic} outcomes with $\oplus$, whereas nondeterminism is treated demonically (the postcondition applies to all nondeterministic paths). This makes DOL a correctness logic only, it cannot witness the \emph{existence} of bad distributions of outcomes. It would be interesting to develop a dual \emph{Angelic Outcome Logic} for reasoning about incorrectness, or even a logic that unifies the two.

Alternatively, it was recently shown that probabilistic monads do compose with multisets \cite{jacobs2021multisets,kozen2023multisets}. In \Cref{sec:walk}, we showed how the Outcome Logic instance based on multisets (\Cref{ex:multisetsemi}) can be used for quantitative analysis, however that example did not correspond to a canonical model of computation. By contrast, a logic using that same multiset instance for probabilistic nondeterminism would model a very pertinent combination of effects, motivating the study of these more exotic forms of nondeterminism in this paper. The powerdomain of indexed valuations \cite{varacca2002powerdomain,varacca_winskel_2006} is a similar computational domain that is also constructed via a distributive law between two weighting functions. Given these many different models, it would be interesting to investigate a variant of Outcome Logic where multiple weighting functions can be composed together.

\heading{Mutable State, Separation Logic, and Concurrency.}
As an extension to our Outcome Logic formulation in this article, it would be possible to add atomic commands for dynamically allocating and mutating heap pointers, and separation logic \cite{sl,localreasoning} style inference rules to reason about mutation. Separation logic introduces new logical primitives such as the points-to predicate $x \mapsto E$---stating that the $x$ is a pointer whose address maps to the value of $E$ on the heap---and the separating conjunction $P\sep Q$---stating that $P$ and $Q$ hold in disjoint regions of the heap, guaranteeing that if $x\mapsto E_1 \sep y\mapsto E_2$ holds, then $x$ and $y$ do not alias each other.

This idea was already explored in Outcome Separation Logic (OSL) \cite{zilberstein2024outcome}, but questions relating to unbounded looping and completeness were not addressed in that work. In separation logic, it is typical to provide \emph{small axioms} \cite{localreasoning}, which specify the behavior of mutation locally. For example, the following \ruleref{Write} rule specifies the behavior of mutation in a singleton heap. Global reasoning is then done via the \ruleref{Frame} rule \cite{yang2002semantic}, which states that a proof done in a local footprint is also valid in a larger heap.
\[
  \ruledef{Write}{\;}{\hoare{x \mapsto -}{[x] \gets E}{x \mapsto E}}
  \qquad\qquad
  \ruledef{Frame}{\hoare PCQ}{\hoare{P\sep F}C{Q\sep F}}
\]
\citet{yang} showed that the small axioms along with the \ruleref{Frame} rule are complete for standard separation logic; any specification can be derived using them. By contrast, \citet{zilberstein2024outcome} showed that the small axioms and \ruleref{Frame} rule are sound in OSL, but it is not known whether they are complete. In fact, there is some evidence to the contrary, since the resources specified outside the footprint of the local mutation may take on many different values, and therefore the \ruleref{Frame} rule would need to be applied for each branch of the computation. OSL could be obtained as an instantiation of the Outcome Logic model in this article, although the \ruleref{Frame} rule would need to be added (and as such, this article does not fully subsume OSL). It would be interesting to add such a frame rule, and explore the notion of completeness in that setting.

\citet{zilberstein2024outcome} also developed an algorithm to mechanically derive triples of the form $\triple{\sure P}C{\sure{Q_1} \oplus \cdots \oplus\sure{Q_n}}$, using the \ruleref{Seq} and \ruleref{Choice} rules to compositionally reason about sequential compositions via the following derived rule.
\[
  \inferrule{
    \triple{\sure P}{C_1}{\textstyle\bigoplus_{i\in I} \sure{Q_i}}
    \\
    \forall i\in I.\quad
    \triple{\sure{Q_i}}{C_2}{\textstyle\bigoplus_{j\in J_i}\sure{R_j}}
  }{
    \triple{\sure P}{C_1 \fatsemi C_2}{\textstyle\bigoplus_{i\in I, j\in J_i}\sure{R_j}}
  }
\]
The \ruleref{Frame} rule is a crucial part of the algorithm too; it is used to lift axioms about actions---such as the \ruleref{Write} rule above---into the current program context. The particular frames are inferred using \emph{biabduction} \cite{biab}, a heuristic procedure for inferring the \emph{missing} antiframe $M$ and leftover \emph{frame} $F$ in the following entailment $P \sep M \vdash Q \sep F$. So, for any precondition $\sure P$, a write action can be analyzed as follows:
\[
  \inferrule*[right=\ruleref{Consequence}]{
    P \sep M \vdash (x\mapsto-) \sep F
    \\
    \inferrule*[right=\ruleref{Frame}]{
      \inferrule*[right=\ruleref{Write}]{\;}{
        \triple{\sure{x \mapsto -}}{[x] \gets E}{\sure{x \mapsto E}}
      }
    }{
        \triple{\sure{(x \mapsto -) \sep F}}{[x] \gets E}{\sure{(x \mapsto E) \sep F}}
    }
  }{
        \triple{\sure{P \sep M}}{[x] \gets E}{\sure{(x \mapsto E) \sep F}}
  }
\]
One weakness of the algorithm is that it only has the ability to reason about loops via bounded unrolling. It would be interesting to extend the algorithm to use some of loop rules that we presented in this article. In particular, an algorithm could be developed for approximating loop behavior in Lisbon Logic using abstract domains \cite{cousot1977abstract}, something that cannot be done in Incorrectness Logic \cite{ascari}.

Going further, it would also be interesting to explore a concurrent variant of OSL. Along with the aforementioned challenges of reasoning about weighted branching in tandem with the nondeterministic behavior of the concurrent scheduler, this would also require a more powerful version of separation than that offered by OSL. The OSL separating conjunction $\varphi \osep P$ is asymmetric; $\varphi$ is an unrestricted outcome assertion, whereas $P$ can only be a basic assertion, like those described in \Cref{sec:assertions}. However, concurrent separation logic \cite{csl} typically requires the entire state space to be separated in order to compositionally analyze concurrent threads, so a symmetric separating conjunction $\varphi \osep\psi$ would be needed. Preliminary work has already been done to this effect \cite{zilberstein2024probabilistic,zilberstein2025denotational} based on Probabilistic Separation Logic \cite{psl}.

\medskip

\noindent Computational effects have traditionally beckoned disjoint program logics across two dimensions: different \emph{kinds} of effects (\eg nondeterminism vs randomization) and different \emph{assertions} about those effects (\eg angelic vs demonic nondeterminism).
Outcome Logic \cite{outcome,zilberstein2024outcome} captures those properties in a unified way, but until now the proof theory and connections to other logics have not been thoroughly explored.

This article expands on the prior Outcome Logic work and provides a comprehensive account of the OL metatheory by presenting a relatively complete proof system for Outcome Logic and a significant number of derived rules.
This shows that programs with branching effects are not only \emph{semantically} similar, but also share common reasoning principles. Specialized techniques (\ie analyzing loops with variants or invariants) are particular modes of use of our more general logic, and are compatible with each other rather than requiring distinct semantics.

\section*{Acknowledgements}

I would like to thank the anonymous TOPLAS referees for their detailed and thoughtful reviews, which greatly helped to improve this article. This work was partially supported by the \emph{Advanced Research and Invention Agency}'s (ARIA)  Safeguarded AI programme, under the project: \emph{Unified Automated Reasoning for Randomised Distributed Systems}.


 \bibliography{completeness}

\ifx\extended\undefined\else
\allowdisplaybreaks
\appendix

\clearpage
\ifx\apponly\undefined\else
\setcounter{page}{1}
\fi

{\noindent \huge\bfseries\sffamily Appendix}

\section{Totality of Language Semantics}
\label{app:totality}

\subsection{Semantics of Tests and Expressions}
\label{app:tests}

Given some semiring $\langle U,+,\cdot,\zero,\one\rangle$, the definition of the semantics of tests $\de{b}_{\mathsf{Test}} \colon \Sigma \to \{ \zero,\one\}$ is below.
\begin{align*}
\de{\tru}_{\mathsf{Test}}(\sigma) &\triangleq \one \\
\de{\fls}_{\mathsf{Test}}(\sigma) &\triangleq \zero \\
\de{b_1 \vee b_2}_{\mathsf{Test}}(\sigma) &\triangleq \left\{
\begin{array}{ll}
\one & \text{if} ~ \de{b_1}_{\mathsf{Test}}(\sigma) = \one ~\text{or}~ \de{b_2}_{\mathsf{Test}}(\sigma) = \one \\
\zero & \text{otherwise}
\end{array}
\right. \\
\de{b_1 \land b_2}_{\mathsf{Test}}(\sigma) &\triangleq \left\{
\begin{array}{ll}
\one & \text{if} ~ \de{b_1}_{\mathsf{Test}}(\sigma) = \one ~\text{and}~ \de{b_2}_{\mathsf{Test}}(\sigma) = \one \\
\zero & \text{otherwise}
\end{array}
\right. \\
\de{\lnot b}_{\mathsf{Test}}(\sigma) &\triangleq \left\{
\begin{array}{ll}
\one & \text{if} ~ \de{b}_{\mathsf{Test}}(\sigma) = \zero \\
\zero & \text{if} ~ \de{b}_{\mathsf{Test}}(\sigma) = \one
\end{array}
\right. \\
\de{t}_{\mathsf{Test}}(\sigma) &\triangleq \left\{
\begin{array}{ll}
\one & \text{if} ~ \sigma\in t \\
\zero & \text{if} ~ \sigma \notin t
\end{array}
\right.
\end{align*}
Based on that, we define the semantics of expressions $\de{e}\colon \Sigma\to U$.
\begin{align*}
\de{b}(\sigma) &\triangleq \de{b}_\mathsf{Test}(\sigma) \\
\de{u}(\sigma) &\triangleq u
\end{align*}

\subsection{Fixed Point Existence}

For all the proofs in this section, we assume that the operations $+$, $\cdot$, and $\sum$ belong to a Scott continuous, naturally ordered, partial semiring with a top element (as described in \Cref{sec:denotation}).

\begin{lemma}\label{lem:sumcont}
Let $\langle U, +, \cdot, \zero,\one\rangle$ be a continuous, naturally ordered, partial semiring. For any family of Scott continuous functions $(f_i : X\to \mathcal W(Y))_{i\in I}$ and directed set $D\subseteq X$:
\[
\sup_{x\in D} \sum_{i\in I} f_i(x) = \sum_{i\in I} f_i(\sup D)
\]
\end{lemma}
\begin{proof}
The proof proceeds by transfinite induction on the size of $I$.
\begin{itemize}
\item \textbf{Base case:} $I=\emptyset$, so clearly:
\[
\sup_{x\in D}\sum_{i\in \emptyset} f_i(x) = \sup_{x\in D} \zero = \zero = \sum_{i\in \emptyset} \sup_{x\in D} f_i(x)
\]
\item \textbf{Successor Case:} Suppose the claim holds for sets of size $\alpha$, and let $|I| = \alpha+1$. We can partition $I$ into $I' \cup \{ i \}$ where $i\in I$ is an arbitrary element and $I' = I\setminus \{i\}$ so that $|I'| = \alpha$. Now, we have that:
\begin{align*}
  \sup_{x\in D} \sum_{j\in I} f_j(x)
  &= \sup_{x\in D} (\sum_{j\in I'} f_j(x)) + f_{i}(x) \\
  \intertext{Now, note that since all the $f_j$ functions are Scott continuous, they must also be monotone, and addition is also monotone. Therefore the following equality holds \cite[Proposition 2.1.12]{abramsky1995domain}.}
  &= \sup_{x\in D}\sup_{y\in D}(\sum_{j\in I'} f_j(x)) + f_{i}(y) \\
  \intertext{By continuity of the semiring:}
  &= (\sup_{x\in D} \sum_{j\in I'} f_j(x)) + (\sup_{y\in D} f_{i}(y)) \\
  \intertext{By continuity of $f_{i}$ and the induction hypothesis:}
  &= \sum_{j\in I'} f_j(\sup D) + f_{i}(\sup D) \\
  &= \sum_{i\in I} f_i(\sup D)
\end{align*}

\item \textbf{Limit case:} suppose that the claim holds for all finite index sets. Now, given the definition of the sum operator:
\begin{align*}
\sup_{x\in D} \sum_{i\in I} f_i(x)
&= \sup_{x\in D} \sup_{J \subseteq_{\mathsf{fin}} I} \sum_{j \in J} f_j(x)\\
\intertext{The finite subsets of $I$ are a directed set and clearly the inner sum is monotone in $x$ and $J$, so we can rearrange the suprema \cite[Proposition 2.1.12]{abramsky1995domain}.}
&=  \sup_{J \subseteq_{\mathsf{fin}} I} \sup_{x\in D} \sum_{j \in J} f_j(x)\\
\intertext{By the induction hypothesis:}
&=  \sup_{J \subseteq_{\mathsf{fin}} I} \sum_{j \in J} f_j(\sup D)\\
&= \sum_{i\in I} f_i(\sup D)
\end{align*}
\end{itemize}
\end{proof}

%
%

\begin{lemma}\label{lem:summult}
If $\sum_{i\in I} u_i$ is defined, then for any $(v_i)_{i\in I}$, $\sum_{i\in I} u_i\cdot v_i$ is defined.
\end{lemma}
\begin{proof}
Let $v$ be the top element of $U$, so $v \ge v_i$ for all $i\in I$. That means that for each $i\in I$, there is a $v_i'$ such that $v_i+v'_i=v$. Now, since multiplication is total, then we know that $(\sum_{i\in I} u_i)\cdot v$ is defined. This gives us:
\[
(\sum_{i\in I} u_i)\cdot v
= \sum_{i\in I} u_i\cdot (v_i + v'_i) 
= \sum_{i\in I}u_i\cdot v_i +  \sum_{i\in I}u_i\cdot v'_i
\]
And since $\sum_{i\in I}u_i\cdot v_i$ is a subexpression of the above well-defined term, then it must be well-defined.
\end{proof}
\begin{lemma}\label{lem:binddef}
For any $m \in \mathcal W(X)$ and $f\colon X\to\mathcal W(Y)$, we get that $f^\dagger \colon \mathcal W(X) \to \mathcal W(Y)$ is a total function.
\end{lemma}
\begin{proof}
First, recall the definition of $(-)^\dagger$:
\[
f^\dagger(m)(y) = \smashoperator{\sum_{x\in\supp(m)}} m(x)\cdot f(x)(y)
\]
To show that this is well-defined, we need to show both that the sum exists, and that the resulting weighting function has a well-defined mass. First, we remark that since $m\in\mathcal W(A)$, then $|m| = \sum_{x\in\supp(m)} m(x)$ must be defined. By \Cref{lem:summult}, the sum in the definition of $(-)^\dagger$ is therefore defined. Now, we need to show that $|f^\dagger(m)|$ is defined:
\begin{align*}
|f^\dagger(m)|
&= \smashoperator{\sum_{y \in \supp(f^\dagger(m))}} f^\dagger(m)(y) \\
&= \smashoperator[l]{\sum_{y \in \bigcup_{a\in \supp(m)} \supp(f(a))}} \; \smashoperator[r]{\sum_{x\in\supp(m)}} m(x)\cdot f(x)(y) \\
\intertext{By commutativity and associativity:}
&= \smashoperator{\sum_{x\in\supp(m)}} m(x) \;\cdot\; \smashoperator{\sum_{y \in \supp(f(x))}} f(x)(y) 
= \smashoperator{\sum_{x\in\supp(m)}} m(x) \cdot |f(x)|
\end{align*}
Now, since $f(x) \in \mathcal W(Y)$ for all $x \in X$, we know that $|f(x)|$ must be defined. The outer sum also must be defined by \Cref{lem:summult}.
\end{proof}

In the following, when comparing functions $f,g\colon X\to\mathcal{W}(Y)$, we will use the pointwise order. That is, $f \sqsubseteq^\bullet g$ iff $f(x) \sqsubseteq g(x)$ for all $x\in X$.

\begin{lemma}\label{lem:bindcont}
$(-)^\dagger \colon (X \to \mathcal W(Y)) \to (\mathcal W(X) \to \mathcal W(Y))$ is Scott continuous.
\end{lemma}
\begin{proof}
Let $D\subseteq (X\to\mathcal W(Y))$ be a directed set.
First, we show that for any $x\in X$, the function $g(f) = m(x)\cdot f(x)(y)$ is Scott continuous:
\begin{align*}
\sup_{f\in D}g(f)
&= \sup_{f\in D} m(x)\cdot f(x)(y) \\
\intertext{By Scott continuity of the $\cdot$ operator:}
&=  m(x)\cdot \sup_{f\in D} f(x)(y) \\
\intertext{Since we are using the pointwise ordering:}
&=  m(x)\cdot (\sup D)(x)(y) = g(\sup D)
\end{align*}
Now, we show that $(-)^\dagger$ is Scott continuous. For any $m$, we have:
\begin{align*}
(\sup_{f\in D} f^\dagger)(m)
&= \sup_{f\in D} f^\dagger(m) \\
&= \sup_{f\in D} \smashoperator[r]{\sum_{x\in\supp(m)}} m(x)\cdot f(x) \\
\intertext{By \Cref{lem:sumcont}, using the property we just proved.}
&= \smashoperator{\sum_{x\in\supp(m)}} m(x)\cdot (\sup D)(x)
= (\sup D)^\dagger(m)
\end{align*}
\end{proof}
\begin{lemma}\label{lem:phicont}
Let $\Phi_{\langle C, e, e'\rangle}(f)(\sigma) = \de{e}(\sigma)\cdot f^\dagger(\de{C}(\sigma)) + \de{e'}(\sigma)\cdot\unit(\sigma)$ and suppose that it is a total function, then $\Phi_{\langle C, e, e'\rangle}$ is Scott continuous with respect to the pointwise order: $f_1 \sqsubseteq^\bullet f_2$ iff $f_1(\sigma) \sqsubseteq f_2(\sigma)$ for all $\sigma\in \Sigma$.
\end{lemma}
\begin{proof}
For all directed sets $D\subseteq (\Sigma\to\mathcal W(\Sigma))$ and $\sigma\in\Sigma$, we have:
\begin{align*}
&\sup_{f\in D} \Phi_{\langle C, e, e'\rangle}(f)(\sigma)\\
&= \sup_{f\in D} \left(\de{e}(\sigma)\cdot f^\dagger(\de{C}(\sigma)) + \de{e'}(\sigma)\cdot \unit(\sigma)\right) \\
\intertext{By the continuity of $+$ and $\cdot$, we can move the supremum up to the $(-)^\dagger$, which is the only term that depends on $f$.}
&= \de{e}(\sigma)\cdot (\sup_{f\in D} f^\dagger(\de{C}(\sigma))) + \de{e'}(\sigma)\cdot \unit(\sigma) \\
\intertext{By \Cref{lem:bindcont}.}
&= \de{e}(\sigma)\cdot (\sup D)^\dagger(\de{C}(\sigma)) + \de{e'}(\sigma)\cdot \unit(\sigma) \\
&= \Phi_{\langle C, e, e'\rangle}(\sup D)(\sigma)
\end{align*}
Since this is true for all $\sigma\in\Sigma$, 
$\Phi_{\langle C, e, e'\rangle}$ is Scott continuous.
\end{proof}
Now, given \Cref{lem:phicont} and the Kleene fixed point theorem, we know that the least fixed point is defined and is equal to:
\[
\mathsf{lfp}\left(\Phi_{\langle C, e, e'\rangle}\right) = \sup_{n\in\mathbb N}\Phi^n_{\langle C, e, e'\rangle}(\lambda \tau.\zero)
\]
Therefore the semantics of iteration loops is well-defined, assuming that $\Phi_{\langle C, e, e'\rangle}$ is total. In the next section, we will see simple syntactic conditions to ensure this.

\subsection{Syntactic Sugar}

If a partial semiring is used to interpret the language semantics, then unrestricted use of the $C_1+C_2$ and $\iter C{e}{e'}$ constructs may be undefined. In this section, we give some sufficient conditions to ensure that program semantics is well-defined. This is based on the notion of compatible expressions, introduced below.
\begin{definition}[Compatibility]
The expressions $e_1$ and $e_2$ are compatible in semiring $\mathcal A = \langle U,+,\cdot,\zero,\one\rangle$ if $\de{e_1}(\sigma) + \de{e_2}(\sigma)$ is defined for any $\sigma\in\Sigma$.
\end{definition}
The nondeterministic (\Cref{ex:powersetsemi,ex:multisetsemi}) and tropical (\Cref{ex:tropicalsemi}) instances use total semirings, so any program has well-defined semantics. In other 
interpretations, we must ensure that programs are well-defined by ensuring that all uses of choice and iteration use compatible expressions. We begin by showing that any two collections can be combined if they are scaled by compatible expressions.
\begin{lemma}\label{lem:compatscale}
If $e_1$ and $e_2$ are compatible, then $\de{e_1}(\sigma)\cdot m_1 + \de{e_2}(\sigma)\cdot m_2$ is defined for any $m_1$ and $m_2$.
\end{lemma}
\begin{proof}
Since $e_1$ and $e_2$ are compatible, then $\de{e_1}(\sigma) +\de{e_2}(\sigma)$ is defined. By \Cref{lem:summult}, that also means that $\de{e_1}(\sigma)\cdot |m_1| + \de{e_2}(\sigma)\cdot |m_2|$ is defined too. Now, we have:
\begin{align*}
&\de{e_1}(\sigma)\cdot |m_1| + \de{e_2}(\sigma)\cdot |m_2| \\
&= \de{e_1}(\sigma)\cdot \smashoperator{\sum_{\tau\in\supp(m_1)}} m_1(\tau) + \de{e_2}(\sigma)\cdot \smashoperator{\sum_{\tau\in\supp(m_2)}} m_2(\tau) \\
&= \smashoperator{\sum_{\tau\in\supp(m_1)}} \de{e_1}(\sigma)\cdot m_1(\tau) + \smashoperator{\sum_{\tau\in\supp(m_2)}} \de{e_2}(\sigma)\cdot m_2(\tau) \\
&= \smashoperator{\sum_{\tau\in\supp(m_1)\cup\supp(m_2)}} \de{e_1}(\sigma)\cdot m_1(\tau) + \de{e_2}(\sigma)\cdot m_2(\tau) \\
&= | \de{e_1}(\sigma)\cdot m_1 +  \de{e_2}(\sigma)\cdot m_2 |
\end{align*}
Therefore $\de{e_1}(\sigma)\cdot m_1 + \de{e_2}(\sigma)\cdot m_2$ must be well-defined.
\end{proof}
Now, we show how this result relates to program semantics. We begin with branching, by showing that guarding the two branches using compatible expressions yields a program that is well-defined.
\begin{lemma}
If $e_1$ and $e_2$ are compatible and $\de{C_1}$ and $\de{C_2}$ are total functions, then:
\[
\de{(\assume{e_1}\fatsemi C_1) + (\assume{e_2}\fatsemi C_2)}
\]
is a total function.
\end{lemma}
\begin{proof}
Take any $\sigma\in\Sigma$, then we have:
\begin{align*}
\de{(\assume{e_1}\fatsemi C_1) + (\assume{e_2}\fatsemi C_2)}(\sigma)
&= \dem{C_1}{\de{\assume{e_1}}(\sigma)} + \dem{C_2}{\de{\assume{e_2}}(\sigma)} \\
&= \dem{C_1}{\de{e_1}(\sigma)\cdot\unit(\sigma)} + \dem{C_2}{\de{e_2}(\sigma)\cdot\unit(\sigma)} \\
&= \de{e_1}(\sigma)\cdot\dem{C_1}{\unit(\sigma)} + \de{e_2}(\sigma)\cdot\dem{C_2}{\unit(\sigma)} \\
&= \de{e_1}(\sigma) \cdot\de{C_1}(\sigma) + \de{e_2}(\sigma)\cdot\de{C_2}(\sigma)
\end{align*}
By \Cref{lem:compatscale}, we know that this sum is defined, therefore the semantics is valid.
\end{proof}
For iteration, we can similarly use compatibility to ensure well-definedness.
\begin{lemma}
If $e$ and $e'$ are compatible and $\de{C}$ is a total function, then $\de{\iter C{e}{e'}}$ is a total function.
\end{lemma}
\begin{proof}
Let $\Phi_{\langle C, e, e'\rangle}(f)(\sigma) = \de{e}(\sigma)\cdot f^\dagger(\de{C}(\sigma)) + \de{e'}(\sigma)\cdot\unit(\sigma)$. Since $e$ and $e'$ are compatible, it follows from \Cref{lem:compatscale} that $\Phi_{\langle C, e, e'\rangle}$ is a total function. By \Cref{lem:phicont}, we therefore also know that $\de{\iter C{e}{e'}}$ is total.
\end{proof}
To conclude, we will provide a few examples of compatible expressions. For any test $b$, it is easy to see that $b$ and $\lnot b$ are compatible. This is because at any state $\sigma$, one of $\de{b}(\sigma)$ or $\de{\lnot b}(\sigma)$ must be $\zero$, and given the semiring laws, $\zero+u$ is defined for any $u\in U$. Given this, our encodings of if statements and while loops from \Cref{sec:sugar} are well-defined in all interpretations.

In the probabilistic interpretation (\Cref{ex:probsemi}), the weights $p$ and $1-p$ are compatible for any $p\in[0,1]$. That means that our encoding of probabilistic choice $C_1 +_p C_2$ and probabilistic iteration $C^{\langle p\rangle}$ are both well-defined too.

\section{Soundness and Completeness of Outcome Logic}
\label{app:soundcomplete}

We provide a formal definition of assertion entailment $\varphi\vDash e=u$, which informally means that $\varphi$ has enough information to determine that the expression $e$ evaluates to the value $u$.
\begin{definition}[Assertion Entailment]
Given an outcome assertion $\varphi$, an expression $e$, and a weight $u\in U$, we define the following:
\[
\varphi\vDash e=u
\qquad\text{iff}\qquad
\forall m\in\varphi, \sigma\in\supp(m).\quad
\de{e}(\sigma) = u
\]
\end{definition}
Note that this can generally be considered equivalent to $\varphi \Rightarrow \always(e = u)$, however we did not assume that equality is in the set of primitive tests.
Occasionally we will also write $\varphi\vDash b$ for some test $b$, which is shorthand for $\varphi\vDash b=\one$. It is relatively easy to see that the following statements hold given this definition:
\[
\arraycolsep=.25em
\begin{array}{rllll}
\top & \vDash & e=u \qquad\;& \text{iff}\qquad\; & \forall\sigma\in\Sigma.\quad\de{e}(\sigma) = u \\
\bot & \vDash & e=u & \multicolumn{2}{l}{\text{always}} \\
\varphi\lor\psi & \vDash & e=u & \text{iff} & \varphi\vDash e=u \quad\text{and}\quad \psi\vDash e=u \\
\varphi\land\psi & \vDash & e=u & \text{iff} & \varphi\vDash e=u \quad\text{or}\quad \psi\vDash e=u \\
\varphi\oplus\psi & \vDash & e=u & \text{iff} & \varphi\vDash e=u \quad\text{and}\quad \psi\vDash e=u \\
\varphi\odot{v} & \vDash & e=u & \text{iff} & v=\zero \quad\text{or}\quad \varphi\vDash e=u \\
v\odot\varphi & \vDash & e=u & \text{iff} & v=\zero \quad\text{or}\quad \varphi\vDash e=u \\
\ind m &\vDash& e=u & \text{iff} & \forall\sigma\in\supp(m).\quad \de{e}(\sigma) = u \\
\wg Pv & \vDash & e=u & \text{iff} & v=\zero\quad\text{or}\quad  \forall\sigma\in P.\quad \de{e}(\sigma) = u \\
\always P & \vDash & e=u & \text{iff} & P \vDash e=u
\end{array}
\]
We now present the soundness proof, following the sketch from \Cref{sec:soundness}. The first results pertain to the semantics of iteration. We start by recalling the characteristic function:
\[
\Phi_{\langle C, e, e'\rangle}(f)(\sigma) = \de{e}(\sigma)\cdot f^\dagger(\de{C}(\sigma)) + \de{e'}(\sigma)\cdot\unit(\sigma)
\]
Note that, as defined in \Cref{fig:denotation}, $\de{\iter C{e}{e'}}(\sigma) = \mathsf{lfp}\left(\Phi_{\langle C, e, e'\rangle}\right)(\sigma)$. The first lemma relates $\Phi_{\langle C, e, e\rangle}$ to a sequence of unrolled commands.
\begin{lemma}\label{lem:ifcharacteristic}
For all $n\in\mathbb N$:
\[
\Phi^{n+1}_{\langle C,e,e'\rangle}(\lambda x.\zero) = \sum_{k=0}^n \de{(\assume{e}\fatsemi C)^k\fatsemi \assume{e'}}
\]
\end{lemma}
\begin{proof}
By mathematical induction on $n$.
\begin{itemize}
\item $n=0$. Unfolding the definition of $\Phi_{\langle C,e_1,e_2\rangle}$, for all $\sigma\in\Sigma$ we get:
\begin{align*}
\Phi_{\langle C,e,e'\rangle}(\lambda x.\zero)(\sigma)
&= \de{e}(\sigma)\cdot (\lambda x.\zero)^\dagger(\de{C}(\sigma)) + \de{e'}(\sigma)\cdot\unit(\sigma) \\
&= \zero + \de{e'}(\sigma)\cdot \unit(\sigma) \\
&= \de{\assume{e'}}(\sigma) \\
&= \de{(\assume{e}\fatsemi C)^0\fatsemi \assume{e'}}(\sigma)
\end{align*}
\item Inductive step, suppose the claim holds for $n$. Now, for all $\sigma\in\Sigma$:
\begin{align*}
\Phi^{n+2}_{\langle C,e,e'\rangle}(\lambda x.\zero)(\sigma)
&= \de{e}(\sigma)\cdot \Phi^{n+1}_{(\langle b,C\rangle)}(\lambda x.\zero)^\dagger(\de{C}(\sigma)) + \de{e'}(\sigma)\cdot\unit(\sigma) \\
\intertext{By the induction hypothesis}
&= \de{e}(\sigma)\cdot \left(\sum_{k=0}^n\de{(\assume{e}\fatsemi C)^k\fatsemi\assume{e'}}\right)^\dagger\!\!\!(\de{C}(\sigma))
   + \de{e'}(\sigma)\cdot\unit(\sigma)\\
&= \sum_{k=1}^{n+1}\de{(\assume{e}\fatsemi C)^k\fatsemi\assume{e'}}(\sigma) + \de{\assume{e'}}(\sigma)\\
&= \sum_{k=0}^{n+1}\de{(\assume{e}\fatsemi C)^k\fatsemi\assume{e'}}(\sigma)
\end{align*}
\end{itemize}
\end{proof}

\lemwhilesem*
\begin{proof}
First, by the Kleene fixed point theorem and the program semantics (\Cref{fig:denotation}), we get:
\begin{align*}
\de{\iter C{e}{e'}}(\sigma)
&= \sup_{n\in\mathbb N} \Phi^n_{\langle C, e, e'\rangle} (\lambda x.\zero)(\sigma) \\
\intertext{Now, since $\Phi^0_{\langle C,e,e'\rangle} (\lambda x.\zero)(\sigma) =\zero$ and $\zero$ is the bottom of the order $\sqsubseteq$, we can rewrite the supremum as follows.}
&= \sup_{n\in\mathbb N} \Phi^{n+1}_{\langle C,e,e'\rangle} (\lambda x.\zero)(\sigma) \\
\intertext{By \Cref{lem:ifcharacteristic}:}
&= \sup_{n\in\mathbb N}\sum_{k=0}^n\de{(\assume{e}\fatsemi C)^k\fatsemi \assume{e'}}(\sigma)\\
\intertext{By the definition of infinite sums:}
&= \sum_{n\in\mathbb N}\de{(\assume{e}\fatsemi C)^n\fatsemi \assume{e'}}(\sigma)
\end{align*}
\end{proof}

\thmsoundness*
\begin{proof}
The triple $\triple\varphi{C}\psi$ is proven using inference rules from \Cref{fig:rules}, or by applying an axiom in $\Omega$. If the last step is using an axiom, then the proof is trivial since we assumed that all the axioms in $\Omega$ are semantically valid. If not, then the proof is by induction on the derivation $\Omega\vdash\triple\varphi{C}\psi$.
\begin{itemize}
\item \ruleref{Skip}. We need to show that $\vDash\triple\varphi\skp\varphi$.
Suppose that $m\vDash\varphi$. Since $\dem\skp m = m$, then clearly $\dem\skp m\vDash\varphi$. 

\item\ruleref{Seq}. Given that $\Omega\vdash\triple\varphi{C_1}\vartheta$ and $\Omega\vdash\triple\vartheta{C_2}\psi$, we need to show that $\vDash\triple\varphi{C_1\fatsemi C_2}\psi$.
Note that since $\Omega\vdash\triple\varphi{C_1}\vartheta$ and $\Omega\vdash\triple\vartheta{C_2}\psi$, those triples must be derived either using inference rules (in which case the induction hypothesis applies), or by applying an axiom in $\Omega$. In either case, we can conclude that $\vDash\triple\varphi{C_1}\vartheta$ and $\vDash\triple\vartheta{C_2}\psi$.
Suppose that $m\vDash\varphi$. Since $\vDash\triple\varphi{C_1}\vartheta$, we get that $\dem{C_1}m\vDash\vartheta$ and using $\vDash\triple\vartheta{C_2}\psi$, we get that $\dem{C_2}{\dem{C_1}m}\vDash\psi$. Since $\dem{C_2}{\dem{C_1}m} = (\de{C_2}^\dagger \circ \de{C_1})^\dagger(m) = \dem{C_1\fatsemi C_2}m$, we are done.

\item\ruleref{Plus}. Given $\Omega\vdash\triple\varphi{C_1}{\psi_1}$ and $\Omega\vdash\triple\varphi{C_2}{\psi_2}$, we need to show $\vDash\triple\varphi{C_1+ C_2}{\psi_1\oplus\psi_2}$. Suppose that $m\vDash\varphi$, so by the induction hypotheses, $\dem{C_1}m\vDash\psi_1$ and $\dem{C_2}m\vDash\psi_2$. Recall from the remark at the end of \Cref{sec:sugar} that we are assuming that programs are well-formed, and therefore $\dem{C_1+C_2}m$ is defined and it is equal to $\dem{C_1}m + \dem{C_2}m$. Therefore by the semantics of $\oplus$, $\dem{C_1+C_2}m \vDash \psi_1\oplus\psi_2$.

\item\ruleref{Assume}. Given $\varphi\vDash e=u$, we must show $\vDash\triple\varphi{\assume e}{\varphi\odot u}$. Suppose $m\vDash\varphi$. Since $\varphi\vDash e=u$, then $\de{e}(\sigma) =u$ for all $\sigma\in\supp(m)$. This means that:
\begin{align*}
\dem{\assume e}m
&= \smashoperator{\sum_{\sigma\in\supp(m)}} m(\sigma) \cdot \de{\assume e}(\sigma) \\
&= \smashoperator{\sum_{\sigma\in\supp(m)}} m(\sigma) \cdot \de{e}(\sigma) \cdot \eta(\sigma) \\
&= \smashoperator{\sum_{\sigma\in\supp(m)}} m(\sigma) \cdot u \cdot \eta(\sigma) \\
&= (\smashoperator{\sum_{\sigma\in\supp(m)}} m(\sigma) \cdot \eta(\sigma)) \cdot u \\
&=  m \cdot u
\end{align*}
And by definition, $m\cdot u\vDash \varphi\odot u$, so we are done.

\item\ruleref{Iter}. We know that $\vDash\triple{\varphi_n}{\assume{e}\fatsemi C}{\varphi_{n+1}}$ and that $\vDash\triple{\varphi_n}{\assume{e'}}{\psi_n}$ for all $n\in\mathbb N$ by the induction hypotheses. Now, we need to show that $\vDash\triple{\varphi_0}{\iter C{e}{e'}}{\psi_\infty}$. Suppose $m\vDash\varphi_0$. It is easy to see that for all $n\in\mathbb N$:
\[
\dem{(\assume{e}\fatsemi C)^n}m\vDash\varphi_n
\]
and
\[
\dem{(\assume{e}\fatsemi C)^n\fatsemi \assume{e'}}m\vDash\psi_n
\]
by mathematical induction on $n$, and the two induction hypotheses. Now, since $\conv\psi$, we also know that:
\[
\sum_{n\in\mathbb N}\dem{(\assume{e}\fatsemi C)^n\fatsemi \assume{e'}}m \vDash \psi_\infty
\]
Finally, by \Cref{lem:whilesem} we get that $\dem{\iter C{e}{e'}}m\vDash\psi_\infty$.

\item\ruleref{False}. We must show that $\vDash\triple\bot{C}\varphi$. Suppose that $m\vDash\bot$. This is impossible, so the claim follows vacuously.

\item\ruleref{True}. We must show that $\vDash\triple\varphi{C}\top$. Suppose $m\vDash\varphi$. It is trivial that $\dem Cm\vDash\top$, so the triple is valid.

\item\ruleref{Scale}. By the induction hypothesis, we get that $\vDash\triple\varphi{C}\psi$ and we must show that $\vDash\triple{u\odot\varphi}C{u\odot\psi}$. Suppose $m\vDash u\odot\varphi$. So there is some $m'$ such that $m'\vDash\varphi$ and $m = u\cdot m'$. We therefore get that $\dem C{m'}\vDash \psi$. Now, observe that $\dem Cm = \dem C{u\cdot m'} = u\cdot\dem Cm$. Finally, by the definition of $\odot$, we get that $u\cdot\dem Cm\vDash u\odot\psi$.

\item\ruleref{Disj}. By the induction hypothesis, we know that $\vDash\triple{\varphi_1}C{\psi_1}$ and $\vDash\triple{\varphi_2}C{\psi_2}$ and we need to show  $\vDash\triple{\varphi_1\vee\varphi_2}C{\psi_1\vee\psi_2}$. Suppose $m\vDash\varphi_1\vee\varphi_2$. Without loss of generality, suppose $m\vDash\varphi_1$. By the induction hypothesis, we get $\dem Cm\vDash\psi_1$. We can weaken this to conclude that $\dem Cm\vDash\psi_1\vee\psi_2$. The case where instead $m\vDash\varphi_2$ is symmteric.

\item\ruleref{Conj}. By the induction hypothesis, we get that $\vDash\triple{\varphi_1}C{\psi_1}$ and $\vDash\triple{\varphi_2}C{\psi_2}$ and we need to show  $\vDash\triple{\varphi_1\land\varphi_2}C{\psi_1\land\psi_2}$. Suppose $m\vDash\varphi_1\land\varphi_2$, so $m\vDash\varphi_1$ and $m\vDash\varphi_2$. By the induction hypotheses, $\dem Cm\vDash\psi_1$ and $\dem Cm\vDash\psi_2$, so $\dem Cm\vDash\psi_1\land\psi_2$.

\item\ruleref{Choice}. By the induction hypothesis, $\vDash\triple{\phi(t)}C{\phi'(t)}$ for all $t\in T$, and we need to show that $\vDash\triple{\bigoplus_{x\in T}\phi(x)}C{\bigoplus_{x\in T} \phi'(x)}$. Suppose $m\vDash \bigoplus_{x\in T}\phi(x)$, so for each $t\in T$ there is an $m_t$ such that $m_t \in \phi(t)$ and $m = \sum_{t\in T} m_t$. By the induction hypothesis, we get that $\dem C{m_t} \in \phi'(t)$ for all $t\in T$. Now, we have:
\begin{align*}
\sum_{t\in T} \dem C{m_t}
&= \sum_{t\in T} \sum_{\sigma\in\supp(m_t)} m_t(\sigma) \cdot \de{C}(\sigma)
\\
&= \sum_{\sigma\in\supp(m)} (\sum_{t\in T} m_t(\sigma)) \cdot \de{C}(\sigma)
\\
&= \sum_{\sigma\in\supp(m)} m(\sigma) \cdot \de{C}(\sigma)
\\
&= \dem Cm
\end{align*}
Therefore, we get that $\dem Cm \in \bigoplus_{x\in T}\phi'(x)$.

\item\ruleref{Exists}. By the induction hypothesis, $\vDash\triple{\phi(t)}C{\phi'(t)}$ for all $t\in T$ and we need to show $\vDash\triple{\exists x:T.\phi(x)}C{\exists x:T.\phi'(x)}$. Now suppose $m\in\exists x:T.\phi(x) = \bigcup_{t\in T}\phi(t)$. This means that there is some $t\in T$ such that $m\in\phi(t)$. By the induction hypothesis, this means that $\dem Cm\vDash\phi'(t)$, so we get that $\dem Cm\vDash\exists x:T.\phi'(x)$.

\item\ruleref{Consequence}. We know that $\varphi'\Rightarrow\varphi$ and $\psi\Rightarrow\psi'$ and by the induction hypothesis $\vDash\triple\varphi C\psi$, and we need to show that $\vDash\triple{\varphi'}C{\psi'}$. Suppose that $m\vDash\varphi'$, then it also must be the case that $m\vDash\varphi$. By the induction hypothesis, $\dem Cm\vDash\psi$. Now, using the second consequence $\dem Cm\vDash\psi'$.
\end{itemize}
\end{proof}
Now, moving to completeness, we prove the following lemma.
\lemcompleteness*
\begin{proof}
By induction on the structure of the program $C$.
\begin{itemize}

\item $C=\skp$. Since $\dem\skp{m} = m$ for all $m$, then clearly $\spost(\skp, \varphi) = \varphi$. We complete the proof by applying the \ruleref{Skip} rule.
\[
\inferrule*[right=\ruleref{Skip}]{\;}{\triple{\varphi}\skp{\varphi}}
\]

\item $C = C_1\fatsemi C_2$. First, observe that:
\begin{align*}
\spost(C_1 \fatsemi C_2, \varphi)
&= \{ \dem{C_1\fatsemi C_2}m \mid m\in\varphi \} \\
&= \{ \dem{C_2}{\dem{C_1}m} \mid m\in\varphi \} \\
&= \{ \dem{C_2}{m'} \mid m' \in \{ \dem{C_1}m \mid m\in\varphi \} \} \\
&= \spost(C_2, \spost(C_1, \varphi))
\end{align*}
Now, by the induction hypothesis, we know that:
\begin{align*}
&\Omega\vdash\triple{\varphi}{C_1}{\spost(C_1, \varphi)}
\\
&\Omega\vdash\triple{\spost(C_1, \varphi)}{C_2}{\spost(C_1\fatsemi C_2, \varphi)}
\end{align*}
Now, we complete the derivation as follows:
\[
\inferrule*[right=\ruleref{Seq}]{
  \inferrule*{\Omega}{
      \triple{\varphi}{C_1}{\spost(C_1,\varphi)}
    }
    \\
    \inferrule*{\Omega}{
      \triple{\spost(C_1, \varphi)}{C_2}{\spost(C_1\fatsemi C_2, \varphi)}
    }
}{
  \triple{\varphi}{C_1\fatsemi C_2}{\spost(C_1\fatsemi C_2, \varphi)}
}
\]

\item $C = C_1 +C_2$. So, we have that:
\begin{align*}
\spost(C_1+C_2, \varphi)
&= \{ \dem{C_1+C_2}m \mid m\in\varphi \} \\
&= \{ \dem{C_1}m + \dem{C_2}m \mid m\in\varphi \} \\
&= \bigcup_{m\in\varphi} \{ \dem{C_1}m +\dem{C_2}m \mid m \in \ind{m} \} \\
&= \exists m:\varphi.~ \spost(C_1, \ind m) \oplus \spost(C_2, \ind m)
\end{align*}
We now complete the derivation as follows:
\[
\inferrule*[right=\ruleref{Exists}]{
  \inferrule*[Right=\ruleref{Plus}]{
    \inferrule*{\Omega}{\triple{\ind m}{C_1}{\spost(C_1, \ind m)}}
    \\
    \inferrule*{\Omega}{\triple{\ind m}{C_2}{\spost(C_2, \ind m)}}
  }{
    \triple{\ind m}{C_1 +C_2}{\spost(C_1, \ind m) \oplus\spost(C_2, \ind m)}
  }
}{
  \triple{\varphi}{C_1+C_2}{\spost(C_1+C_2, \varphi)}
}
\]

\item $C = \assume e$, and $e$ must either be a test $b$ or a weight $u\in U$. Suppose that $e$ is a test $b$. Now, for any $m$ define the operator $b\? m$ as follows:
\[
b\? m(\sigma) = \left\{
\begin{array}{lll}
\zero & \text{if} & \de{b}(\sigma) = \zero \\
m(\sigma) \;\;& \text{if}\;\; & \de{b}(\sigma) = \one
\end{array}\right.
\]
Therefore $m = b\?m + \lnot b\? m$ and $\ind m = \ind{b\?m} \oplus \ind{\lnot b\?m}$. We also have:
\begin{align*}
\spost(\assume b, \varphi)
&= \{ \dem{\assume b}m \mid m\in\varphi \} \\
&= \{ \dem{\assume b}{b\? m + \lnot b\? m} \mid m \in\varphi \} \\
&= \{ b\? m \mid m \in\varphi \} \\
&= \exists m:\varphi.~ \ind{b\? m}
\end{align*}
Clearly also $\ind{b\? m} \vDash b$ and $\ind{\lnot b\? m}\vDash\lnot b$. We now complete the derivation: 
\[
\inferrule*[right=\ruleref{Exists}]{
  \inferrule*[Right=\ruleref{Choice}]{
    \inferrule*[right=\ruleref{Assume}]{
      \ind{b\? m}\vDash b = \one
    }{
      \triple{\ind{b\? m}}{\assume b}{{\ind{b\? m}}\odot\one}
    }
    \\
    \inferrule*[Right=\ruleref{Assume}]{
      \ind{\lnot b\? m}\vDash b = \zero
    }{
      \triple{\ind{\lnot b\? m}}{\assume b}{\ind{\lnot b\? m}\odot\zero}
    }  
  }{
    \triple{\ind{b\? m} \oplus \ind{\lnot b\? m}}{\assume b}{\ind{b\? m}}
  }
}{
  \triple\varphi{\assume b}{\spost(\assume b, \varphi)}
}
\]
Now, suppose $e=u$, so $\varphi\vDash u=u$ and $\dem{\assume u}m = m\cdot u$ for all $m\in\varphi$ and therefore $\spost(\assume u, \varphi) = \varphi\odot{u}$. We can complete the proof as follows:
\[
\inferrule*[right=\ruleref{Assume}]{
    \varphi\vDash u=u
  }{
    \triple{\varphi}{\assume u}{\varphi\odot{u}}
  }
\]

\item $C = \iter C{e}{e'}$. For all $n\in\mathbb N$, let $\varphi_n(m)$ and $\psi_n(m)$ be defined as follows:
\begin{align*}
\varphi_n(m) &\triangleq \spost((\assume e\fatsemi C)^n, \ind m) 
= \ind{\dem{(\assume e\fatsemi C)^n}m} \\
\psi_n(m) &\triangleq \spost(\assume{e'}, \varphi_n(m)) 
= \ind{\dem{(\assume e\fatsemi C)^n\fatsemi \assume{e'}}m} \\
\psi_\infty(m) &\triangleq \spost(\iter Ce{e'}, \ind m) 
= \ind{\dem{\iter Ce{e'}}m}
\end{align*}
Note that by definition, $\varphi_0(m) = \ind m$, $\varphi = \exists m:\varphi.\varphi_0(m)$, and $\spost(\iter Ce{e'}, \varphi) = \exists m:\varphi.\psi_\infty(m)$.

We now show that $(\psi_n)_{n\in\mathbb N^\infty}$ converges ($\conv{\psi}$). Take any $(m_n)_{n\in\mathbb N}$ such that $m_n\vDash\psi_n$ for each $n$. That means that $m_n = \dem{(\assume e\fatsemi C)^n\fatsemi \assume{e'}}m$. Therefore by \Cref{lem:whilesem} we get that $\sum_{n\in\mathbb N}m_n = \dem{\iter Ce{e'}}m$, and therefore $\sum_{n\in\mathbb N}m_n\vDash\psi_\infty(m)$. We now complete the derivation as follows:
\[
\inferrule*[right=\ruleref{Exists}]{
  \inferrule*[Right=\ruleref{Iter}]{
    \inferrule*{\Omega}{\triple{\varphi_n(m)}{\assume e\fatsemi C}{\varphi_{n+1}(m)}}
    \\
    \inferrule*{\Omega}{\triple{\varphi_n(m)}{\assume{e'}}{\psi_n(m)}}
  }{
    \triple{\varphi_0(m)}{\iter Ce{e'}}{\psi_\infty(m)}
  }
}{
  \triple{\varphi}{\iter Ce{e'}}{\spost(\iter Ce{e'}, \varphi)}
}
\]

\item $C=a$. We assumed that $\Omega$ contains all the valid triples pertaining to atomic actions $a\in\mathsf{Act}$, so $\Omega\vdash \triple{\varphi}a{\spost(a,\varphi)}$ since $\vDash \triple{\varphi}a{\spost(a, \varphi)}$.
\end{itemize}
\end{proof}

\section{Variables and State}
\label{app:state}

We now give additional definitions and proofs from \Cref{sec:state}. First, we give the interpretation of expressions $\de{E}_{\mathsf{Exp}} \colon \mathcal S \to \mathsf{Val}$ where $x\in\mathsf{Var}$, $v\in\mathsf{Val}$, and $b$ is a test.
\begin{align*}
\de{x}_\mathsf{Exp}(s) &\triangleq s(x) \\
\de{v}_\mathsf{Exp}(s) &\triangleq v \\
\de{b}_\mathsf{Exp}(s) &\triangleq \de{b}_\mathsf{Test}(s) \\
\de{E_1 + E_2}_\mathsf{Exp}(s) &\triangleq \de{E_1}_\mathsf{Exp}(s) + \de{E_2}_\mathsf{Exp}(s) \\
\de{E_1 - E_2}_\mathsf{Exp}(s) &\triangleq \de{E_1}_\mathsf{Exp}(s) - \de{E_2}_\mathsf{Exp}(s) \\
\de{E_1 \times E_2}_\mathsf{Exp}(s) &\triangleq \de{E_1}_\mathsf{Exp}(s) \cdot \de{E_2}_\mathsf{Exp}(s)
\end{align*}
Informally, the free variables of an assertion $P$ are the variables that are used in $P$. Given that assertions are semantic, we define $\mathsf{free}(P)$ to be those variables that $P$ constrains in some way. Formally, $x$ is free in $P$ iff reassigning $x$ to some value $v$ would not satisfy $P$.
\[
\mathsf{free}(P) \triangleq \{ x\in \mathsf{Var} \mid \exists s\in P, v\in\mathsf{Val}.~ s[x\mapsto v]\notin P \}
\]
The modified variables of a program $C$ are the variables that are assigned to in the program, determined inductively on the structure of the program.
\begin{align*}
\mathsf{mod}(\skp) &\triangleq \emptyset \\
\mathsf{mod}(C_1 \fatsemi C_2) &\triangleq \mathsf{mod}(C_1) \cup \mathsf{mod}(C_2) \\
\mathsf{mod}(C_1 + C_2) &\triangleq \mathsf{mod}(C_1) \cup \mathsf{mod}(C_2) \\
\mathsf{mod}(\assume e) &\triangleq \emptyset \\
\mathsf{mod}\left(\iter Ce{e'}\right) &\triangleq \mathsf{mod}(C) \\
\mathsf{mod}(x \coloneqq E) &\triangleq \{x\}
\end{align*}
Now, before the main soundness and completeness result, we prove a lemma stating that $\triple{\always P}C{\always P}$ is valid as long as $P$ does not contain information about variables modified by $C$.
\begin{lemma}\label{lem:constancy} If $\mathsf{free}(P) \cap \mathsf{mod}(C) = \emptyset$, then:
\[
\vDash\triple{\always P}C{\always P}
\]
\end{lemma}
\begin{proof}
By induction on the program $C$:
\begin{itemize}
\item $C=\skp$. Clearly the claim holds using \ruleref{Skip}.
\item $C=C_1\fatsemi C_2$. By the induction hypotheses, $\vDash\triple{\always P}{C_i}{\always P}$ for $i\in\{1,2\}$. We complete the proof using \ruleref{Seq}.
\item $C=C_1+C_2$. By the induction hypotheses, $\vDash\triple{\always P}{C_i}{\always P}$ for $i\in\{1,2\}$. We complete the proof using \ruleref{Plus} and the fact that $\always P \oplus \always P \Leftrightarrow \always P$.
\item $C = \assume e$. Since $\assume e$ can only remove states, it is clear that $\always P$ must still hold after running the program.
\item $C = \iter Ce{e'}$. The argument is similar to that of the soundness of \ruleref{Invariant}. Let $\varphi_n = \psi_n = \psi_\infty= \always P$. It is obvious that $\conv\psi$. We also know that $\vDash\triple{\always P}C{\always P}$ by the induction hypothesis. The rest is a straightforward application of the \ruleref{Iter} rule, also using the argument about $\mathsf{assume}$ from the previous case.
\item $C = x\coloneqq E$. We know that $x \notin \mathsf{free}(P)$, so for all $s \in P$ and $v\in\mathsf{Val}$, we know that $s[x\mapsto v] \in P$. We will now show that $P[E/x] = P$. Suppose $s\in P[E/x]$, this means that $s[x\mapsto \de{E}_{\mathsf{Exp}}(s)] \in P$, which also means that:
\[
(s[x\mapsto \de{E}_{\mathsf{Exp}}(s)])[x \mapsto s(x)] = s \in P
\]
Now suppose that $s \in P$, then clearly $s[x\mapsto \de{E}_{\mathsf{Exp}}(s)] \in P$, so $s \in P[E/x]$. Since $P[E/x] = P$, then $(\always P)[E/x]$ = $\always P$, so the proof follows from the \ruleref{Assign} rule.
\end{itemize}
\end{proof}
We now prove the main result. Recall that this result pertains specifically to the OL instance where variable assignment is the only atomic action.
\soundcompleteassign*
\begin{proof}\;
\iffcases{
Suppose $\vDash\triple\varphi{C}\psi$. By \Cref{thm:completeness}, we already know that this triple is derivable for all commands other than assignment so it suffices to show the case where $C = x\coloneqq E$.

Now suppose $\vDash\triple\varphi{x\coloneqq E}\psi$. For any $m\in\varphi$, we know that $(\lambda s.\eta(s[x\mapsto \de{E}(s)]))^\dagger(m) \in \psi$. By definition, this means that $m\in \psi[E/x]$, so we have shown that $\varphi \Rightarrow \psi[E/x]$. Finally, we complete the derivation as follows:
\[
\inferrule*[right=\ruleref{Consequence}]{
  \varphi\Rightarrow\psi[E/x]
  \\
  \inferrule*[Right=\ruleref{Assign}]{\;}{
    \triple{\psi[E/x]}{x \coloneqq E}{\psi}
  }
}{
  \triple{\varphi}{x\coloneqq E}{\psi}
}  
\]
}{
The proof is by induction on the derivation $\vdash\triple\varphi{C}\psi$. All the cases except for the two below follow from \Cref{thm:soundness}.
\begin{itemize}
\item\ruleref{Assign}.
Suppose that $m\vDash \varphi[E/x]$. By the definition of substitution, we immediately know that
\[(\lambda s.\eta(s[x\mapsto\de{E}(s)]))^\dagger(m) \in \varphi\]
Since $\de{x\coloneqq E}(s) = \eta(s[x\mapsto\de{E}(s)])$, we are done.
\item\ruleref{Constancy}. Follows immediately from \Cref{lem:constancy} and the soundness of the \ruleref{Conj} rule.
\end{itemize}
}
\end{proof}

\section{Subsumption of Program Logics}

In this section, we provide proofs for the theorems in \Cref{sec:dynamichoare}. Note that the following two theorems assume a nondeterministic interpretation of Outcome Logic, using the semiring defined in \Cref{ex:powersetsemi}.

\thmhoare*
\begin{proof} We only prove that $\vDash\triple PC{\always Q}$ iff $P\subseteq \dlbox CQ$, since $P\subseteq \dlbox CQ$ iff $\vDash\hoare PCQ$ is a well-known result \cite{pratt1976semantical}.
\iffcases{
Suppose $\sigma\in P$, then $\unit(\sigma) \vDash \sure P$ and since $\vDash\triple{\sure P}C{\always Q}$ we get that $\dem C{\unit(\sigma)}\vDash \always Q$, which is equivalent to $\de C(\sigma)\vDash\exists u:U.\wg Qu$. This means that $\supp(\de C(\sigma)) \subseteq Q$, therefore by definition $\sigma\in \dlbox CQ$. Therefore, we have shown that $P\subseteq \dlbox CQ$.
}{
Suppose that $P\subseteq \dlbox CQ$ and $m\vDash P$, so $|m| = \one$ and $\supp(m) \subseteq P \subseteq \dlbox CQ$. This means that $\supp(\de{C}(\sigma)) \subseteq Q$ for all $\sigma \in \supp(m)$, so we also get that:
\[
\supp(\dem Cm) = \smashoperator{\bigcup_{\sigma\in\supp(m)}}\supp(\de{C}(\sigma)) \subseteq Q
\]
This means that $\dem Cm\vDash \always Q$, therefore $\vDash\triple{\sure P}C{\always Q}$.
}
\end{proof}

\thmlisbon*
\begin{proof} We only prove that $\vDash\triple{\sure P}C{\sometimes Q}$ iff $P\subseteq \dldia CQ$, since $P\subseteq \dldia CQ$ iff $\vDash\lisbon PCQ$ follows by definition \cite{outcome,ilalgebra}.
\iffcases{
Suppose $\sigma\in P$, then $\unit(\sigma) \vDash \sure P$ and since $\vDash\triple{\sure P}C{\sometimes Q}$ we get that $\dem C{\unit(\sigma)}\vDash \sometimes Q$, which is equivalent to saying that there exists a $\tau \in \supp(\de{C}(\sigma))$ such that $\tau \in Q$, therefore by definition $\sigma \in \dldia CQ$.
So, we have shown that $P\subseteq \dldia CQ$.
}{
Suppose that $P\subseteq \dldia CQ$ and $m\vDash \sure P$, so $|m| = \one$ and $\supp(m) \subseteq P \subseteq \dldia CQ$. This means that $\supp(\de{C}(\sigma)) \cap Q \neq\emptyset$ for all $\sigma \in \supp(m)$. In other words, for each $\sigma\in\supp(m)$, there exists a $\tau \in \supp(\de{C}(\sigma))$ such that $\tau \in Q$. Since $\supp(\dem Cm) = \bigcup_{\sigma\in\supp(m)}\supp(\de{C}(\sigma))$, then there is also a $\tau \in\supp(\dem Cm)$ such that $\tau\in Q$, so $\dem Cm\vDash \sometimes Q$, therefore $\vDash\triple{\sure P}C{\sometimes Q}$.
}
\end{proof}

\section{Derived Rules}
\label{app:derived}

\subsection{Sequencing in Hoare and Lisbon Logic}
\label{app:seq}

We first prove the results about sequencing Hoare Logic encodings.
\begin{lemma}\label{lem:liftbox}
The following inference is derivable.
\[
\inferrule{
\triple{\sure P}C{\always Q}
}{
  \triple{\always P}C{\always Q}
}
\]
\end{lemma}
\begin{proof}
We first establish two logical consequences. First, we have:
\begin{align*}
\always P & \implies \exists m:\always P.~ \ind m \\
&\implies \exists m:\always P.~ \textstyle\bigoplus_{\sigma\in\supp(m)}  m(\sigma)\odot \ind{\eta(\sigma)} \\
&\implies \exists m:\always P.~ \textstyle\bigoplus_{\sigma\in\supp(m)} m(\sigma)\odot\sure P
\end{align*}
And also:
\begin{align*}
\exists m:\always P.~ \textstyle\bigoplus_{\sigma\in\supp(m)}  m(\sigma)\odot \always Q
&\implies \exists m:\always P.~ \textstyle\bigoplus_{\sigma\in\supp(m)} \always Q \\
&\implies \exists m:\always P.~ \always Q \\
&\implies \always Q
\end{align*}
We now complete the derivation.
\[
\inferrule*[right=\ruleref{Consequence}]{
  \inferrule*[Right=\ruleref{Exists}]{
    \forall m\in \always P.
    \\
    \inferrule*[Right=\ruleref{Choice}]{
      \forall \sigma\in\supp(m).
      \\
      \inferrule*[Right=\ruleref{Scale}]{
        \triple{\sure P}C{\always Q}
      }{
        \triple{m(\sigma)\odot \sure{P}}C{m(\sigma)\odot\always Q}
      }
    }{
      \triple{\textstyle\bigoplus_{\sigma\in\supp(m)} m(\sigma)\odot \sure P}C{\textstyle\bigoplus_{\sigma\in\supp(m)} m(\sigma)\odot \always Q}
    }
  }{
    \triple{\exists m:\always P.~ \textstyle{\bigoplus_{\sigma\in\supp(m)}} m(\sigma)\odot \sure P}C{\exists m:\always P.~ \textstyle\bigoplus_{\sigma\in\supp(m)} m(\sigma)\odot\always Q}
  }
}{
  \triple{\always P}C{\always Q}
}
\]
%
\end{proof}

\begin{lemma}
The following inference is derivable.
\[
\inferrule{
  \triple{\sure P}{C_1}{\always Q}
  \\
  \triple{\sure Q}{C_2}{\always R}
}{
  \triple{\sure P}{C_1\fatsemi C_2}{\always R}
}
\]
\end{lemma}
\begin{proof}
\[
\inferrule*[right=\ruleref{Seq}]{
  \triple{\sure P}{C_1}{\always Q}
  \\
  \inferrule*[Right=\Cref{lem:liftbox}]{
    \triple{\sure Q}{C_2}{\always R}
  }{
    \triple {\always Q}{C_2}{\always R}
  }
}{
\triple{\sure P}{C_1\fatsemi C_2}{\always R}
}
\]
\end{proof}
\noindent Now, we turn to Lisbon Logic

\begin{lemma}\label{lem:liftdiamond}
The following inference is derivable.
\[
\inferrule{
\triple{\sure P}C{\sometimes Q}
}{
  \triple{\sometimes P}C{\sometimes Q}
}
\]
\end{lemma}
\begin{proof}
First, for each $m\in\sometimes P$, let $\sigma_m$ be an arbitrary state such that $\sigma_m \in\supp(m)$ and $\sigma_m\in P$. This state must exists given that $m\in\sometimes P$. We also let $u_m = m(\sigma_m)$ and note that $u_m \neq \zero$. We therefore have the following consequences:
\begin{align*}
\sometimes P
&\implies \exists m:\sometimes P.~\ind m \\
&\implies \exists m:\sometimes P.~(u_m\odot \ind{\eta(\sigma_m)} ) \oplus\top \\
&\implies \exists m:\sometimes P.~(u_m\odot \sure P) \oplus\top \\
\\
\exists m:\sometimes P.~(u_m\odot \sometimes Q) \oplus\top
&\implies \exists m:\sometimes P.~\sometimes Q \oplus\top \\
&\implies \exists m:\sometimes P.~\sometimes Q \\
&\implies \sometimes Q
\end{align*}
We now complete the derivation as follows:
\[
\inferrule*[right=\ruleref{Consequence}]{
  \inferrule*[Right=\ruleref{Exists}]{
    \forall m\in \sometimes P.
    \\
    \inferrule*[Right=\ruleref{Choice}]{
      \inferrule*[right=\ruleref{Scale}]{
        \triple{\sure P}C{\sometimes Q}
      }{
        \triple{u_m\odot \sure P}C{u_m\odot \sometimes Q}
      }
      \\
      \inferrule*[Right=\ruleref{True}]{\;}{
        \triple{\top}C\top
      }
    }{
      \triple{(u_m\odot \sure P) \oplus \top}C{(u_m\odot \sometimes Q)\oplus\top}
    }
  }{
    \triple{\exists m:\sometimes P.~ (u_m\odot \sure{P}) \oplus \top}C{\exists m:\sometimes P.~ (u_m\odot\sometimes Q) \oplus\top}
  }
}{
  \triple{\sometimes P}C{\sometimes Q}
}
\]
\end{proof}
\begin{lemma}
The following inference is derivable.
\[
\inferrule{
  \triple{\sure P}{C_1}{\sometimes Q}
  \\
  \triple{\sure Q}{C_2}{\sometimes R}
}{
  \triple{\sure P}{C_1\fatsemi C_2}{\sometimes R}
}
\]
\end{lemma}
\begin{proof}
\[
\inferrule*[right=\ruleref{Seq}]{
  \triple{\sure P}{C_1}{\sometimes Q}
  \\
  \inferrule*[Right=\Cref{lem:liftdiamond}]{
    \triple{\sure Q}{C_2}{\sometimes R}
  }{ 
    \triple {\sometimes Q}{C_2}{\sometimes R}
  }
}{
\triple{\sure P}{C_1\fatsemi C_2}{\sometimes R}
}
\]
\end{proof}

\subsection{If Statements and While Loops}

\begin{lemma}
The following inference is derivable.
\[
\inferrule{
  \varphi_1\vDash b
  \\
  \triple{\varphi_1}{C_1}{\psi_1}
  \\
  \varphi_2 \vDash\lnot b
  \\
    \triple{\varphi_2}{C_2}{\psi_2}
}{
  \triple{\varphi_1 \oplus\varphi_2}{\iftf b{C_1}{C_2}}{\psi_2\oplus\psi_2}
}{\ruleref{If}}
\]
\end{lemma}
\begin{proof}
First note that $\varphi \vDash b$ is syntactic sugar for $\varphi\vDash b=\one$, and so from the assumptions that $\varphi_1\vDash b$ and $\varphi_2\vDash\lnot b$, we get $\varphi_1\vDash b=\one$, $\varphi_2\vDash b=\zero$, $\varphi_1\vDash \lnot b = \zero$, and $\varphi_2\vDash\lnot b =\one$.
We split the derivation into two parts. Part (1) is shown below:
\[
\inferrule*[right=\ruleref{Seq}]{
  \inferrule*[right=\ruleref{Choice}]{
    \inferrule*[right=\ruleref{Assume}]{
      \varphi_1 \vDash b = \one
    }{
      \triple{\varphi_1}{\assume b}{{\varphi_1}\odot\one}
    }
    \\
    \inferrule*[Right=\ruleref{Assume}]{
      \varphi_2 \vDash b = \zero
    }{
      \triple{\varphi_2}{\assume b}{\varphi_2\odot\zero}
    }
  }{
    \triple{\varphi_1\oplus\varphi_2}{\assume b}{\varphi_1}
  }
  \\
  \triple{\varphi_1}{C_1}{\psi_1}
}{
  \triple{\varphi_1\oplus\varphi_2}{\assume b\fatsemi C_1}{\psi_1}
}
\]
We omit the proof with part (2), since it is nearly identical.
Now, we combine (1) and (2):
\[
\inferrule*{
  \inferrule*[Right=\ruleref{Plus}]{
    \inferrule*{(1)
    }{
      \triple{\varphi_1\oplus\varphi_2}{\assume b\fatsemi C_1}{\psi_1}
    }
    \\
    \inferrule*{(2)
    }{
      \triple{\varphi_1\oplus\varphi_2}{\assume{\lnot b}\fatsemi C_2}{\psi_2}
    }
  }{
    \triple{\varphi_1\oplus\varphi_2}{(\assume b\fatsemi C_1) + (\assume{\lnot b}\fatsemi C_2)}{\psi_1\oplus\psi_2}
  }
}{
  \triple{\varphi_1 \oplus\varphi_2}{\iftf b{C_1}{C_2}}{\psi_1\oplus\psi_2}
}
\]
\end{proof}

\begin{lemma}
The following inference is derivable:
\[
\inferrule{
  \varphi \vDash b
  \\
  \triple{\varphi}{C_1}{\psi}
}{
  \triple{\varphi}{\iftf b{C_1}{C_2}}{\psi}
}{\ruleref{If1}}
\]
\end{lemma}
\begin{proof}
\[
\inferrule*[right=\ruleref{Consequence}]{
  \inferrule*[Right=\ruleref{If}]{
    \varphi\vDash b
    \\
    \triple{\varphi}{C_1}{\psi}
    \\
    \zero\odot\top \vDash\lnot b
    \\
    \inferrule*[Right=\ruleref{Scale}]{
      \inferrule*[Right=\ruleref{True}]{\;}{
        \triple{\top}{C_2}{\top}
      }
    }{
      \triple{\zero\odot\top}{C_2}{\zero\odot\top}
    }
  }{
    \triple{\varphi \oplus (\zero\odot\top)}{\iftf b{C_1}{C_2}}{\psi\oplus (\zero\odot\top)}
  }
}{
  \triple{\varphi}{\iftf b{C_1}{C_2}}{\psi}
}
\]
\end{proof}

\begin{lemma}
The following inference is derivable:
\[
\inferrule{
  \varphi \vDash \lnot b
  \\
  \triple{\varphi}{C_2}{\psi}
}{
  \triple{\varphi}{\iftf b{C_1}{C_2}}{\psi}
}{\ruleref{If2}}
\]
\end{lemma}
\begin{proof}
\[
\inferrule*[right=\ruleref{Consequence}]{
  \inferrule*[Right=\ruleref{If}]{
    \zero\odot\top \vDash b
    \\
    \inferrule*[right=\ruleref{Scale}]{
      \inferrule*[Right=\ruleref{True}]{\;}{
        \triple{\top}{C_1}{\top}
      }
    }{
      \triple{\zero\odot\top}{C_1}{\zero\odot\top}
    }
    \\
    \varphi\vDash \lnot b
    \\
    \triple{\varphi}{C_2}{\psi}
  }{
    \triple{\varphi \oplus (\zero\odot\top)}{\iftf b{C_1}{C_2}}{\psi\oplus (\zero\odot\top)}
  }
}{
  \triple{\varphi}{\iftf b{C_1}{C_2}}{\psi}
}
\]
\end{proof}

\begin{lemma}[Hoare Logic If Rule]\label{lem:ifhoare}
The following inference is derivable.
\[
\inferrule{
  \triple{\sure{P\land b}}{C_1}{\always Q}
  \\
  \triple{\sure{P\land\lnot b}}{C_2}{\always Q}
}{
  \triple{\sure P}{\iftf b{C_1}{C_2}}{\always Q}
}{\ruleref{If (Hoare)}}
\]
\end{lemma}
\begin{proof} The derivation is shown below:
\[
\inferrule*[right=\ruleref{Consequence}]{
  \inferrule*[Right=\ruleref{If}]{
    \always (P\land b) \vDash b
    \hspace{-2.5em}
    \inferrule*[Right=\Cref{lem:liftbox},vdots=3em]{
      \triple{\sure{P\land b}}{C_1}{\always Q}
    }{
      \triple{\always(P\land b)}{C_1}{\always Q}
    }
    \hspace{-2.5em}
    \always (P\land \lnot b) \vDash \lnot b
    \\
    \inferrule*[right=\Cref{lem:liftbox}]{
      \triple{\sure{P\land \lnot b}}{C_1}{\always Q}
    }{
      \triple{\always(P\land \lnot b)}{C_2}{\always Q}
    }
  }{
    \triple{\always (P\land b) \oplus \always(P\land\lnot b)}{\iftf b{C_1}{C_2}}{\always Q \oplus \always Q}
  }
}{
  \triple{\sure P}{\iftf b{C_1}{C_2}}{\always Q}
}
\]

\end{proof}


\begin{lemma}[Lisbon Logic If Rule]\label{lem:iflisbon}
The following inference is derivable.
\[
\inferrule{
  \triple{\sure{P\land b}}{C_1}{\sometimes Q}
  \\
  \triple{\sure{P\land\lnot b}}{C_2}{\sometimes Q}
}{
  \triple{\sure P}{\iftf b{C_1}{C_2}}{\sometimes Q}
}{\ruleref{If (Lisbon)}}
\]
\end{lemma}
\begin{proof}
The derivation is shown below:
\[
\inferrule*[right=\ruleref{Consequence}]{
  \inferrule*[Right=\ruleref{Disj}]{
    \inferrule*[right=\Cref{lem:liftdiamond}]{
      \inferrule*[Right=\ruleref{If1}]{
        \sure{P\land b} \vDash b
        \\
        \triple{\sure{P\land b}}{C_1}{\sometimes Q}
      }{
        \triple{\sure{P\land b}}{\iftf b{C_1}{C_2}}{\sometimes Q}
      }
    }{
      \triple{\sometimes(P\land b)}{\iftf b{C_1}{C_2}}{\sometimes Q}
    }
    \hspace{-6em}
    \inferrule*[Right=\Cref{lem:liftdiamond},vdots=5em]{
      \inferrule*[Right=\ruleref{If2}]{
        \sure{P\land \lnot b} \vDash \lnot b
        \\
        \triple{\sure{P\land \lnot b}}{C_2}{\sometimes Q}
      }{
        \triple{\sure{P\land\lnot b}}{\iftf b{C_1}{C_2}}{\sometimes Q}
      }
    }{
      \triple{\sometimes(P\land\lnot b)}{\iftf b{C_1}{C_2}}{\sometimes Q}
    }
  }{
    \triple{\sometimes(P\land b) \vee \sometimes(P\land\lnot b)}{\iftf b{C_1}{C_2}}{\sometimes Q \vee \sometimes Q}
  }
}{
  \triple{\sure P}{\iftf b{C_1}{C_2}}{\sometimes Q}
}
\]

\end{proof}

\begin{lemma}[While Rule]
The following inference is derivable.
\[
\inferrule{
  \conv{\psi}
  \quad
  \triple{\varphi_n}{C}{\varphi_{n+1} \oplus \psi_{n+1}}
  \quad \varphi_n \vDash b
  \quad \psi_n\vDash \lnot b
}{
  \triple{\varphi_0 \oplus \psi_0}{\whl bC}{\psi_\infty}
}
{\ruleref{While}}
\]
\end{lemma}
\begin{proof}
First, let $\varphi'_n = \varphi_n \oplus \psi_n$. The derivation has two parts. First, part (1):
\[
\inferrule*[right=\ruleref{Seq}]{
  \inferrule*[right=\ruleref{Choice}]{
    \inferrule*[right=\ruleref{Assume}]{
      \varphi_n\vDash b
    }{
      \triple{\varphi_n}{\assume b}{\varphi_n}
    }
    \\
    \inferrule*[Right=\ruleref{Assume}]{
      \psi_n\vDash \lnot b
    }{
      \triple{\psi_n}{\assume b}{{\psi_n}\odot\zero}
    }
  }{
    \triple{\varphi_n\oplus\psi_n}{\assume b}{\varphi_n}
  }
  \\
  \triple{\varphi_n}{C}{\varphi_{n+1} \oplus\psi_{n+1}}
}{
  \triple{\varphi_n\oplus\psi_n}{\assume b\fatsemi C}{\varphi_{n+1} \oplus\psi_{n+1}}
}
\]
Now part (2):
\[
\inferrule*[right=\ruleref{Choice}]{
  \inferrule*[right=\ruleref{Assume}]{
    \varphi_n\vDash b
  }{
    \triple{\varphi_n}{\assume{\lnot b}}{{\varphi_n}\odot\zero}
  }
  \\
  \inferrule*[Right=\ruleref{Assume}]{
    \psi_n\vDash \lnot b
  }{
    \triple{\psi_n}{\assume{\lnot b}}{\psi_n}
  }
}{
  \triple{\varphi_n\oplus\psi_n}{\assume{\lnot b}}{\psi_n}
}
\]
Finally, we complete the derivation as follows:
\[
  \inferrule*[right=\ruleref{Iter}]{
    \forall n\in\mathbb{N}.
    \\
    \inferrule*{(1)}{
      \triple{\varphi_n\oplus\psi_n}{\assume b\fatsemi C}{\varphi_{n+1} \oplus\psi_{n+1}}
    }
    \\
    \inferrule*{(2)}{
      \triple{\varphi_n\oplus\psi_n}{\assume{\lnot b}}{\psi_n}
    }
}{
  \triple{\varphi_0 \oplus \psi_0}{\whl bC}{\psi_\infty}
}
\]
\end{proof}

\subsection{Loop Invariants}

\begin{lemma}[Loop Invariant Rule] The following inference is derivable.
\[
\inferrule{
  \triple{\sure{P\land b}}C{\always P}
}{
  \triple{\sure P}{\whl bC}{\always(P \land \lnot b)}
}
{\ruleref{Invariant}}
\]
\end{lemma}
\begin{proof}
We will derive this rule using the \ruleref{While} rule.
For all $n$, let $\varphi_n = \always(P\land b)$ and $\psi_n = \always(P\land \lnot b)$. We will now show that $(\psi_n)_{n\in\mathbb N^\infty}$ converges. Suppose that $m_n\vDash \always(P\land \lnot b)$ for each $n\in\mathbb N$. That means that $\supp(m_n) \subseteq P \land \lnot b$. We also have that $\supp(\sum_{n\in\mathbb N} m_n) = \bigcup_{n\in\mathbb N} \supp(m_n)$, and since for each $n\in\mathbb N$, $\supp(m_n) \subseteq P \land \lnot b$, then $\bigcup_{n\in\mathbb N} \supp(m_n) \subseteq P\land\lnot b$, and thus $\sum_{n\in\mathbb N} m_n\vDash\always(P\land\lnot b)$. We remark that $\always(P\land b)\vDash b$ and $\always(P\land\lnot b)\vDash \lnot b$ trivially. We complete the derivation as follows:
\[
\inferrule*[right=\ruleref{Consequence}]{
  \inferrule*[Right=\ruleref{While}]{
    \inferrule*[Right=\ruleref{Consequence}]{
      \inferrule*[Right=\Cref{lem:liftbox}]{
        \triple{\sure{P\land b}}C{\always P}
      }{
        \triple{\always (P\land b)}C{\always P}
      }
    }{
      \triple{\always (P\land b)}C{\always(P\land b) \oplus \always (P\land\lnot b)}
    }
  }{
    \triple{\always (P\land b) \oplus \always(P\land\lnot b)}{\whl bC}{\always(P\land \lnot b)}
  }
}{
  \triple{\sure P}{\whl bC}{\always(P\land\lnot b)}
}
\]
\end{proof}

\subsection{Loop Variants}

\begin{lemma}\label{lem:variant}
The following inference is derivable.
\[
\inferrule{
  \forall n <N.
  \\
  \varphi_0 \vDash \lnot b
  \\
  \varphi_{n+1}\vDash b
  \\
  \triple{\varphi_{n+1}}C{\varphi_n}
}{
  \triple{\exists n:\mathbb N.\varphi_n}{\whl bC}{\varphi_0}
}{\ruleref{Variant}}
\]
\end{lemma}
\begin{proof}
For the purpose of applying the \ruleref{While} rule, we define the following for all $n$ and $N$:
\[\arraycolsep=4pt
\varphi'_n = \left\{\begin{array}{lll}
\varphi_{N - n} & \text{if} & n < N \\
\zero\odot\top & \text{if} & n \ge N
\end{array}\right.
\qquad
\psi_n = \left\{\begin{array}{lll}
\varphi_0 & \text{if} & n\in\{N,\infty\} \\ 
\zero\odot\top & \multicolumn{2}{l}{\text{otherwise}}
\end{array}\right.
\]
It is easy to see that $\conv\psi$. Each $\psi_n$ except for $\psi_N$ and $\psi_\infty$ is only satisfied by $\zero$, so taking $(m_n)_{n\in\mathbb N}$ such that $m_n\vDash\psi_n$ for each $n\in\mathbb N$, it must be the case that $\sum_{n\in\mathbb N}m_n = m_N$. By assumption, we know that $m_N \vDash \psi_\infty$ since $\psi_\infty = \psi_N = \varphi_0$. 
We also know that $\varphi'_n\vDash b$ and $\psi_n\vDash\lnot b$ by our assumptions and the fact that $\zero \odot \top\vDash e=u$ for any $e$ and $u$.
There are two cases for the premise of the \ruleref{While} rule (1) where $n<N$ (left) and $n\ge N$ (right).
\[
\inferrule*{
  \inferrule*{
    \forall m<N.\;
    \triple{\varphi_{m+1}}C{\varphi_m}
  }{
    \forall n<N.\;
    \triple{\varphi_{N-n}}C{\varphi_{N-(n+1)}}
  }
}{
  \forall n<N.\;
  \triple{\varphi'_n}{C}{\varphi'_{n+1} \oplus \psi_{n+1}}
}
\qquad\qquad
\inferrule*{
  \inferrule*[Right=\ruleref{Scale}]{
    \inferrule*[Right=\ruleref{True}]{\;}{
      \triple{\top}C{\top}
    }
  }{
    \triple{\zero\odot\top}C{\zero\odot\top}
  }
}{
  \forall n\ge N.\;
  \triple{\varphi'_n}{C}{\varphi'_{n+1} \oplus \psi_{n+1}}
}
\]
Finally, we complete the derivation.
\[
\inferrule*[right=\ruleref{Exists}]{
  \forall N\in\mathbb N.
  \quad
  \inferrule*[Right=\ruleref{While}]{
    \inferrule*{(1)
    }{
      \triple{\varphi'_n}{C}{\varphi'_{n+1} \oplus \psi_{n+1}}
    }
  }{
    \triple{\varphi_N}{\whl bC}{\varphi_0}
  }
}{
  \triple{\exists n:\mathbb N.\varphi_n}{\whl bC}{\varphi_0}
}
\]
\end{proof}

\begin{lemma}[Lisbon Logic Loop Variants]
The following inference is derivable.
\[
\inferrule{
  \forall n\in\mathbb N.
  \quad
  \sure{P_0} \vDash \lnot b
  \quad
  \sure{P_{n+1}}\vDash b
  \quad
  \triple{\sure{P_{n+1}}}C{\sometimes P_n}
}{
  \triple{\exists n:\mathbb N. \sure{P_n}}{\whl bC}{\sometimes P_0}
}{\ruleref{Lisbon Variant}}
\]
\end{lemma}
\begin{proof}
First, for all $N\in \mathbb N$, let $\varphi_n$ and $\psi_n$ be defined as follows:
\[\arraycolsep=4pt
\varphi_n = \left\{\begin{array}{lll}
\sometimes P_{N-n} & \text{if} & n \le N \\
 \top & \text{if} & n > N
\end{array}\right.
\qquad
\psi_n = \left\{\begin{array}{lll}
\sometimes P_0 & \text{if} & n\in\{N,\infty\} \\ 
\top & \multicolumn{2}{l}{\text{otherwise}}
\end{array}\right.
\]
Now, we prove that $\conv\psi$. Take any $(m_n)_{n\in\mathbb N}$ such that $m_n\vDash\psi_n$ for each $n\in\mathbb N$. Since $m_N\vDash\sometimes P_0$, then there is some $\sigma\in\supp(m_N)$ such that $\sigma\in P_0$. By definition $(\sum_{n\in\mathbb N} m_n)(\sigma)\ge m_N(\sigma) > \zero$, so $\sum_{n\in\mathbb N} m_n\vDash\sometimes P_0$ as well.
We opt to derive this rule with the \ruleref{Iter} rule rather than \ruleref{While} since it is inconvenient to split the assertion into components where $b$ is true and false.
We complete the derivation in two parts, and each part is broken into two cases. We start with (1), and the case where $n < N$. In this case, we know that $\varphi_n = \sometimes P_{N-n}$ and $\varphi_{n+1} = \sometimes P_{N-n-1}$ (even if $n = N-1$, then we get $\varphi_N = \sometimes P_0 = \sometimes P_{N-(N-1) -1}$).
\[\hspace{-4em}
\inferrule*{
  \inferrule*[Right=\Cref{lem:liftdiamond}]{
    \inferrule*[Right=\ruleref{Seq}]{
      \inferrule*[right=\ruleref{Assume}]{
        \sure{P_{N-n}}\vDash b
      }{
        \triple{\sure{P_{N-n}}}{\assume b}{\sure{P_{N-n}}}
      }
      \\
        \triple{\sure{P_{N-n}}}C{\sometimes P_{N-n-1}}
    }{
      \triple{\sure{P_{N-n}}}{\assume b\fatsemi C}{\sometimes P_{N-n-1}}
    }
  }{
    \triple{\sometimes P_{N-n}}{\assume b\fatsemi C}{\sometimes P_{N-n-1}}
  }
}{
  \triple{\varphi_n}{\assume b\fatsemi C}{\varphi_{n+1}}
}
\]
Now, we prove (1) where $n\ge N$, $\varphi_{n+1} = \top$.
\[
\inferrule*{
  \inferrule*[Right=\ruleref{True}]{
  }{
    \triple{\varphi_n}{\assume b\fatsemi C}{\top}
  }
}{
  \triple{\varphi_n}{\assume b\fatsemi C}{\varphi_{n+1}}
}
\]
We now move on to (2) below. On the left, $n = N$, and so $\varphi_n = \sometimes P_0$ and $\psi_n = \sometimes P_0$. On the right, $n \neq N$, so $\psi_n = \top$.
\[
\inferrule*{
  \inferrule*[Right=\Cref{lem:liftdiamond}]{
    \inferrule*[Right=\ruleref{Consequence}]{
      \inferrule*[Right=\ruleref{Assume}]{
        \sure{P_0}\vDash \lnot b
      }{
        \triple{\sure{P_0}}{\assume{\lnot b}}{\sure{P_0}}
      }
    }{
      \triple{\sure{P_0}}{\assume{\lnot b}}{\sometimes P_0}
    }
  }{
    \triple{\sometimes P_0}{\assume{\lnot b}}{\sometimes P_0}
  }
}{
  \triple{\varphi_n}{\assume{\lnot b}}{\psi_n}
}
\qquad\qquad\qquad\qquad
\inferrule*{
  \inferrule*[Right=\ruleref{True}]{\;}{
    \triple{\varphi_n}{\assume{\lnot b}}{\top}
  }
}{
  \triple{\varphi_n}{\assume{\lnot b}}{\psi_n}
}
\]
Finally, we complete the derivation using the \ruleref{Iter} rule.
\[
\inferrule*[right=\ruleref{Exists}]{
  \forall N\in\mathbb N. \quad 
  \inferrule*[Right=\ruleref{Consequence}]{
    \inferrule*{
      \inferrule*[Right=\ruleref{Iter}]{
        \inferrule*{(1)}{
          \triple{\varphi_n}{\assume b\fatsemi C}{\varphi_{n+1}}
        }
        \\
        \inferrule*{(2)}{
          \triple{\varphi_n}{\assume{\lnot b}}{\psi_n}
        }
      }{
        \triple{\varphi_0}{\iter C{b}{\lnot b}}{\psi_\infty}
      }
    }{
      \triple{\sometimes P_N}{\whl bC}{\sometimes P_0}
    }
  }{
      \triple{\sure{P_N}}{\whl bC}{\sometimes P_0}
  }
}{
  \triple{\exists n:\mathbb N.\sure{P_n}}{\whl bC}{\sometimes P_0}
}
\]
\end{proof}

\section{Hyper Hoare Logic}
\label{app:hyperhoare}

In this section, we present the omitted proof related the derived rules in Hyper Hoare Logic from \Cref{sec:hyper}. First, we give the inductive definition of negation for syntactic assertions below:
\begin{align*}
  \lnot(\varphi\land\psi) &\triangleq \lnot\varphi \lor\lnot\psi
  \\
  \lnot(\varphi\lor\psi) &\triangleq \lnot\varphi \land\lnot\psi
  \\
  \lnot(\forall x\colon\mathsf{Val}.\varphi) &\triangleq \exists x\colon\mathsf{Val}.\lnot\varphi
  \\
  \lnot(\exists \langle\sigma\rangle.\varphi) &\triangleq \forall x\colon\mathsf{Val}.\lnot\varphi
  \\
  \lnot(\forall \langle\sigma\rangle.\varphi) &\triangleq \exists \langle\sigma\rangle.\lnot\varphi
  \\
  \lnot(\exists \langle\sigma\rangle.\varphi) &\triangleq \forall \langle\sigma\rangle.\lnot\varphi
  \\
  \lnot(B) &\triangleq \lnot B
\end{align*}

\subsection{Variable Assignment}

For convenience, we repeat the definition of the syntactic variable assignment transformation.
\begin{align*}
  \mathcal A_x^E[\varphi\land \psi] & \triangleq \mathcal A_x^E[\varphi] \land \mathcal A_x^E[\psi]
  \\
  \mathcal A_x^E[\varphi\lor \psi] & \triangleq \mathcal A_x^E[\varphi] \lor \mathcal A_x^E[\psi]
  \\
  \mathcal A_x^E[\forall y: T.\varphi] & \triangleq \forall y: T. \mathcal A_x^E[\varphi]
  \\
  \mathcal A_x^E[\exists y: T.\varphi] & \triangleq  \exists y:\ T. \mathcal A_x^E[\varphi]
  \\
  \mathcal A_x^E[\forall \langle\sigma\rangle.\varphi] & \triangleq  \forall \langle\sigma\rangle. \mathcal A_x^E[\varphi[E[\sigma]/\sigma(x)]]
  \\
  \mathcal A_x^E[\exists \langle\sigma\rangle.\varphi] & \triangleq  \exists \langle\sigma\rangle. \mathcal A_x^E[\varphi[E[\sigma]/\sigma(x)]]
  \\
  \mathcal A_x^E[B] & \triangleq  B
\end{align*}

\begin{lemma}\label{lem:syntactic-subst}
\[
  \mathcal A_x^E[\varphi] = \{ m \in \mathcal W(\mathcal S) \mid \dem{x\coloneqq E}m \in\varphi \}
\]
\end{lemma}
\begin{proof} By induction on the structure of $\varphi$.
 \begin{itemize}
 \item $\varphi= \varphi_1\land\varphi_2$.
   \begin{align*}
     \mathcal A_x^E[\varphi_1 \land\varphi_2] &= \mathcal A_x^E[\varphi_1] \land \mathcal A_x^E[\varphi_2]
     \\
     \intertext{By the induction hypothesis}
     &=\{ m \in \mathcal W(\mathcal S) \mid \dem{x\coloneqq E}m \in\varphi_1 \} \land \{ m \in \mathcal W(\mathcal S) \mid
     \dem{x\coloneqq E}m \in\varphi_2 \}
     \\
     &=\{ m \in \mathcal W(\mathcal S) \mid \dem{x\coloneqq E}m \in\varphi_1 \land \varphi_2 \}
   \end{align*}
   
 \item $\varphi= \varphi_1\lor\varphi_2$. 
   \begin{align*}
     \mathcal A_x^E[\varphi_1 \lor\varphi_2] &= \mathcal A_x^E[\varphi_1] \lor \mathcal A_x^E[\varphi_2]
     \\
     \intertext{By the induction hypothesis}
     &=\{ m \in \mathcal W(\mathcal S) \mid \dem{x\coloneqq E}m \in\varphi_1 \} \lor \{ m \in \mathcal W(\mathcal S) \mid
     \dem{x\coloneqq E}m \in\varphi_2 \}
     \\
     &=\{ m \in \mathcal W(\mathcal S) \mid \dem{x\coloneqq E}m \in\varphi_1 \lor \varphi_2 \}
   \end{align*}

\item $\varphi = \forall y\colon T.\ \psi$.
  \begin{align*}
    \mathcal{A}_x^E[\forall y\colon T.\ \psi]
    &= \forall y\colon T.\ \mathcal{A}_x^E[\psi]
    \\
    &= \bigcup_{v \in T} \{ m \in \mathcal W(\mathcal S) \mid \dem{x\coloneqq E}m \in\psi[v/y] \}
    \\
    &= \{ m \in \mathcal W(\mathcal S) \mid \dem{x\coloneqq E}m \in \bigcup_{v \in  T}\psi[v/y] \}
    \\
    &= \{ m \in \mathcal W(\mathcal S) \mid \dem{x\coloneqq E}m \in \forall y \colon T.\ \psi \}
 \end{align*}

\item $\varphi = \exists y\colon T.\ \psi$.
  \begin{align*}
    \mathcal{A}_x^E[\exists y\colon T.\ \psi]
    &= \exists y\colon T.\ \mathcal{A}_x^E[\psi]
    \\
    &= \bigcap_{v \in T} \{ m \in \mathcal W(\mathcal S) \mid \dem{x\coloneqq E}m \in\psi[v/y] \}
    \\
    &= \{ m \in \mathcal W(\mathcal S) \mid \dem{x\coloneqq E}m \in \bigcap_{v \in T}\psi[v/y] \}
    \\
    &= \{ m \in \mathcal W(\mathcal S) \mid \dem{x\coloneqq E}m \in \exists y \colon T.\ \psi \}
 \end{align*}
   
 \item $\forall\langle\sigma\rangle.\varphi$.
 \begin{align*}
   \mathcal A_x^E[\forall \langle\sigma\rangle.\varphi]
   &=  \forall \langle\sigma\rangle. \mathcal A_x^E[\varphi[E[\sigma]/\sigma(x)]] \\
   \intertext{By the induction hypothesis.}
   &= \forall \langle\sigma\rangle. \{ m \mid \dem{x\coloneqq E}m \in\varphi[E[\sigma]/\sigma(x)] \}
   \\
   &= \{ m \mid \forall t\in\supp(m).~ \dem{x\coloneqq E}m \in\varphi[E[t]/\sigma(x)][t/\sigma] \}
  \intertext{Since $m$ and $\dem{x\coloneqq E}m$ differ only in the values of $x$, we can instead quantify over the post-states, which are updated such that $t(x) = E[t]$.}
   &=  \{ m \mid \dem{x\coloneqq E}m \in
     \{ m' \mid \forall t\in\supp(m').\ m' \in \varphi[t/\sigma]
   \}
   \\
   &=  \{ m \mid \dem{x\coloneqq E}m \in \forall\langle\sigma\rangle.\varphi \}
 \end{align*}
 \end{itemize}
\end{proof}

\subsection{Nondeterministic Assignment}
\label{app:havoc}

\begin{align*}
  \mathcal H^S_x[\varphi\land \psi] & \triangleq \mathcal H^S_x[\varphi] \land \mathcal H^S_x[\psi]
  \\
  \mathcal H^S_x[\varphi\lor \psi] & \triangleq \mathcal H^S_x[\varphi] \lor \mathcal H^S_x[\psi]
  \\
  \mathcal H^S_x[\forall y\colon T.\varphi] & \triangleq \forall y\colon T. \mathcal H^S_x[\varphi]
  \\
  \mathcal H^S_x[\exists y\colon T.\varphi] & \triangleq  \exists y\colon T. \mathcal H^S_x[\varphi]
  \\
  \mathcal H^S_x[\forall \langle\sigma\rangle.\varphi] & \triangleq  \forall \langle\sigma\rangle. \
    \forall v\colon S.\
    \mathcal H^S_x[\varphi[v/\sigma(x)]]
  \\
  \mathcal H^S_x[\exists \langle\sigma\rangle.\varphi] & \triangleq  \exists \langle\sigma\rangle. \
    \exists v\colon S.\
    \mathcal H^S_x[\varphi[v/\sigma(x)]]
  \\
  \mathcal H^S_x[B] & \triangleq  B
\end{align*}

\begin{lemma}
If $m \in \mathcal H^S_x[\varphi]$, then $\bigoplus_{v\in S} \ind{\dem{x\coloneqq v}m} \Rightarrow \varphi$
%
\end{lemma}
\begin{proof}
Suppose that $m' \vDash\bigoplus_{v\in S}\ind{\dem{x\coloneqq v}m}$. The proof proceeds by induction on the structure of $\varphi$.
\begin{itemize}

\item $\varphi = \varphi_1\land\varphi_2$. We know that $m\in \mathcal H^S_x[\varphi_i]$ for each $i\in\{1,2\}$, so by the induction hypothesis $m'\vDash\varphi_i$ and therefore $m'\vDash \varphi_1\land\varphi_2$.

\item $\varphi = \varphi_1\lor\varphi_2$. Without loss of generality, suppose that $m\in \mathcal H^S_x[\varphi_1]$, so by the induction hypothesis $m'\vDash \varphi_1$ and therefore we can weaken the assertion to conclude that $m'\vDash \varphi_1\lor\varphi_2$.

\item $\varphi = \forall y : T.\ \psi$. We know that $m \in \mathcal H^S_x[\psi[v/y]]$ for all $v\in T$, therefore by the induction hypothesis, $m' \vDash\varphi[v/y]$ for all $v\in T$. This means that $m' \vDash \forall y\colon T.\ \varphi$.

\item $\varphi = \exists y : T.\ \psi$. We know that $m \in \mathcal H^S_x[\psi[v/y]]$ for some $v\in T$, therefore by the induction hypothesis, $m' \vDash\varphi[v/y]$. This means that $m' \vDash \exists y\colon T.\ \varphi$.

\item $\varphi = \forall\langle\sigma\rangle.\ \psi$. We know that $m\in \mathcal H_x^S[\psi[v/\sigma(x)][s/\sigma]]$ for all $s\in\supp(m)$ and $v\in S$. Therefore, by the induction hypothesis, $m'\vDash \psi[v/\sigma(x)][s/\sigma]$. Now, note that  $\psi[v/\sigma(x)][s/\sigma] = \psi[s[x\coloneqq v] / \sigma]$, so clearly $m'\vDash \psi[s[x\coloneqq v] / \sigma]$ for each $s \in \supp(m)$ and $v\in S$. In addition, $\supp(m') = \{ s[x \coloneqq v] \mid s\in \supp(m), v\in S\}$, and so $m' \vDash \psi[t/\sigma]$ for each $t\in \supp(m')$, therefore $m'\vDash \forall\langle\sigma\rangle.\ \psi$.

\item $\varphi = \exists\langle\sigma\rangle.\ \psi$. We know that $m\in \mathcal H_x^S[\psi[v/\sigma(x)][s/\sigma]]$ for some $s\in\supp(m)$ and $v\in S$. Therefore, by the induction hypothesis, $m'\vDash \psi[v/\sigma(x)][s/\sigma]$. Now, note that  $\psi[v/\sigma(x)][s/\sigma] = \psi[s[x\coloneqq v] / \sigma]$, so clearly $m'\vDash \psi[s[x\coloneqq v] / \sigma]$ for some $s \in \supp(m)$ and $v\in S$. In addition, $\supp(m') = \{ s[x \coloneqq v] \mid s\in \supp(m), v\in S\}$, and so $m' \vDash \psi[t/\sigma]$ for some $t\in \supp(m')$, therefore $m'\vDash \exists\langle\sigma\rangle.\ \psi$.

\item $\varphi = B$. If $m \in \mathcal H^S_x[B] = B$, then $B$ must be a tautology, so it must be the case that $m'\vDash B$.


\end{itemize}
\end{proof}

\begin{lemma}
The following rule is derivable
\[
  \inferrule{\;}{
    \triple{\mathcal H_x^{\{a, b\}}[\varphi]}{(x\coloneqq a) + (x\coloneqq b)}{\varphi}
  }
\]
\end{lemma}
\begin{proof}
\[
\inferrule*[right=\ruleref{Exists}]{
  \forall m\in\mathcal H^{\{a,b\}}_x[\varphi].
  \inferrule*[Right=\ruleref{Consequence}]{
    \inferrule*[Right=\ruleref{Plus}]{
      \inferrule*[right=\ruleref{Assign}]{\;}{
        \triple{\ind m}{x\coloneqq a}{\ind{\dem{x = a}m}}
      }
      \\
      \inferrule*[Right=\ruleref{Assign}]{\;}{
        \triple{\ind m}{x\coloneqq b}{\ind{\dem{x = b}m}}
      }
    }{
      \triple{\ind m}{(x\coloneqq a)+(x\coloneqq b)}{\textstyle\bigoplus_{v\in\{a, b\}} \ind{\dem{x = v}m}}
    }
  }{
    \triple{\ind m}{(x\coloneqq a)+(x\coloneqq b)}{\varphi}
  }
}{
  \triple{\mathcal H_x^{\{a,b\}}[\varphi]}{(x\coloneqq a)+(x\coloneqq b)}{\varphi}
}
\]
\end{proof}

\begin{lemma}
The following rule is derivable
\[
  \inferrule{\;}{
    \triple{\mathcal H_x^{\mathbb N}[\varphi]}{x \coloneqq \bigstar}{\varphi}
  }
\]
\end{lemma}
\begin{proof}
Recall that $x\coloneqq \bigstar$ is syntactic sugar for $x \coloneqq 0\fatsemi (x\coloneqq x+1)^\star$ and $(x\coloneqq x+1)^\star$ is syntactic sugar for $\iter{(x\coloneqq x+1)}\one\one$.
We will derive this rule using the \ruleref{Iter} rule. The first step is to select the assertion families:
\[
  \varphi_n \triangleq \psi_n \triangleq \ind{\dem{x\coloneqq n}m}
  \qquad
  \psi_\infty \triangleq \bigoplus_{k\in\mathbb N} \ind{\dem{x\coloneqq k}m}
\]
So clearly $\conv\psi$. We now complete derivation $(1)$ below:
\[
  \inferrule*[right=\ruleref{Iter}]{
    \inferrule*[Right=\ruleref{Seq}]{
      \inferrule*[right=\ruleref{Assume}]{
        \ind{\dem{x\coloneqq n}m} \vDash \one=\one
      }{
        \triple{\ind{\dem{x\coloneqq n}m}}{\assume\one}{\ind{\dem{x\coloneqq n}m}}
      }
      \quad
      \inferrule*[Right=\ruleref{Assign}]{\;}{
        \triple{\ind{\dem{x\coloneqq n}m}}{x\coloneqq x+1}{\ind{\dem{x\coloneqq n+1}m}}
      }
    }{
      \triple{\ind{\dem{x\coloneqq n}m}}{\assume\one \fatsemi x\coloneqq x+1}{\ind{\dem{x\coloneqq n+1}m}}
    }
  }{
    \triple{\ind{\dem{x\coloneqq n}m}}{(x\coloneqq x+1)^\star}{\textstyle\bigoplus_{k\in\mathbb N}\ind{\dem{x\coloneqq k}m}}
  }
\]
Using $(1)$ above, we complete the derivation as follows:
\[
\inferrule*[right=\ruleref{Exists}]{
  \forall m\in\mathcal H^{\mathbb N}_x[\varphi].
  \quad
  \inferrule*[Right=\ruleref{Consequence}]{
    \inferrule*[Right=\ruleref{Seq}]{
      \inferrule*[right=\ruleref{Assign}]{\;}{
        \triple{\ind m}{x\coloneqq 0}{\ind{\dem{x\coloneqq 0}m}}
      }
      \\
      (1)
    }{
      \triple{\ind m}{x\coloneqq\bigstar}{\textstyle\bigoplus_{v\in\mathbb N} \ind{\dem{x = v}m}}
    }
  }{
    \triple{\ind m}{x\coloneqq\bigstar}{\varphi}
  }
}{
  \triple{\mathcal H_x^{\mathbb N}[\varphi]}{x\coloneqq\bigstar}{\varphi}
}
\]

\end{proof}

\subsection{Assume}

\begin{align*}
  \Pi_b[\varphi\land \psi] & \triangleq \Pi_b[\varphi] \land \Pi_b[\psi]
  \\
  \Pi_b[\varphi\lor \psi] & \triangleq \Pi_b[\varphi] \lor \Pi_b[\psi]
  \\
  \Pi_b[\forall x\colon T.\varphi] & \triangleq \forall x\colon T. \Pi_b[\varphi]
  \\
  \Pi_b[\exists x\colon T.\varphi] & \triangleq  \exists x\colon T. \Pi_b[\varphi]
  \\
  \Pi_b[\forall \langle\sigma\rangle.\varphi] & \triangleq  \forall \langle\sigma\rangle. \
    \lnot b[\sigma] \vee \Pi_b[\varphi]
  \\
  \Pi_b[\exists \langle\sigma\rangle.\varphi] & \triangleq  \exists \langle\sigma\rangle. \
    b[\sigma] \land
    \Pi_b[\varphi]
  \\
  \Pi_b[B] & \triangleq  B
\end{align*}

\begin{lemma}\label{lem:syntactic-assume}
For any assertion $\varphi$ (from the syntax in \Cref{sec:hyper}) and test $b$, $\vdash\triple{\Pi_b[\varphi]}{\assume b}{\varphi}$
\end{lemma}
\begin{proof}
We first establish that for any $m$, if $m\vDash\Pi_b[\varphi]$, then $m\vDash (\varphi\land \always b) \oplus \always(\lnot b)$ by induction on the structure of $\varphi$.
\begin{itemize}

\item $\varphi = \varphi_1\land\varphi_2$. We know that $m\vDash \Pi_b[\varphi_i]$ for each $i\in\{1,2\}$. By the induction hypothesis, we get that $m\vDash(\varphi_i\land\always b) \oplus \always(\lnot b)$, so $m\vDash ((\varphi_1\land\always b) \oplus \always(\lnot b)) \land ((\varphi_2\land\always b) \oplus \always(\lnot b))$. Given that $b$ and $\lnot b$ are disjoint, we can simplify this to $m\vDash (\varphi_1\land \varphi_2\land\always b) \oplus \always(\lnot b)$.

\item $\varphi = \varphi_1\lor\varphi_2$. Without loss of generality, suppose $m\vDash \Pi_b[\varphi_1]$. By the induction hypothesis, we get that $m\vDash(\varphi_1\land\always b) \oplus \always(\lnot b)$. We can weaken this to $m\vDash((\varphi_1\lor\varphi_2)\land\always b) \oplus \always(\lnot b)$.

\item $\varphi = \forall x\colon T.\varphi$. Similar to the case for conjunctions above.
\item $\varphi = \exists x\colon T.\varphi$. Similar to the case for disjunctions above.

\item $\varphi = \forall\langle\sigma\rangle.\varphi$. We know that $m \vDash \lnot b[s] \lor \Pi_b[\varphi[s/\sigma]]$ for all $s\in\supp(m)$. So, for each $s$, we know that either $\de{b}_{\mathsf{Test}}(s) = \fls$, of $m\vDash \Pi_b[\varphi[s/\sigma]]$, in which case $m\vDash (\varphi[s/\sigma] \land\always b) \oplus (\always(\lnot b))$ by the induction hypothesis. Since this is true for all $s$, we get $m\vDash ((\forall\langle\sigma\rangle.\varphi) \land\always b) \oplus (\always(\lnot b))$.

\item $\varphi = \exists\langle\sigma\rangle.\varphi$. We know that $m \vDash b[s] \land \Pi_b[\varphi[s/\sigma]]$ for some $s\in\supp(m)$. So, by the induction hypothesis $m\vDash (\varphi[s/\sigma] \land \always b) \oplus \always(\lnot b)$, which we can weaken to $m\vDash ((\exists\langle\sigma\rangle.\varphi) \land \always b) \oplus \always(\lnot b)$

\item $\varphi = B$. Since $m\vDash B$, then $B$ must not contain any free variables, and therefore $B$ is a tautology. So we have:
\[
  m\vDash\top
  \quad\Leftrightarrow\quad
  m\vDash \always b \oplus \always(\lnot b)
  \quad\Leftrightarrow\quad
    m\vDash (B \land \always b) \oplus \always(\lnot b)
\]

\end{itemize}
We complete the derivation as follows:
\[
\inferrule*[right=\ruleref{Consequence}]{
    \inferrule*[Right=\ruleref{Choice}]{
      \inferrule*[right=\ruleref{Assume}]{
        \varphi\land\always b \vDash b
      }{
        \triple{\varphi\land\always b}{\assume b}{\varphi\land \always b}
      }
      \\
      \inferrule*[Right=\ruleref{Assume}]{
        \always (\lnot b) \vDash b = \zero
      }{
        \triple{\always(\lnot b)}{\assume b}{(\always(\lnot b))\odot\zero}
      }
    }{
    \triple{(\varphi \land \always b) \oplus \always(\lnot b)}{\assume b}{(\varphi \land\always b) \oplus (\always(\lnot b))\odot\zero}
  }
}{
  \triple{\Pi_b[\varphi]}{\assume b}{\varphi}
}
\]

\end{proof}

\section{Reusing Proof Fragments}
\label{app:examples}

\subsection{Integer Division}

Recall the definition of the program below that divides two integers.
\[
\code{Div} \triangleq \left\{
\begin{array}{l}
\var q \coloneqq 0 \fatsemi \var r \coloneqq \var a \fatsemi \phantom{x} \\
\whl{\var r \ge \var b}{\;} \\
\quad \var r \coloneqq \var r - \var b \fatsemi\phantom{x} \\
\quad \var q \coloneqq \var q+1
\end{array}
\right.
\]
To analyze this program with the \ruleref{Variant} rule, we need a family of variants $(\varphi_n)_{n \in \mathbb{N}}$, defined as follows.
\[
\varphi_n \triangleq \left\{
\begin{array}{lll}
\sure{\var q +n = \lfloor \var a\div \var b\rfloor \land \var r =(\var a\bmod \var b) + n \times \var b} & \text{if} & n \le \lfloor \var a\div \var b\rfloor \\
\bot & \text{if} & n > \lfloor \var a\div \var b\rfloor
\end{array}
\right.
\]
Additionally, it must be the case that $\varphi_n \vDash \var r \ge \var b$ for all $n\ge 1$ and $\varphi_0\vDash \var r < \var b$. For $n > \lfloor \var a\div\var b\rfloor$, $\varphi_n = \bot$ and $\bot\vDash \var r \ge \var b$ vacuously. If $1 \le n \le \lfloor \var a\div\var b\rfloor$, then we know that $\var r = (\var a \bmod{\var b}) + n\times \var b$, and since $n \ge 1$, then $\var r \ge \var b$. When $n=0$, we know that $\var r = \var a\bmod\var b$, and so by the definition of $\bmod$, it must be that $\var r < \var b$.

\begin{figure}[t]
\[\def\arraystretch{1.25}
\begin{array}{l}
\ob{\sure{\var a \ge 0 \land \var b > 0}}\\
\;\; \var q \coloneqq 0 \fatsemi \phantom{X}\\
\ob{\sure{\var a \ge 0 \land \var b > 0 \land \var q=0}}\\
\;\; \var r \coloneqq \var a \fatsemi \phantom{x} \\
\ob{\sure{\var a \ge 0 \land \var b > 0 \land \var q=0 \land \var r = \var a}} \implies\\
\ob{\sure{\var q + \lfloor \var a \div \var b\rfloor = \lfloor \var a\div \var b\rfloor \land \var r = (\var a \bmod{\var b}) + \lfloor \var a\div \var b\rfloor\times \var b}} \implies \\
\ob{\varphi_{\lfloor \var a \div \var b\rfloor }} \implies\\
\ob{\exists n:\mathbb N. \varphi_n}\\
\;\; \whl{\var r \ge \var b}{\;}\\
\quad \ob{\varphi_n} \implies\\
\quad \ob{\sure{\var q +n = \lfloor \var a\div \var b\rfloor \land \var r =(\var a\bmod \var b) + n \times \var b}}\\
\;\;\quad \var r \coloneqq \var r - \var b \fatsemi\phantom{x} \\
\quad \ob{\sure{\var q +n = \lfloor \var a\div \var b\rfloor \land \var r =(\var a\bmod \var b) + (n-1) \times \var b}}\\
\;\;\quad \var q \coloneqq \var q+1 \\
\quad \ob{\sure{\var q +(n-1) = \lfloor \var a\div \var b\rfloor \land \var r =(\var a\bmod \var b) + (n-1) \times \var b}} \implies\\
\quad \ob{\varphi_{n-1}}\\
\ob{\varphi_0} \implies\\
\ob{\sure{\var q + 0 = \lfloor \var a\div \var b\rfloor \land \var r =(\var a\bmod \var b) + 0 \times \var b}} \implies\\
\ob{\sure{\var q = \lfloor \var a\div \var b\rfloor \land \var r =(\var a\bmod \var b)}}
\end{array}
\]
\caption{Derivation for the $\mathsf{DIV}$ program.}
\label{fig:div}
\end{figure}

The derivation is given in \Cref{fig:div}. Most of the steps are obtained by straightforward applications of the inference rules, with consequences denoted by $\implies$. In the application of the \ruleref{Variant} rule, we only show the case where $n \le \lfloor \var a\div \var b\rfloor$. The case where $n > \lfloor \var a\div \var b\rfloor$ is vacuous by applying the \ruleref{False} rule from \Cref{fig:rules}.

\subsection{The Collatz Conjecture}

\begin{figure*}
\begin{minipage}{\linewidth}
\[\def\arraystretch{1.2}
\begin{array}{l}
\ob{\sure{\var a = n \land n>0}} \\
\;\; \var i \coloneqq 0 \fatsemi \phantom{x} \\
\ob{\sure{\var a = n \land n>0 \land \var i = 0}} \implies \\
\ob{\sure{\var a = f^{\var i}(n) \land \forall k<\var i. f^k(n) \neq 1}} \\
\;\;\whl{\var a \neq 1}{\;}\\
\quad \ob{\sure{\var a = f^{\var i}(n) \land \forall k<\var i. f^k(n) \neq 1 \land \var a \neq 1}} \implies \\
\quad \ob{\sure{\var a = f^{\var i}(n) \land \forall k<\var i+1. f^k(n) \neq 1}} \\
\;\;\quad \var b \coloneqq 2 \fatsemi \phantom{x} \\
\quad \ob{\sure{\var a = f^{\var i}(n) \land \forall k<\var i+1. f^k(n) \neq 1 \land \var b = 2}} \\
\;\;\quad \code{Div} \fatsemi\phantom{x}\\
\quad \ob{\sure{\var a = f^{\var i}(n) \land \forall k<\var i+1. f^k(n) \neq 1 \land \var b = 2 \land \var q = \lfloor \var a\div \var b\rfloor \land \var r =(\var a\bmod \var b)}} \implies\\
\quad \ob{\sure{\var a = f^{\var i}(n) \land \forall k<\var i+1. f^k(n) \neq 1 \land \var q = \lfloor f^{\var i}(n)\div 2\rfloor \land \var r =(f^{\var i}(n)\bmod 2)}}\\
\;\; \quad \code{if} ~\var r = 0 ~ \code{then} \\
\quad\quad \ob{\sure{\var a = f^{\var i}(n) \land \forall k<\var i+1. f^k(n) \neq 1 \land \var q = \lfloor f^{\var i}(n)\div 2\rfloor \land \var r =(f^{\var i}(n)\bmod 2)  \land \var r = 0}} \implies\\
\quad\quad \ob{\sure{\forall k<\var i+1. f^k(n) \neq 1 \land \var q = \lfloor f^{\var i}(n)\div 2\rfloor \land \var (f^{\var i}(n)\bmod 2) = 0}}\\
\;\;\quad\quad {\var a \coloneqq \var q} \\
\quad\quad \ob{\sure{\forall k<\var i+1. f^k(n) \neq 1 \land \var a = \lfloor f^{\var i}(n)\div 2\rfloor \land (f^{\var i}(n)\bmod 2)= 0}} \implies\\
\quad\quad \ob{\always( \var a =  f^{\var i +1 }(n) \land \forall k<\var i+1. f^k(n) \neq 1)}\\
\;\;\quad\code{else} \\
\quad\quad \ob{\sure{\var a = f^{\var i}(n) \land \forall k<\var i+1. f^k(n) \neq 1 \land \var q = \lfloor f^{\var i}(n)\div 2\rfloor \land \var r =(f^{\var i}(n)\bmod 2)  \land \var r \neq 0}} \implies\\
\quad\quad \ob{\sure{\var a = f^{\var i}(n) \land \forall k<\var i+1. f^k(n) \neq 1 \land (f^{\var i}(n)\bmod 2) =1}} \\
\;\;\quad\quad {\var a \coloneqq 3 \times \var a + 1} \fatsemi\phantom{x} \\
\quad\quad \ob{\sure{\var a = 3\times f^{\var i}(n) +1 \land \forall k<\var i+1. f^k(n) \neq 1 \land (f^{\var i}(n)\bmod 2) = 1}} \implies\\
\quad\quad \ob{\always( \var a =  f^{\var i +1 }(n) \land \forall k<\var i+1. f^k(n) \neq 1)}\\
\quad \ob{\always( \var a =  f^{\var i +1 }(n) \land \forall k<\var i+1. f^k(n) \neq 1)}\\
\;\;\quad \var i \coloneqq \var i+1 \\
\quad \ob{\always( \var a =  f^{\var i }(n) \land \forall k<\var i. f^k(n) \neq 1)}\\
\ob{\always( \var a =  f^{\var i }(n) \land \forall k<\var i. f^k(n) \neq 1 \land \var a = 1)}\implies\\
\ob{\always(\var i = S_n)}
\end{array}
\]
\end{minipage}
\caption{Derivation for the $\mathsf{COLLATZ}$ program.}
\label{fig:collatz}
\end{figure*}

Recall the definition of the program below that finds the stopping time of some positive number $n$.
\[
\code{Collatz} \triangleq \left\{
\begin{array}{l}
\var i \coloneqq 0 \fatsemi \phantom{x} \\
\whl{\var a \neq 1}{\;}\\
\quad \var b \coloneqq 2 \fatsemi
 \code{Div} \fatsemi\phantom{x}\\
\quad \iftf{\var r = 0}{\var a \coloneqq \var q}{\var a \coloneqq 3 \times \var a + 1} \fatsemi\phantom{x} \\
\quad \var i \coloneqq \var i+1
\end{array}
\right.
\]
The derivation is shown in \Cref{fig:collatz}.
Since we do not know if the program will terminate, we use the \ruleref{Invariant} rule to obtain a partial correctness specification. We choose the loop invariant:
\[
\var a = f^{\var i}(n) \quad\land\quad \forall k<\var i.\ f^k(n) \neq 1
\]
So, on each iteration of the loop, $\var a$ holds the value of applying $f$ repeatedly $\var i$ times to $n$, and 1 has not yet appeared in this sequence.

Immediately upon entering the while loop, we see that $\var a = f^{\var i}(n) \neq 1$, and so from that and the fact that $\forall k<\var i. f^k(n)\neq 1$, we can conclude that $\forall k<\var i+1. f^k(n)\neq 1$.

The $\mathsf{DIV}$ program is analyzed by inserting the proof from \Cref{fig:div}, along with an application of the rule of \ruleref{Constancy} to add information about the other variables. We can omit the $\always$ modality from rule of \ruleref{Constancy}, since $\sure P \land \always Q \Leftrightarrow \sure{P\land Q}$.

When it comes time to analyze the if statement, we use the \ruleref{If (Hoare)} rule (\Cref{lem:ifhoare}) to get a partial correctness specification. The structure of the if statement mirrors the definition of $f(n)$, so the effect is the same as applying $f$ to $\var a$ one more time, therefore we get that $\var a = f^{\var i+1}(n)$.

After exiting the while loop, we know that $f^{\var i}(n) = 1$ and $f^k(n)\neq 1$ for all $k <\var i$, therefore $\var i$ is (by definition) the stopping time, $S_n$.

\subsection{Embedding Division in a Probabilistic Program}

Recall the program that loops an even number of iterations with probability $\frac23$ and an odd number of iterations with probability $\frac13$.
\[
\var a := 0 \fatsemi \var r := 0 \fatsemi \left(\var a := \var a +1\fatsemi \var b:= 2 \fatsemi \code{Div}\right)^{\langle \frac12\rangle}
\]
To analyze this program with the \ruleref{Iter} rule, we define the two families of assertions below for $n\in\mathbb{N}$.
\[
\varphi_n \triangleq \wg{\var a = n \land \var r = \var a \bmod 2}{\frac1{2^n}}
\qquad
\psi_n \triangleq \wg{\var a = n \land \var r = \var a \bmod 2}{\frac1{2^{n+1}}}
\]
Additionally, let $\psi_\infty = \sure{\var r = 0}\oplus_{\frac23} \sure{\var r = 1}$. We now show that $(\psi_n)_{n\in\mathbb{N}^\infty}$ converges. Suppose that $m_n \vDash \psi_n$ for each $n\in\mathbb N$. So $m_n \vDash \wg{\var r = 0}{\frac1{2^{n+1}}}$ for all even $n$ and $m_n \vDash \wg{\var r = 1}{\frac1{2^{n+1}}}$ for all odd $n$. In other words, the cummulative probability mass for each $m_n$ where $n$ is even is:
\[
\sum_{k \in \mathbb N}\frac1{2^{2k+1}} = \frac12\cdot\sum_{k\in\mathbb N}\left(\frac14\right)^k = \frac12\cdot \frac1{1 - \frac14} = \frac23
\]
Where the second-to-last step is obtained using the standard formula for geometric series. Similarly, the total probability mass for $n$ being odd is:
\[
\sum_{k \in \mathbb N}\frac1{2^{2k+2}} = \frac14\cdot\sum_{k\in\mathbb N}\left(\frac14\right)^k = \frac14\cdot \frac1{1 - \frac14} = \frac13
\]
We therefore get that $\sum_{n\in\mathbb N} m_n\vDash \sure{\var r = 0}\oplus_{\frac23}\sure{\var r = 1}$. Having shown that, we complete the derivation, shown in \Cref{fig:probmain}. 
Two proof obligations are generated by applying the \ruleref{Iter} rule, the first is proven in \Cref{fig:probmain}. Note that in order to apply our previous proof for the $\mathsf{DIV}$ program, it is necessary to use the \ruleref{Scale} rule. The second proof obligation of the \ruleref{Iter} rule is to show $\triple{\varphi_n}{\assume{\frac12}}{\psi_n}$, which is easily dispatched using the \ruleref{Assume} rule.
\begin{figure}[t]
\[
\begin{array}{l|l}
\begin{array}{l}
\ob{\sure\tru} \\
\;\; \var a := 0 \fatsemi \phantom{\;} \\
\ob{\sure{\var a = 0}} \\
\;\;\var r := 0 \fatsemi \phantom{\;} \\ 
\ob{\sure{\var a = 0 \land \var r = 0}} \implies \\
\ob{\wg{\var a = 0 \land \var r = 0 \bmod 2}1} \implies \\
\ob{\varphi_0} \\
\left(\begin{array}{l}
\var a := \var a +1\fatsemi \phantom{\;} \\
 \var b:= 2 \fatsemi \phantom{\;} \\
\code{Div}
\end{array}
\right)^{\langle \frac12\rangle} \\
\ob{\psi_\infty} \implies \\
\ob{\sure{\var r = 0}\oplus_{\frac23}\sure{\var r = 1}}
\end{array}
\quad
&
\quad
\begin{array}{l}
\ob{\varphi_n} \implies \\
\ob{\wg{\var a = n}{\frac1{2^n}}} \\
\;\; \assume{\frac12} \fatsemi \phantom{\;}\\
\ob{\wg{\var a = n}{\frac1{2^{n+1}}}} \\
\;\; \var a := \var a + 1 \fatsemi \phantom{\;}\\
\ob{\wg{\var a = n+1}{\frac1{2^{n+1}}}} \\
\;\;  \var b:= 2 \fatsemi \phantom{\;} \\
\ob{\wg{\var a = n+1 \land \var b = 2}{\frac1{2^{n+1}}}} \\
\;\; \code{Div} \\
\ob{\wg{\var a = n+1 \land \var r = \var a \bmod 2}{\frac1{2^{n+1}}}} \implies\\
\ob{\varphi_{n+1}}
\end{array}
\end{array}
\]
\caption{Left: derivation for the main body of the probabilistic looping program. Right: derivation of the probabilistic loop.}
\label{fig:probmain}
\end{figure}

\section{Graph Problems and Quantitative Analysis}

\subsection{Counting Random Walks}
\label{app:walk}

\begin{figure*}
\begin{minipage}{\linewidth}
\[\def\arraystretch{1}
\begin{array}{l}
\ob{\sure{\var x = 0 \land \var y = 0}} \implies \\
\ob{\varphi_{N+M}} \implies \\
\ob{\exists n:\mathbb{N}.~\varphi_n} \\
\whl{\var x < N \vee \var y < M}{}\\
\quad\ob{\varphi_{n+1}} \\
\quad \code{if}~\var x < N \land \var y < M ~\code{then} \\
\quad\quad \ob{ \smashoperator{\bigoplus_{k = \max(1, n+1-M)}^{\min(N, n)}} \wg{\var x = N-k \land \var y = M - (n+1-k) }{{N+M - (n+1) \choose N-k}} } \\
\quad\quad (\var x \coloneqq \var x +1 ) + (\var y \coloneqq \var y +1 ) \\
\quad\quad \ob{ \smashoperator{\bigoplus_{k = \max(1, n+1-M)}^{\min(N, n)}}
\wg{\var x = N-k+1 \land \var y = M - (n+1-k) }{{N+M - (n+1) \choose N-k}} \oplus 
\wg{\var x = N-k \land \var y = M - (n-k) }{{N+M - (n+1) \choose N-k}}
} \\
\quad \code{else if}~\var x \ge N ~\code{then} \\
\quad\quad \ob{ \smashoperator{\bigoplus_{k = \max(0, n+1-M)}^{0}} \wg{\var x = N-k \land \var y = M - (n+1-k)}{{N+M - (n+1) \choose N-k}} } \\
\quad\quad \var y \coloneqq \var y + 1 \\
\quad\quad \ob{ \smashoperator{\bigoplus_{k = \max(0, n+1-M)}^{0}} \wg{\var x = N-k \land \var y = M - (n-k)}{{N+M - (n+1) \choose N-k}} } \\
\quad \code{else} \\
\quad\quad \ob{ \smashoperator{\bigoplus_{k = n+1}^{\min(N, n+1)}} \wg{\var x = N-k \land \var y = M - (n+1-k)}{{N+M - (n+1) \choose N-k}} } \\
\quad\quad \var x \coloneqq \var x + 1 \\
\quad\quad \ob{ \smashoperator{\bigoplus_{k = n+1}^{\min(N, n+1)}} \wg{\var x = N-k+1 \land \var y = M - (n+1-k)}{{N+M - (n+1) \choose N-k}} } \\
\quad \ob{\varphi_n} \\
\ob{\wg{\var x = N \land \var y = M}{{N+M\choose N}} }
\end{array}
\]
\end{minipage}
\caption{Random walk proof}
\label{fig:walk}
\end{figure*}

Recall the following program that performs a random walk on a two dimensional grid in order to discover how many paths exist between the origin $(0,0)$ and the point $(N,M)$.
\[
\code{Walk} \triangleq \left\{
\begin{array}{l}
\whl{\var x < N \vee \var y < M}{} \\
\quad \code{if}~\var x < N \land \var y < M ~\code{then} \\
\quad\quad (\var x \coloneqq \var x +1 ) + (\var y \coloneqq \var y +1 ) \\
\quad \code{else if}~\var x \ge N ~\code{then} \\
\quad\quad \var y \coloneqq \var y + 1 \\
\quad \code{else} \\
\quad\quad \var x \coloneqq \var x + 1
\end{array}\right.
\]
The derivation is provided in \Cref{fig:walk}.
Since this program is guaranteed to terminate after exactly $N+M$ steps, we use the following loop \ruleref{Variant}, where the bounds for $k$ are described in \Cref{sec:walk}.
\[
\varphi_n \triangleq \smashoperator{\bigoplus_{k = \max(0, n-M)}^{\min(N, n)}} \wg{\var x = N-k \land \var y = M - (n-k)}{{N+M - n \choose N-k}}
\]
Recall that $n$ indicates how many steps $(x,y)$ is from $(N,M)$, so $\varphi_{N+M}$ is the precondition and $\varphi_0$ is the postcondition.
Upon entering the while loop, we encounter nested if statements, which we analyze with the \ruleref{If} rule. This requires us to split $\varphi_{n+1}$ into three components, satisfying $n < N \land y<M$, $n \ge N$, and $y \ge M$, respectively. The assertion $x \ge N$ is only possible if we have already taken at least $N$ steps, or in other words, if $n+1 \le (N+M) - N = M$. Letting $k$ range from $\max(0, n+1-M)$ to $0$ therefore gives us a single term $k=0$ when $n+1 \le M$ and an empty conjunction otherwise. A similar argument holds when $y \ge M$. All the other outcomes go into the first branch, where we preclude the $k=0$ and $k=n+1$ cases since it must be true that $x\neq N$ and $y\neq M$.

Let $P(n, k) = (x=N-k\land y=M-(n-k))$. Using this shorthand, the postcondition at the end of the if statement is obtained by taking an outcome conjunction of the results from the three branches.
\begin{align*}
&\smashoperator{\bigoplus_{k = \max(1, n+1-M)}^{\min(N, n)}}
\wg{P(n, k-1)}{{N+M - (n+1) \choose N-k}}
\oplus 
\smashoperator{\bigoplus_{k = \max(1, n+1-M)}^{\min(N, n)}}\wg{P(n, k)}{{N+M - (n+1) \choose N-k}}
\\& \oplus
\smashoperator{\bigoplus_{k = \max(0, n+1-M)}^{0}} \wg{P(n, k)}{{N+M - (n+1) \choose N-k}}
\oplus
\smashoperator{\bigoplus_{k = n+1}^{\min(N, n+1)}} \wg{P(n, k-1)}{{N+M - (n+1) \choose N-k}}
\end{align*}
Now, we can combine the conjunctions with like terms.
\[
\smashoperator{\bigoplus_{k = \max(1, n+1-M)}^{\min(N, n+1)}}\wg{P(n, k-1)}{{N+M - (n+1) \choose N-k}}
\oplus
\smashoperator{\bigoplus_{k = \max(0, n+1-M)}^{\min(N, n)}}\wg{P(n, k)}{{N+M - (n+1) \choose N-k}}
\]
And adjust the bounds on the first conjunction by subtracting 1 from the lower and upper bounds of $k$:
\[
\smashoperator{\bigoplus_{k = \max(0, n-M)}^{\min(N-1, n)}}\wg{P(n, k)}{{N+M - (n+1) \choose N-(k+1)}}
\oplus
\smashoperator{\bigoplus_{k = \max(0, n+1-M)}^{\min(N, n)}}\wg{P(n, k)}{{N+M - (n+1) \choose N-k}}
\]
Now, we examine when the bounds of these two conjunctions differ. If $n \ge M$, then the first conjunction has an extra $k = n-M$ term. Similarly, the second conjunction has an extra $k = N$ term when $n \ge N$.
Based on that observation, we split them as follows:
\begin{align*}
&\smashoperator{\bigoplus_{k \in \{n-M \mid n \ge M \}}}\wg{P(n, k)}{{N+M - (n+1) \choose N-(k+1)}}
\oplus\smashoperator{\bigoplus_{k \in \{ N \mid n \ge N \}}}\wg{P(n, k)}{{N+M - (n+1) \choose N-k}}
\\&\oplus
\smashoperator{\bigoplus_{k = \max(0, n+1-M)}^{\min(N-1, n)}}\wg{P(n, k)}{{N+M - (n+1) \choose N-(k+1)} + {N+M - (n+1) \choose N-k}}
\end{align*}
Knowing that $k=n-M$ in the first conjunction, we get that:
\[
{N+M-(n+1)\choose N-(k+1)}
=
{N+M-(n+1) \choose N- (n-M +1)}
=
1
=
{N+M - n \choose N - k}
\]
Similarly, for the second conjunction we get the same weight. Also, observe that for any $a$ and $b$:
\begin{align*}
{a \choose b}+{a \choose b +1 }
&= \frac{a!}{b!(a-b)!} + \frac{a!}{(b+1)!(a-b-1)!} \\
&= \frac{a!}{b!(a-b)(a-b-1)!} + \frac{a!}{(b+1)b!(a-b-1)!} \\
&= \frac{a!(b+1) + a!(a-b)}{(b+1)b!(a-b)(a-b-1)!} \\
&= \frac{a!(b+1+a-b)}{(b+1)!(a-b)!} \\
&= \frac{(a+1)!}{(b+1)!((a+1)-(b+1))!} \\
&= {a+1 \choose b+1}
\end{align*}
So, letting $a=N+M-(n+1)$ and $b = N-(k+1)$, it follows that:
\[
{N+M - (n+1) \choose N-(k+1)} + {N+M - (n+1) \choose N-k}
=
{N+M - n \choose N-k}
\]
We can therefore rewrite the assertion as follows:
\begin{align*}
&\smashoperator{\bigoplus_{k \in \{n-M \mid n \ge M \}}}\wg{P(n, k)}{{N+M - n \choose N-k}}
\oplus\smashoperator{\bigoplus_{k \in \{ N \mid n \ge N \}}}\wg{P(n, k)}{{N+M - n \choose N-k}}
\\&\oplus
\smashoperator{\bigoplus_{k = \max(0, n+1-M)}^{\min(N-1, n)}}\wg{P(n, k)}{{N+M - n \choose N-k}}
\end{align*}
And by recombining the terms, we get:
\[
\smashoperator{\bigoplus_{k = \max(0, n-M)}^{\min(N, n)}}\wg{P(n, k)}{{N+M - n \choose N-k}}
\]
Which is precisely $\varphi_n$. According to the \ruleref{Variant} rule, the final postcondition is just $\varphi_0$.

\subsection{Shortest Paths}

Recall the following program that nondeterministically finds the shortest path from $s$ to $t$ using a model of computation based on the tropical semiring (\Cref{ex:tropicalsemi}).
\[
\code{SP} \triangleq \left\{
\begin{array}{l}
\whl{\var{pos}\neq \var{t}}{} \\
\quad \var{next} \coloneqq 1 \fatsemi \phantom{x} \\
\quad \iter{\left( \var{next} \coloneqq \var{next}+1\right)}{\var{next} < N}{~G[\var{pos}][\var{next}]} \fatsemi \phantom{x} \\
\quad \var{pos} \coloneqq \var{next}\fatsemi\phantom{x} \\
\quad \assume 1
\end{array}
\right.
\]
The derivation is shown in \Cref{fig:shortestpath}. We use the \ruleref{While} rule to analyze the outer loop. This requires the following families of assertions, where $\varphi_n$ represents the outcomes where the guard remains true after exactly $n$ iterations and $\psi_n$ represents the outcomes where the loop guard is false after $n$ iterations. Let $I = \{1, \ldots, N\} \setminus \{t\}$.
\[
\varphi_n \triangleq \smashoperator{\bigoplus_{i \in I}} \wg{\var{pos} = i}{\short_n^t(G, \var s, i) + n}
\]
\[
\psi_n \triangleq \wg{\var{pos} = \var t}{\short_n^t(G, \var s, \var t) + n}
\qquad\qquad
\psi_\infty \triangleq \wg{\var{pos} = \var t}{\short(G, \var s, \var t)}
\]
We now argue that $\conv\psi$. Take any $(m_n)_{n\in\mathbb N}$ such that $m_n \vDash\psi_n$ for each $n$, which means that $|m_n| = \short_n^t(G,s,t)+n$ and $\supp(m_n) \subseteq (\var{pos} = t)$. In the tropical semiring, $|m_n|$ corresponds to the minimum weight of any element in $\supp(m_n)$, so we know there is some $\sigma \in \supp(m_n)$ such that $m_n(\sigma) = \short_n^t(G,s,t) + n$, and since $\short_n^t(G,s,t)$ is Boolean valued and $\tru = 0$ and $\fls = \infty$, then $m_n(\sigma)$ is either $n$ or $\infty$.

By definition, the minimum $n$ for which $\short_n^t(G,s,t) = \tru$ is $\short(G, s, t)$, so for all $n < \short_n^t(G,s,t)$, it must be the case that $|m_n| = \infty$ and for all $n \ge \short_n^t(G,s,t)$, it must be the case that $|m_n| = n$. Now, $|\sum_{n\in\mathbb N}m_n| = \min_{n\in\mathbb N} |m_n| = \short_n^t(G,s,t)$, and since all elements of each $m_n$ satisfies $\var{pos} = t$, then we get that $\sum_{n\in\mathbb N}m_n\vDash\psi_\infty$.

Now, we will analyze the inner iteration using the \ruleref{Iter} rule and the following two families of assertions, which we will assume are 1-indexed for simplicity of the proof.
\begin{figure}[t]
\begin{minipage}{\linewidth}
\[\def\arraystretch{1.5}
\begin{array}{l}
\ob{\sure{\var{pos} = \var s}} \implies \\
\ob{\bigoplus_{i = 1}^N \wg{\var{pos} = i}{\short_0^t(G, \var s, i)}} \implies \\
\ob{\varphi_0 \oplus \psi_0} \\
\;\whl{\var{pos}\neq \var{t}}{} \\
\quad\ob{\varphi_n} \implies \\
\quad \ob{\smashoperator{\bigoplus_{i \in I}} \wg{\var{pos} = i}{\short_n^t(G, \var s, i)+n}} \\
\quad\; \var{next} \coloneqq 1 \fatsemi \phantom{x} \\
\quad \ob{{\bigoplus_{i \in I}} \wg{\var{pos} = i \land \var{next} = 1}{\short_n^t(G, \var s, i)+n}} \\
\quad\; \iter{\left(
\var{next} \coloneqq \var{next}+1
\right)}{\var{next} < N}{~G[\var{pos}][\var{next}]} \fatsemi \phantom{x} \\
\quad \ob{ \bigoplus_{j=1}^N {\bigoplus_{i \in I}} \wg{\var{pos} = i \land \var{next} = j}{(\short_n^t(G, \var s, i)\land G[i][j])+n} } \implies \\
\quad \ob{ \bigoplus_{j=1}^N {\bigoplus_{i \in I}} \wg{\var{next} = j}{(\short_n^t(G, \var s, i)\land G[i][j])+n} } \\
\quad\; \var{pos} \coloneqq \var{next} \fatsemi \phantom{x}\\
\quad \ob{ \bigoplus_{j=1}^N {\bigoplus_{i \in I}} \wg{\var{pos} = j}{(\short_n^t(G, \var s, i)\land G[i][j])+n} } \\
\quad\; \assume 1 \\
\quad \ob{ \bigoplus_{j=1}^N {\bigoplus_{i \in I}} \wg{\var{pos} = j}{(\short_n^t(G, \var s, i)\land G[i][j])+n+1} } \implies \\
\quad \ob{ \bigoplus_{j=1}^N \wg{\var{pos} = j}{\short_{n+1}^t(G, \var s, j)+n+1} } \implies\\
\quad \ob{ \varphi_{n+1} \oplus \psi_{n+1}} \\
\ob{\psi_\infty} \implies \\
\ob{ \wg{\var{pos} = \var t}{\short(G, \var s, \var t)} }
\end{array}
\]
\end{minipage}
\caption{Shortest path proof}
\label{fig:shortestpath}
\end{figure}
\begin{align*}
\vartheta_j &\triangleq \left\{
\def\arraystretch{1.4}
\begin{array}{lll}
\bigoplus_{i\in I} \wg{\var{pos}=i \land \var{next} = j}{\short_n^t(G, s, i) +n} & \text{if} & j < N \\
\top\odot\zero & \text{if} & j \ge N
\end{array}\right.
\\
\xi_j &\triangleq\left\{
\def\arraystretch{1.4}
\begin{array}{lll}
\bigoplus_{i\in I} \wg{\var{pos}=i \land \var{next} = j}{(\short_n^t(G, s, i)\land G[i][j]) +n} & \text{if} & j < N \\
\top\odot\zero & \text{if} & j \ge N
\end{array}\right.
\\
\xi_\infty &\triangleq \bigoplus_{j=1}^N {\bigoplus_{i \in I}} \wg{\var{pos} = i \land \var{next} = j}{(\short_n^t(G, \var s, i)\land G[i][j])+n}
\end{align*}
It is easy to see that $\conv\xi$ since $\xi_\infty$ is by definition an outcome conjunction of all the non-empty terms $\xi_j$. When $j < N$, then we get $\vartheta_j \vDash \var{next} < N$, so we dispatch the first proof obligation of the \ruleref{Iter} rule as follows:
\[\def\arraystretch{1.25}
\begin{array}{l}
\ob{\vartheta_j} \implies \\
\ob{ \bigoplus_{i\in I} \wg{\var{pos} = i \land \var{next}= j}{\short_n^t(G, s, i) + n}} \\
\; \assume{\var{next} < N} \fatsemi\phantom{x} \\
\ob{ \bigoplus_{i\in I} \wg{\var{pos} = i \land \var{next}= j}{\short_n^t(G, s, i) + n}} \\
\; \var{next} \coloneqq \var{next}+1 \\
\ob{ \bigoplus_{i\in I} \wg{\var{pos} = i \land \var{next}= j+1}{\short_n^t(G, s, i) + n}} \implies \\
\ob{\vartheta_{j+1}}
\end{array}
\]
When instead $j \ge N$, then we know that $\vartheta_j\vDash \lnot (\var{next}< N)$ and so it is easy to see that:
\[
\triple{\vartheta_j}{\assume{\var{next}<N}\fatsemi \var{next} \coloneqq \var{next}+1}{\top\odot\zero}
\]
For the second proof obligation, we must show that:
\[
\triple{\vartheta_j}{\assume{G[\var{pos}][\var{next}]}}{\xi_j}
\]
For each outcome, we know that $\var{pos} = i$ and $\var{next}= j$. If $G[i][j] = \tru = 0$, then $(\short_n^t(G,s,i)\land G[i][j]) +n= \short_n^t(G, s, i)+n$, so the postcondition is unchanged. If $G[i][j] = \fls = \infty$, then $(\short_n^t(G,s,i)\land G[i][j]) +n = \infty$ and the outcome is eliminated as expected.
We now justify the consequence after $\assume 1$. Consider the term:
\[
\bigoplus_{i \in I} \wg{\var{pos} = j}{(\short_n^t(G, \var s, i)\land G[i][j])+n+1}
\]
This corresponds to just taking the outcome of minimum weight, which will be $n+1$ if $\short_n(G, \var s, i)\land G[i][j]$ is true for some $i\in I$ and $\infty$ otherwise. By definition, this corresponds exactly to $\short_{n+1}^t(G, s, j) + n+1$.



\fi

\end{document}